\documentclass[11pt]{article}
\usepackage[utf8]{inputenc}
\usepackage[margin=1in]{geometry}

\usepackage{mathtools}
\usepackage{xspace,amsmath,amssymb,shortcuts,mathtools,tikz,subcaption,amsthm,hyperref}
\usepackage{natbib,multirow}
\usepackage[capitalize]{cleveref}

\usepackage{multicol}
\usepackage{booktabs}
\usepackage{graphicx}
\usepackage{float}
\usepackage{subcaption}

\usepackage{etoolbox}

\RequirePackage[normalem]{ulem}

\PassOptionsToPackage{authoryear,sort}{natbib}
\Crefname{assumption}{Assumption}{Assumptions}
\Crefname{exmp}{Example}{Examples}

\newcommand{\Ecal}{\mathcal{E}}

\newcommand{\Hcal}{\mathcal{H}}
\newcommand{\Ical}{\mathcal{I}}

\newcommand{\Lcal}{\mathcal{L}}

\newcommand{\Ncal}{\mathcal{N}}

\newcommand{\Qcal}{\mathcal{Q}}
\newcommand{\Rcal}{\mathcal{R}}

\newcommand{\Wcal}{\mathcal{W}}

\newcommand{\EE}{\mathbb{E}}

\newcommand{\PP}{\mathbb{P}}

\newcommand{\hP}{\mathbb{P}}
\newcommand{\hG}{\mathbb{G}}

\newcommand{\operator}{P}
\newcommand{\cl}{\op{cl}}

\newtheorem{theorem}{Theorem}

\newtheorem{lemma}{Lemma}
\newtheorem{proposition}{Proposition}
\newtheorem{assumption}{Assumption}
\newtheorem{definition}{Definition}
\theoremstyle{definition}
\newtheorem{example}{Example}
\AtEndEnvironment{example}{\qed}
\newtheorem{remark}{Remark}
\AtEndEnvironment{remark}{\qed}

\newtheorem{subassumption}{Assumption}[assumption]

\newenvironment{assumption*}
  {\ifnum\value{subassumption}=0 \stepcounter{assumption}\fi\subassumption}
 {\endsubassumption}

\newenvironment{assumption+}[2]
 {\renewcommand{\thesubassumption}{\getrefnumber{#1}#2}\subassumption}
 {\endsubassumption}

\makeatletter
\def\blfootnote{\gdef\@thefnmark{}\@footnotetext}
\makeatother

 \newenvironment{continuance}[2][Example]
  {\medskip\par\noindent\textbf{Example #2, Cont'd} (#1).}
  {\par}

\title{Inference on Strongly Identified Functionals\\of Weakly Identified Functions}

\author{
Andrew Bennett \\ Cornell University \\ awb222@cornell.edu
\and Nathan Kallus \\ Cornell University \\ kallus@cornell.edu
\and Xiaojie Mao \\ Tsinghua University \\ maoxj@sem.tsinghua.edu.cn \\
\and Whitney Newey \\ Massachusetts Institute of Technology \\ wnewey@mit.edu
\and Vasilis Syrgkanis \\ Stanford University \\ vsyrgk@stanford.edu
\and Masatoshi Uehara \\ Cornell University \\ mu223@cornell.edu
}

\begin{document}
\maketitle

\begin{abstract}
In a variety of applications, including nonparametric instrumental variable (NPIV) analysis, proximal causal inference under unmeasured confounding, and missing-not-at-random data with shadow variables, we are interested in inference on a continuous linear functional (\eg, average causal effects) of nuisance function (\eg, NPIV regression) defined by conditional moment restrictions.
These nuisance functions are generally weakly identified, in that the conditional moment restrictions can be severely ill-posed as well as admit multiple solutions. This is sometimes resolved by imposing strong conditions that imply the function can be estimated at rates that make inference on the functional possible.
In this paper, we study a novel condition for the functional to be strongly identified even when the nuisance function is \textit{not}; that is, the functional is amenable to asymptotically-normal estimation at $\sqrt{n}$-rates. 
The condition implies the existence of debiasing nuisance functions, and we propose penalized minimax estimators for both the primary and debiasing nuisance functions. 
The proposed nuisance estimators can accommodate flexible function classes, and importantly they can converge to fixed limits determined by the penalization regardless of the identifiability of the nuisances. 
We use the penalized nuisance estimators to form a debiased estimator for the functional of interest and prove its asymptotic normality under generic high-level conditions, which provide for asymptotically valid confidence intervals.  
We also illustrate our method in a novel partially linear proximal causal inference problem and a partially linear instrumental variable regression problem.\blfootnote{Accepted for presentation at the Conference on Learning Theory (COLT) 2023}
\end{abstract}

\section{Introduction}\label{sec: setup}

Many causal or structural parameters of interest can be expressed as linear functionals of unknown functions that satisfy conditional moment restrictions. 
For example, in a nonparametric instrumental variable (NPIV) model, parameters such as average policy effects, weighted average derivatives, and average partial effects are all linear functionals of the NPIV regression, which is characterized by a conditional moment equation {given by the exclusion restriction} \citep{newey2003instrumental}. 
Similarly, in proximal causal inference under unmeasured confounding \citep{tchetgen2020introduction}, the average treatment effect and various policy effects can be expressed as linear functionals of a bridge function defined as a solution to some conditional moment restrictions \citep{cui2020semiparametric,kallus2021causal}.
Similarly, in missing-not-at-random data problems with shadow variables, some parameters for the missing subpopulation can be also written as linear functionals of an unknown density ratio function that solves a certain conditional moment restriction \citep{LiMiao2022,miao2015identification}. 

In this paper, we tackle the commonplace problem that the conditional moment restrictions only \emph{weakly identify} the nuisance functions that define the target parameter. 
That is, the conditional moment restrictions can be severely ill-posed as well as admit multiple solutions, so that the nuisance functions are not uniquely identified and are not stable as underlying distributions vary.
This problem can easily occur in many applications. 
For example, the identification of NPIV regressions requires the so-called ``completeness condition,'' and stronger conditions yet are needed to control ill-posedness in nonparametric settings.
These conditions can be easily violated if the instrumental variables are not very strong, as is common in practice \cite[\eg, ][]{andrews2005inference,Andrews2019weak}. 
Even with additional restrictions on the function space, 
\cite{santos2012inference} shows that {NP}IV regressions are  unidentifiable in a {variety of} of models.   
Moreover,
\cite{canay2013testability} argues that completeness conditions are generally not testable, so the  unidentifiability problem may not be diagnosable.
Similar phenomena are also common in proximal causal inference and missing data problems with shadow variables (see \cref{ex: proximal,ex: shadow} in \Cref{sec: examples}).

Fortunately, even when the nuisance functions are not identifiable, 
the linear functionals of interest can still be identifiable.
In particular, these functionals often capture some identifiable aspects of the unidentifiable nuisance functions. 
For example, in proximal causal inference, merely the existence -- but not uniqueness -- of bridge functions is sufficient for identification of the average treatment effect and some policy effects \citep{miao2018a,cui2020semiparametric,kallus2021causal}.
But even if the functional is identifiable, the ill-posedness of the conditional moment restrictions that define it can raise significant challenges for statistical inference. 
In this setting, without conditions that limit ill-posedness, common nuisance estimators may be unstable and resulting estimators for the functional may not be $\sqrt{n}$-consistent and asymptotically normal. 

\paragraph*{Setup and Main Results.} To tackle this challenge in generality, we study continuous linear functionals of nuisance functions characterized by linear conditional moment restrictions. 
Our parameter of interest is a functional of some unknown nuisance function $h^\star \in \Hcal \subseteq \Lcal_2(S)$:
\begin{align}\label{eq: theta-intro}
\theta^\star = \Eb{m(W; h^\star)},
\end{align}
where $\Hcal$ is a closed linear space (\emph{i.e.} Hilbert space) and $m$ is a given function such that $h \mapsto \Eb{m(W; h)}$ is a continuous linear functional over $\Hcal$.
For example,  $\Hcal$ can be the whole $\Lcal_2(S)$ space, or some other structured class like a  partially linear function class or an additive function class. 
We posit that 
$h^\star$  solves the following {linear} conditional moment restriction in $h$:
\begin{align}\label{eq: cond-moment-h}
\Eb{g_1(W)h(S) \mid T} = \Eb{g_2(W) \mid T} \,,
\end{align}
for some given functions $g_1 \in \Lcal_2(W)$ and $g_2 \in \Lcal_2(W)$, where $S=S(W)$ and $T=T(W)$ are two $W$-measurable variables (usually, potentially-overlapping subsets of the components of the vector $W$).
Letting $P: \Lcal_2(S) \to \Lcal_2(T)$ denote the linear operator given by $(Ph)(T) = \EE[g_1(W) h(S) \mid T]$, and $r_0 = \EE[g_w(W) \mid T]$, this corresponds to $h_0$ solving the linear inverse problem
\begin{equation*}
    P h_0 = r_0 \,.
\end{equation*}
This general formulation includes a wide variety of problems,  including NPIV, proximal causal inference, and
missing data with shadow variable  (see \Cref{ex: shadow,ex: proximal,ex: IV-reg} in \Cref{sec: examples}). To distinguish $h^\star$ from other nuisance functions introduced later, we refer to $h^\star$ as the \emph{primary nuisance}.

Note that $h(S)$ in \Cref{eq: cond-moment-h} is a function of variables $S$ that may not appear in the conditioning variable $T$.
Thus, solving \Cref{eq: cond-moment-h} is generally an ill-posed inverse problem: it may not have a unique solution and its solutions may  depend on the data distributions discontinuously \citep{Carrasco2007}. 
We here allow the conditional moment restrictions to be severely ill-posed and have nonunique solutions. 
However, we introduce a novel general condition that ensures \emph{strong identification} of the linear functional, regardless of the identification of the nuisance functions given by the conditional moment restrictions. 
Under this condition we develop estimation and inference methods, when given access to $n$ independent and identically distributed data samples $W_1, \dots, W_n\sim W$, with guarantees that are robust to the weak identification of $h^\star$. This provides a general and unified solution to a wide range of applications where concerns regarding weakly identified nuisances arise naturally.
Note also that here (and throughout the paper) we use the term ``weak identification'' of $h^\dagger$ to mean that it is set identified rather than point-identified; this is in contrast to the meaning of weak identification within some of the parametric IV literature \citep{staiger1997instrumental,stock2005testing}.

Specifically, we define the linear functional to be \emph{strongly identified} if a certain minimization problem admits a solution.
\begin{definition}[Functional Strong Identification]\label{def:strong-ident} We say $\theta^\star$ is strongly identified if
\begin{align}\label{eq: delta-eq-2}
    \Xi_0\neq\emptyset,\quad\text{where}\quad
    \Xi_0 \coloneqq \argmin_{\xi\in\Hcal}\frac{1}{2}\Eb{\Eb{g_1(W) \xi(S)\mid T}^2} - \Eb{m(W; \xi)} \,.
\end{align}
\end{definition} 
Furthermore, we make the key assumption that $\theta^\star$ is strongly identified.
\begin{assumption}\label{assump: new-nuisance}
    $\theta^\star$ is strongly identified 
\end{assumption}
Then, not only do we have $\theta^\star=\E[m(W;h_0)]$ for \emph{any} solution $h_0$ to the conditional moment equations in \cref{eq: cond-moment-h} (\ie, $\theta^\star$ is identified even though $h^\star$ is not), but we also have that if we let $\xi_0\in\Xi_0$ be any solution to \cref{eq: delta-eq-2} and if we set $q^\dagger(T)\coloneqq \E[g_1(W)\xi_0(S)\mid T]$, then $\theta^\star$ further admits the following de-biased or Neyman orthogonal representation (\cref{thm:dr-ident} in \Cref{sec: identification}): 
\begin{align}\label{eq:debiased representation}
&\theta^\star = \Eb{\psi(W; h_0, q^\dagger)} &
\psi(W; h, q) \coloneqq~& {m(W; h) + q(T)(g_2(W) -  g_1(W)h(S))},
\end{align} 
where we call $q^\dagger$ the \emph{debiasing nuisance}.
Crucially, for this representation, the estimation error rates can be expressed solely based on ``weak'' metrics that avoid dependence on any ill-posedness measure (\Cref{thm:dr-ident}). For illustration, with $g_1(W) = 1$, for any $\hat{h} \in \Hcal$ and any $\hat{\xi} \in \Hcal$, letting $\tilde{q}(T)=\E[\hat{\xi}(S)\mid T]$, we have
\begin{align}\label{eq: DR-identification-intro}
 \left|\Eb{\psi(W; \hat{h}, \tilde{q})} - \theta^\star\right|^2 &\leq  {\Eb{\Eb{\hat{h}(S)-h_0(S)\mid T}^2}} \cdot
 {\Eb{(\tilde q(T)-q^\dagger(T))^2}}
 \\\label{eq: DR-identification-intro 2}
 &={\Eb{\Eb{\hat{h}(S)-h_0(S)\mid T}^2}} \cdot{\Eb{\Eb{\hat{\xi}(S)-\xi_0(S)\mid T}^2}}.
 \end{align}
We see from equation (5) that $\psi(W; \hat{h}, \tilde{q})$ is a doubly robust estimating function, having expectation equal to the true parameter if either $\Eb{\hat{\xi}(S)-\xi_0(S)\mid T}=0$
or $\Eb{\hat{h}(S) - h_0(S)\mid T}=0$. The fact that both estimation metrics project the functions onto $T$ allows us to argue convergence rates for both of these functions without invoking any measure of ill-posedness of the corresponding inverse problems that define them.
In particular, ${\EE[\EE[\hat{h}(S)-h_0(S)\mid T]^2}]$ exactly corresponds to the average squared slack in $\hat{h}$ satisfying the moment restrictions in \cref{eq: cond-moment-h}, and this does \emph{not} involve how such slack translates to the distance from an exact solution.
This is one of the key ingredients that enables our main results and it hinges on our key assumption that a solution to the minimization problem in \cref{eq: delta-eq-2} exists. Note that \cref{eq: DR-identification-intro,eq: DR-identification-intro 2} are meant only for illustration, and, in practice, given $\hat\xi$ we will still need to estimate its projection $\tilde{q}$.

To demystify \cref{def:strong-ident}, we further examine its equivalent formulations 
based on the structure of the Riesz representer of the functional, \ie, the function $\alpha$ that satisfies 
\begin{align}\label{eq: Riesz}
\Eb{m(W; h)} = \Eb{\alpha(S)h(S)}\quad\forall h\in\Hcal,
\end{align}
whose existence is guaranteed by $h\mapsto\Eb{m(W; h)}$ being linear and continuous \citep[Section 5.3]{luenberger1997optimization}.
When $h_0$ corresponds to the solution of an exogenous regression problem (\ie, $S=T$), weak identification of $h^\star$ is not a concern and \cref{assump: new-nuisance} is automatic with $\Xi_0=\braces{\alpha}$.
In our setting, the solution $\xi_0$ is no longer the Riesz representer, but is closely related to it. For instance, when the function space $\Hcal=\Lcal_2(S)$ and $g_1(W) = 1$, then $\Eb{\Eb{\xi_0(S)\mid T}\mid S}=\alpha(S)$, and our \cref{assump: new-nuisance} is equivalent to the existence of $\xi_0\in \Hcal$ satisfying the latter equality. Note that we can still write a representation as in \cref{eq: DR-identification-intro} if we only assume the existence of some $q_0$ solving $\Eb{q_0(T)\mid S}=\alpha(S)$ as, for example, assumed in \citet{severini2012efficiency}, even if it is not of the form  $q_0 = \Eb{\xi_0(S) \mid T}$. However, crucially, the resulting error in \cref{eq: DR-identification-intro} then would not reduce to \emph{estimation errors of projections} as in \cref{eq: DR-identification-intro 2}, therefore requiring control on ill-posedness.
While the existence of (arbitrarily approximate) solutions to $\Eb{q_0(T)\mid S}=\alpha(S)$ is equivalent to mere identification of $\theta^\star$ (\cref{lemma: identifiability}),
\cref{assump: new-nuisance} provides the added regularity for \emph{strong} identification of $\theta^\star$. We in fact show that \cref{assump: new-nuisance} holds if and only if the Riesz representer $\alpha$ lies in a relatively smooth sub-space of the function space $\Hcal$. This can be expressed as a ``source condition'' on the Riesz representer. Prior work on nonparametric inference for ill-posed inverse problems typically places such ``source conditions'' on the primary nuisance function $h^\star$. Thus, our assumption can be viewed as a dual approach, imposing source conditions only on objects related to the \emph{functional}.

Armed with our \cref{assump: new-nuisance}, we turn to the task of estimating the nuisance functions that can be plugged into the debiased representation in \cref{eq:debiased representation} and control the projected-error bound in \cref{eq: DR-identification-intro 2}.
We propose minimax (adversarial) estimators $\hat h,\hat q$ for the primary and debiasing nuisances. These estimators are generic and admit highly flexible function classes like reproducing kernel Hilbert space (RKHS) and neural networks. 
Moreover, our estimators do not require deriving the closed form of the functional Riesz representer $\alpha$, just like the automatic debiased machine learning methods (see a review in \Cref{sec: literature}). 
We leverage $L_2$-norm penalization for primary-nuisance estimation, so as to
ensure $\hat h$ converges to some fixed limit $h^\dagger$ in $L_2$-norm even when $h^\star$ is in general not identified and \cref{eq: cond-moment-h} admits many solutions (\Cref{thm:h-estimator-bound}). Conversely, our debiasing-nuisance estimator converges to a fixed limit even without penalization, because $\EE[g_1(W) \xi_0(S) \mid T] = q^\dagger(T)$ is shown to be unique for \emph{all} $\xi_0\in\Xi_0$ satisfying \Cref{eq: delta-eq-2} (\Cref{lemma: minimum-norm,thm:q-estimator-bound}).

We derive a novel analysis yielding finite-sample convergence rates for our nuisance estimators $\hat{h}$ and $\hat{q}$, bounding the weak metric for the primary nuisance, 
$\epsilon_n^2 = \EE[{\EE[g_1(W)(\hat{h}(S)-h_0(S))\mid T}]^2]$, 
and the strong metric for the debiasing nuisance
 $\kappa_n^2=\EE[({\hat{q}(S)-q^\dagger(S)})^2]$.
Our bounds rely solely on the critical radius $r_n$ (a well-studied and typically tight notion of statistical complexity of a function space) of approximating function spaces $\Hcal_n, \Qcal_n$ for the nuisance estimators $\hat h, \hat q$ and their approximation error $\delta_n$. Both finite sample-rates do not involve any measures of ill-posedness. The intuition behind the strong-metric result for $q^\dagger$ is that it essentially corresponds to a weak-metric convergence for $\xi_0$, for which we can provide ill-posedness-free rates. For the primary nuisance function, we derive fast estimation rates of the order of $\epsilon_n\sim r_n + \delta_n$, due to the fact that it corresponds to a 
conditional moment problem. For the de-biasing nuisance function we derive slow rates of the order of $\kappa_n\sim \sqrt{r_n} + \delta_n$,
since it roughly corresponds to a minimum distance problem, that can be thought as corresponding to a mis-specified  conditional moment problem; the mis-specification leads to the slower rate for the minimum distance solution.

Putting everything together, we develop a complete estimation and inference pipeline for parameters defined by \cref{eq: theta-intro,eq: cond-moment-h} that avoids dependence on ill-posedness measures or point-identification assumptions on the nuisance functions, where our main assumptions are \cref{assump: new-nuisance} and high-level conditions on the complexity and approximability of general function classes.
We construct de-biased estimators for the functional of interest by plugging our penalized nuisance estimates, estimated in a cross-fitting manner, into the debiased representation in \cref{eq:debiased representation}.
We prove that the resulting functional estimator is asymptotically linear and has an asymptotic normal distribution (\Cref{thm:dml-asymp}) under \cref{assump: new-nuisance}, some regularity conditions, and assuming
the critical radii and approximation errors decay fast enough in that $r_n^{3/2}$, $\sqrt{r_n}\delta_n$, $\delta_n^2$ are all $o(n^{-1/2})$.
In particular, these constraints apply to many non-parametric function spaces of interest.

We illustrate the application of our proposed method to partially linear proximal causal inference (\Cref{sec: partial-linear-proximal}) and partially linear IV estimation (\Cref{sec: partial-linear-iv}). This semiparametric proximal causal inference settings are new as, to our knowledge, existing literature on proximal causal inference focus on either parametric or fully nonparametric estimation.
In particular, we characterize when the primary nuisance function in proximal causal inference is a partially linear function (\Cref{prop: partial-linear-proximal}) and discuss how to apply our proposed method to estimate the linear coefficients. 
Our method can be similarly applied to partially linear IV estimation problems, extending the existing literature that either assumes the unique identification of IV regression \citep{chernozhukov2018double,florens2012instrumental,ai2003efficient} or allows under-identified IV regression but focus on linear-sieve estimation \citep{chen2021robust}.

\paragraph*{Roadmap.} The rest of this paper is organized as follows. We first review examples of our setup in \Cref{sec: examples} as well as related literature in \Cref{sec: literature}. 
We then further characterize the strong identification of linear functionals of weakly identified nuisances in \Cref{sec: identification}, interpreting \cref{assump: new-nuisance} in \Cref{sec: interpretation}, instantiating the condition in concrete examples in \Cref{sec: examples2},  
and discussing the statistical challenges raised by weakly identified nuisances in \Cref{sec: challenge}. Then we propose our minimax nuisance estimators in \Cref{sec: est-nuisance}, considering the primary nuisance in \Cref{sec: primary-nuisance-est} and the debiasing nuisance in \Cref{sec: debias-nuisance-est}. 
We further construct debiased functional estimators, establish their asymptotic normality, and discuss  variance estimation and confidence intervals in \Cref{sec: est-functional}. 
In \Cref{sec: partial-linear}, we illustrate the application of our method in partially linear IV estimation and partially linear proximal causal inference. 
Finally, we conclude this paper and discuss future directions in \Cref{sec: conclusion}. 

\paragraph*{Notation.} 
For a generic  random vector $W \in \Wcal$, we use $\Lcal_2(W)$ to denote the space of square integrable functions of $W$ with respect to the  probability distribution of $W$. For any $f(W), g(W) \in \Lcal_2(W)$, we denote the $L_2$-norm by $\|f\|_2 = \sqrt{\Eb{f^2(W)}}$ and inner product by $\langle f, g \rangle = \Eb{f(W)g(W)}$. 
We denote the empirical $L_2$-norm with respect to data $W_1, \dots, W_n$ by $\|f\|_n = \sqrt{\sum_{i=1}^n{f^2(W_i)}/n}$.
We let $\hP\prns{f(W)} = \int f(w)\diff \hP(w)$
be the expectation with respect to $W$ alone. We differentiate this from $\Eb{f(W; W_1, \dots,W_n)}$, which we use to denote full expectation with respect to both $W$ and data $W_1, \dots, W_n$. Thus if $\hat h$ depends on the data $W_1, \dots, W_n$, then $\hP(f(W; \hat h))$ remains a function of $\hat h$ (and thus the data) but $\mathbb{E}[f(W; \hat h)]$ is a nonrandom scalar. We use both $\E_n$ and $\hP_n$ to denote the empirical expectation with respect to $W$ given data $W_1, \dots, W_n$: $\E_n(f(W)) = \hP_n(f(W)) = \frac{1}{n}\sum_{i=1}^n f(W_i)$. We further define the empirical process $\hG_n$ by $\hG_n(f) = \sqrt{n}\prns{\hP_n - \hP}(f)$.
For any linear operator $L: \mathcal{A} \mapsto \mathcal{B}$ where  $\mathcal{A}$ and $\mathcal{B}$ are Hilbert spaces, we denote its range space by $\Rcal(L) = \braces{La: a \in \mathcal{A}} \subseteq \mathcal{B}$ and its null space by $\Ncal(L) = \braces{a: La = 0} \subseteq \mathcal{A}$.
Unless otherwise stated, the default norm for the $\Lcal_2(W)$ space is $\|\cdot\|_2$. Thus, the convergence of functions in $\Lcal_2(W)$, the compactness of subsets of $\Lcal_2(W)$, and the continuity of functionals defined on $\Lcal_2(W)$ are all understood in terms of the norm $\|\cdot\|_2$, and the default inner product is $\langle \cdot, \cdot \rangle$. 
Furthermore, for vector-valued square-integrable functions, we use the notation $\|f\|_{2,2} = \sqrt{\EE[f(W)^\top f(W)]}$ for the $L_2$ with the euclidean norm as the base vector norm, and similarly $\|f\|_{n,2} = \sqrt{\sum_{i=1}^n f(W_i)^\top f(W_i)}$, and we generalize the above conventions using the inner product $\langle f, g \rangle = \EE[f(W)^\top g(W)]$, and the default norm 
$\|\cdot\|_{2,2}$ instead of $\|\cdot\|_2$.

For any set $D$, we denote its closure as $\cl\prns{D}$ and its interior as $\op{int}(D)$. We say that $D$ is star-shaped if for any $d \in D$ and $\alpha \in [0, 1]$, we have $\alpha d\in D$. The star hull of $D$ is defined as $\starcls(D) = \braces{\alpha d: d \in D, \alpha \in [0, 1]}$. 
For a function class $\mathcal{G}\subseteq \Lcal_2(W)$, we say it is $b$-uniformly bounded if $\abs{g(W)} \le b$ almost surely for any $g \in \mathcal{G}$. The localized Rademacher complexity of a function class $\mathcal{G}$ is defined as $\mathcal{R}_n(\delta; \mathcal{G}) = \EE[{\sup_{g \in \mathcal{G}, \|g\|_2 \le \delta}\abs{\frac{1}{n}\sum_{i=1}^n \epsilon_i g(W_i)}}]$, where $W_1,\dots,W_n,\epsilon_1,\dots,\epsilon_n$ are independent with $W_i\sim W$ and $\epsilon_i$ taking values in $\braces{-1, 1}$ equiprobably. The critical radius $\eta_n$ of the function class $\mathcal{G}$ is defined as any solution to the inequality $\mathcal{R}_n(\delta; \mathcal{G}) \le \delta^2$. 
{We use $\abs{\mathcal{G}}$ to denote the cardinality of ${\mathcal{G}}$ up to equality almost everywhere.}
For real-valued sequences $a_n,b_n$, we use the standard ``big-O'' notations $a_n = o(b_n)$ to denote that $a_n / b_n \to 0$ as $n \to \infty$, and $a_n = \omega(b_n)$ to denote that $a_n / b_n \to \infty$ as $n \to \infty$.

\section{Examples}\label{sec: examples}

Before proceeding we review important examples of parameters defined by \cref{eq: theta-intro,eq: cond-moment-h}.

\begin{example}[Functionals of NPIV Regression]\label{ex: IV-reg}
Consider a causal inference problem with an observed outcome $Y \in \R{}$, potentially endogenous  variables $X \in \R{d_X}$, and instrumental variables $Z \in \R{d_Z}$. We are interested in the NPIV regression model \citep[\eg, ][]{newey2003instrumental,darolles2011nonparametric,hall2005nonparametric}:
\begin{align*}
Y = h^\star(X) + \epsilon, ~~ \text{where }\Eb{\epsilon \mid Z} = 0, ~ h^\star \in \Hcal = \Lcal_2(X). 
\end{align*}
The NPIV regression $h^\star$ solves the following conditional moment restriction: 
\begin{align}\label{eq: IV-h}
\Eb{h(X) \mid Z} = \Eb{Y \mid Z},
\end{align}
which is an example of \Cref{eq: cond-moment-h} with $g_1(W) = 1$, $g_2(W) = Y$, $S = X$, and $T = Z$. 

We are interested in linear functionals of the IV regression $h^\star$ when $h^\star$ is not necessarily identifiable. 
One example is the coefficient of the best linear approximation to the IV regression function, $\argmin_{\beta\in\R{d_X}} \Eb{\prns{h^\star(X) - \beta^\top X}^2}=\prns{\Eb{XX^\top}}^{-1}\theta^\star$, where
\begin{align}
\label{eq: best-linear-iv}
{\theta}^\star = \Eb{\alpha(X)h^\star(X)}, ~~ \alpha(X) = X.
 \end{align}
Alternatively, we can  consider other linear functionals of NPIV regressions, such as the weighted average derivatives described in \cite{ai2007estimation,chen2015sieve}. 

It is known that the IV regression $h^\star$ is identifiable if and only if a completeness condition on the distribution of $X \mid Z$ is satisfied (Proposition 2.1 in \citealp{newey2003instrumental}). 
However, as \cite{severini2006some} showed, the completeness condition can be easily violated, especially when the instrumental variables are not very strong. 
Moreover, the completeness condition is  impossible to test in general nonparametric models \citep{canay2013testability}, so the failure of identifying $h^\star$ may not be  detectable. 
Fortunately, even when the NPIV regression $h^\star$ is unidentifiable, the functional of interest can be still identifiable as we will discuss in \Cref{sec: identification}. 

Here we consider a general nonparametric IV model with $\Hcal = \Lcal_2(X)$. In \Cref{sec: partial-linear-iv}, we will further study a partially linear IV model, where $X=(X_a,X_b)$, the 
function class $\Hcal$ is the direct sum of linear functions in $X_a$ and of $\Lcal_2(X_b)$, and $\theta^\star$ is the linear coefficient in $X_a$. 
\end{example}

\begin{example}[Proximal Causal Inference]\label{ex: proximal}
Consider a causal inference problem with potential outcomes $Y(a)$ that would be realized if the treatment assignment were equal to $a\in\{0,1\}$.
 We are interested in the average treatment effect, $\theta^\star = \Eb{Y(1)-Y(0)}$.
 The actual treatment assignment is denoted as $A$, the corresponding observed outcome is $Y = Y(A)$, and some additional covariates $X$ are also observed, but these do not account for all confounders and there exist {unmeasured confounders} $U$.
We consider the proximal causal inference framework \citep[\eg,][]{tchetgen2020introduction,miao2018identifying,deaner2018proxy} that requires two different sets of proxy  variables $Z, V$ that strongly correlate with the unobserved confounders. 
The so-called negative control treatment $Z$ 
cannot directly affect the outcome $Y$, and the so-called negative control outcome $V$ cannot be affected by either the treatment $A$ or the negative control treatment $Z$ (see Assumptions 4 to 7 in \citealp{cui2020semiparametric} for formal statements). 

If $h^\star \in \Hcal = \Lcal_2(V, X, A)$ is a so-called outcome bridge function satisfying
 \begin{align}\label{eq: proximal-h-unobs}
 \Eb{{Y - h^\star(V, X, A)} \mid U, X, A} = 0,
 \end{align}
then our target parameter can be written as a linear functional of it:
\begin{align}\label{eq: proximal-parameter}
&\theta^\star = \Eb{h^\star(V, X,1)-h^\star(V, X,0)} = \Eb{\alpha(V, X, A)h^\star(V, X, A)}, \\
&\text{where }\alpha(V, X, A) = \frac{A-\Prb{A=1\mid  V, X}}{\Prb{A=1\mid  V, X}(1-\Prb{A=1\mid  V, X})}.\nonumber
\end{align}
\Cref{eq: proximal-h-unobs} involves unobserved variables, but under negative control assumptions, \Cref{eq: proximal-h-unobs}  implies that $h^\star$ also solves the following (observable) conditional moment restriction
\begin{align}\label{eq: proximal-h}
 \Eb{Y - h(V, X, A) \mid Z, X, A} = 0.
 \end{align}
When the negative control treatment $Z$ consists of sufficiently strong proxies for the unmeasured confounders $U$ (see Assumption 8 in \cite{cui2020semiparametric}), any solution to \Cref{eq: proximal-h} also solves \Cref{eq: proximal-h-unobs}. 
As a result, our target parameter in \Cref{eq: proximal-parameter} can be also viewed as a linear functional of a solution to \Cref{eq: proximal-h}, which is  a special example of our general framework with $g_1(W) = 1$, $g_2(W) = Y$, $S = \prns{V^\top, X^\top, A}^\top$, and $T = \prns{Z^\top, X^\top, A}^\top$. 
Similarly, we can consider various average policy effects in the proximal causal inference framework, even with continuous treatments, as they can also be written as linear functionals of bridge functions \citep{kallus2021causal,QiZhengling2021PLfI}.

\cite{kallus2021causal} points out that the solution to \Cref{eq: proximal-h} is very likely to be nonunique and gives various concrete examples. This particularly occurs in  data-rich settings where there are more proxy variables than the unobserved confounders.
Since the unobserved confounders are unknown in practice, it is generally impossible to know a priori whether bridge functions are unique or not.  
Fortunately, it is well known that even with nonunique bridge functions, any of them can still lead to the same average treatment effect or average policy effect under suitable conditions \citep{cui2020semiparametric,miao2018a,kallus2021causal} (however, for inference, this literature still requires bridge functions to be unique, parametric, and/or satisfy restricted ill-posedness, all of which we will avoid). In \Cref{sec: identification}, we will derive identification conditions for the target linear functionals.

In \Cref{sec: partial-linear-proximal}, we will further study a partially linear proximal causal inference model. 
We will show that when the regression function $\Eb{Y(a) \mid U, X}$ is partially linear in $a$, there always exists a bridge function $h^\star(V, X, A)$ partially linear in $A$. 
In this case, the partially linear coefficient is a causal parameter and it is also a linear functional of the bridge function. 
\end{example}

 \begin{example}[Missing-Not-at-Random Data with Shadow Variables]\label{ex: shadow}
Consider a partially missing outcome $Y$ and an indicator $A \in \braces{0, 1}$ denoting whether $Y$ is observed, so that we only observe $V=AY$.
We are interested in the average missing outcome, $\theta^\star = \Eb{(1-A)Y}$ (which also gives the outcome mean via $\Eb{Y} = \theta^\star + \Eb{V}$).
However, suppose that the outcome is \emph{missing not at random},
namely, although we observed some covariates $X$, we generally allow that $Y \not \perp A \mid X$. 
Nonetheless, if $h^\star$ were the $Y$-conditional missingness propensity ratio $h^\star(X, Y) = \Prb{A=0\mid X,Y}/\Prb{A=1\mid X, Y}$, then $\theta^\star$ is a linear functional of it, using only the observables $W=(X^\top, V, A)^\top$:
\begin{align}\label{eq: shadow-parameter}
 \theta^\star =  \Eb{\alpha(X, V)h^\star(X, V)},~~ \alpha(X, V) = V. 
 \end{align}
 Here we again consider a general nonparametric model with $h^\star \in \Hcal = \Lcal_2(X, V)$. 

Since the definition of $h^\star$ involves conditioning on unobservables, we cannot generally learn it.
We therefore additionally consider  so-called shadow variables $Z$ satisfying $Z \perp A \mid X, Y$ and $Z \not \perp Y \mid X$ \citep[\eg, ][]{wang2014instrumental,d2010new,miao2015identification,miao2016varieties,LiMiao2022}.
These conditions are particularly relevant when the missingness is directly driven by the outcome $Y$ and the shadow variables $Z$ are strong proxies for the outcome $Y$.
Under these conditions, $h^\star$ will necessarily satisfy the following conditional moment restriction \citep{LiMiao2022}:
\begin{align}\label{eq: shadow-h}
 \Eb{Ah(X, V)\mid X, Z} = \Eb{1-A\mid X, Z}.
 \end{align}
This is an example of \Cref{eq: cond-moment-h} with $g_1(W) = A$, $g_2(W) = 1-A$, $S = (X^\top, V)^\top$, and $T = (X^\top, Z)^\top$. 

Again, the conditional moment restriction in \Cref{eq: shadow-h} admits multiple solutions (\ie, does not identify $h^\star$) unless a strong completeness condition on the distribution of $V \mid X, Z$ holds.
In \Cref{sec: identification}, we will discuss conditions that ensure the identification of $\theta^\star$ even when $h^\star$ is not identified. 
 \end{example}

\section{Related Literature}\label{sec: literature}

Our paper is related to the literature on the estimation of and inference on point-identified  (finitely and/or infinitely dimensional) parameters  defined by conditional moment restrictions \citep[\eg, ][]{newey1990efficient,chamberlain1987asymptotic,newey2003instrumental,ai2003efficient,blundell2007semi,chen2012estimation,chen2009efficient,chen2011rate,hall2005nonparametric,darolles2011nonparametric}. 
Some other literature study partial identification sets and inference thereon when unconditional or conditional moment restrictions underidentify the parameters \citep[\eg, ][]{andrews2013inference,andrews2014nonparametric,andrews2012inference,chernozhukov2007estimation,belloni2019subvector,canay2017practical,bontemps2017set,hong2017inference,santos2012inference}.

Our paper studies the estimation and inference of \textit{identifiable} functionals of unknown nuisance functions satisfying certain conditional moment restrictions.
We do not restrict the nuisance functions to be parametric and focus on inference after using flexible nonparametric estimation of nuisances. 
Existing literature usually study this problem when the nuisance function is point-identified \citep[\eg, ][]{ai2007estimation,ai2012semiparametric,chen2015sieve,chen2018optimal,brown1998efficient,newey1993efficiency,chen2021efficient}. Recently, a few works further consider functionals of partially identified nuisance functions.
In particular, \cite{severini2006some,escanciano2013identification,freyberger2015identification} study the point identification and partial identification of continuous linear functionals of partially identified  NPIV regressions.
\cite{severini2012efficiency} discuss efficiency considerations for point-identified linear functionals of unidentified NPIV regressions.
Some works also study the  estimation of identifiable linear functionals of unidentifiable nuisances and derive their asymptotic distributions, either in the IV setting \citep{santos2011instrumental,babii2017completeness,chen2021robust,escanciano2021optimal} or in the shadow variable setting  
\citep{LiMiao2022}. These works are reviewed in more detail in  \Cref{sec: identification}.
Our paper builds on these literature and develops them in several aspects. First, our paper proposes a unified solution to a wider variety of problems, including IV, shadow variables, and proximal causal inference \citep{tchetgen2020introduction}.
Second, these existing literature are restricted to series or kernel estimation for conditional moment models, while our paper adopts a minimax estimation framework to accommodate generic function classes and thus more flexible machine-learning models like RKHS and neural networks. For example, \citet{chen2003estimation,chen2007large} achieve a similar result as us in that they can establish $\sqrt{n}$-asymptotically normal convergence of $\hat\theta_n$ while only requiring fast convergence of $\hat h_n$ under the weak (projected) norm, which is possible since they implicitly assume \Cref{assump: new-nuisance}. However, they only provide results for linear sieves, whereas we obtain these results for general function classes.
Third, much of the existing literature directly restricts the ill-posedness of the conditional moment restrictions that define the nuisance functions (\eg, the NPIV regressions) for the target functional to be well estimable \citep[\eg, ][]{babii2017completeness,escanciano2021optimal,santos2011instrumental,LiMiao2022}. 
In contrast, we characterize an alternatve condition for the target functional to be strongly identifiable, without controlling the level of ill-posedness of the conditional moment equations at all. We also reveal its close connection to ostensibly different conditions in \citet{ai2007estimation,ichimura2022influence} (see \Cref{sec: connection} for an expanded discussion).

Our proposed estimation method is based on the minimax estimation framework formalized in \cite{DikkalaNishanth2020MEoC,bennett2020variational}. 
Such minimax methods have been employed in average treatment effect estimation under unconfoundedness \citep{hirshberg2021augmented,kallus2020generalized,chernozhukov2020adversarial} and policy evaluation  \citep{kallus2018balanced,FengYihao2019AKLf,YangMengjiao2020OEvt,UeharaMasatoshi2021FSAo}, but in these settings the nuisances are inherently unique regression functions. 
The minimax framework and its variants have also been successfully applied to causal inference under unmeasured confounding, including IV estimation \citep{LewisGreg2018AGMo,ZhangRui2020MMRf,LiaoLuofeng2020PENE,NIPS2019_8615,MuandetKrikamol2019DIVR} and proximal causal inference \citep{kallus2021causal,GhassamiAmirEmad2021MKML,mastouri2021proximal}, but this literature typically assumes that the unknown nuisance functions are uniquely identified whenever considering inference.

In contrast, our paper tackles the challenge of inference when the unknown functions are not unique solutions. 
To this end, we employ penalization to target certain unique nuisance function among all possible ones. Penalization is a common technique for solving ill-posed inverse problems \citep{Carrasco2007,engl1996regularization}, which has been in particular applied to series or kernel estimation for underidentified conditional moment models \citep{chen2012estimation,santos2011instrumental,babii2017completeness,chen2021robust,escanciano2021optimal,LiMiao2022,florens2011identification}. 
Our paper shows the effectiveness of penalization in the general minimax estimation framework and investigates the impact of penalization on the estimation of linear functionals.

Our paper is also related to the debiased machine learning literature \citep[see][and the references therein]{chernozhukov2018double}. 
This literature typically studies the estimation and inference of smooth functionals of certain  regression functions {that are inherently unique, in contrast to solutions of general moment restrictions.} 
To alleviate the inherent bias of  machine learning regression  estimators, this literature leverages Neyman orthogonal estimating equations for functionals of interest, which requires estimating some Riesz representers first. 
The Riesz representers can be estimated by fitting regressions according to their analytic forms (like propensity scores in average treatment effect estimation) \citep[\eg, ][]{farrell2015robust,farrell2021deep,chernozhukov2017double,chernozhukov2018double,semenova2021debiased}.
Alternatively, some recent literature propose to estimate the Riesz representers by exploiting the representer property directly. 
These methods do not need to derive the analytic forms of Riesz representers on a case-by-case basis, and are therefore termed \textit{automatic} debiased machine learning \citep{chernozhukov2020adversarial,chernozhukov2018learning,chernozhukov2022automatic,chernozhukov2019double,chernozhukov2021automatic,chernozhukov2022riesznet}.
Our proposed method also avoids deriving Riesz representers {explicitly} and is therefore automatic in the same sense. 
However, existing automatic debiased machine learning methods focus on exogenous/unconfounded settings where {nuisances} are naturally well-posed and unique, while we focus on ill-posed nuisances. 
Interestingly, \Cref{eq: delta-eq-2} in our \Cref{assump: new-nuisance}, when specialized to 
the exogenous setting with $S = T$, recovers the formulation used to learn Riesz representers in \cite{chernozhukov2021automatic,chernozhukov2022riesznet}.

\section{Functional Strong Identification}\label{sec: identification}

Defining the linear operator $\operator: \Hcal \to \Lcal_2(T)$ by $[\operator h](T) = \Eb{g_1(W)h(S) \mid T}$, 
the set of all solutions to \Cref{eq: cond-moment-h} is given by
\begin{align}\label{eq: H0}
\Hcal_0 = \braces{h\in\Hcal:[\operator h](T) = \Eb{g_2(W)\mid T}} = h^\star +  \Ncal(\operator).
\end{align}
Therefore, the conditional moment restriction in \Cref{eq: cond-moment-h} uniquely identifies the nuisance function $h^\star$ only when 
the linear operator $\operator$ is injective so that $\Ncal(\operator) = \braces{0}$. As discussed in \cref{sec: setup}, this condition 
often fails.

Fortunately, even when the primary nuisance function $h^\star$ is not uniquely identified, 
the functional $\theta^\star$ can still be identifiable.
In this paper, we impose \cref{assump: new-nuisance} 
to enable \emph{both} the identification of and estimation and inference on the functional. 
\Cref{assump: new-nuisance} requires the existence of solutions to the minimization problem in \Cref{eq: delta-eq-2}, or, more succinctly, $\Xi_0\neq\emptyset$. 
It turns out that this condition is non-trivial, and is equivalent to a smoothness condition on the Riesz representer $\alpha$, as will be explained in \Cref{sec: interpretation}.

In the following theorem, we show that \Cref{assump: new-nuisance} immediately implies the identification of the target parameter $\theta^\star$, and solutions $\xi_0\in\Xi_0$ in \Cref{assump: new-nuisance} can be used in the identification. Furthermore, it also justifies that the identification formula given by $\psi$ satisfies a certain doubly robust property.

\begin{theorem}\label{thm:dr-ident}
Let
\begin{equation*}
    \psi(W; h, q) \coloneqq {m(W; h) + q(T)(g_2(W) -  g_1(W)h(S))} \,.
\end{equation*}
If \Cref{assump: new-nuisance} holds, then for any $h \in \Hcal$, $\xi\in\Hcal$, $h_0 \in \Hcal_0$, and $\xi_0\in\Xi_0,$ and letting $q=P\xi$ and $q^\dagger=P\xi_0$, we have
\begin{equation*}
    \EE[\psi(W;h,q)] - \theta^\star = \langle P(h-h_0), P(\xi-\xi_0) \rangle \,.
\end{equation*}
Consequently, we have
\begin{equation*}
\theta^\star = \Eb{m(W; h_0)} = \Eb{q^\dagger(T)g_2(W)} = \Eb{\psi(W; h_0, q^\dagger)} \,,
\end{equation*} 
and
\begin{align}\label{eq: DR-identification}
 &\abs{\Eb{\psi(W; h, q)} - \theta^\star} = \abs{\langle{P\prns{h-h_0}, P\prns{\xi - \xi_0}}\rangle} \le \|P\prns{h-h_0}\|_2\|P\prns{\xi - \xi_0}\|_2.
 \end{align}
for any such $h$, $\xi$, $h_0$, $\xi_0$, $q$, and $q^\dagger$.
\end{theorem}

\Cref{eq: DR-identification}
shows that the doubly robust identification formula enjoys a mixed bias property. 
In particular we see the identification formula is doubly robust in that $\Eb{\psi(W; h, q)}$ equals $\theta^\star$ if either $P\prns{h}=P\prns{h_0}$ or  $P\prns{\xi}=P\prns{\xi_0}$.
The property in \Cref{eq: DR-identification}  
 means that if we could plug estimates $\hat h$ and $\hat \xi$ into the doubly robust formula to estimate the target functional, then  the estimation bias only depends on the \emph{product} of their estimation errors, in terms of \emph{weak} metrics $\|P(\cdot)\|_2$. 
 If either $\hat h$ or $\hat q$ is consistent in terms of the weak metric, then the resulting functional is consistent.
Importantly, these weak-metric errors   can be directly bounded without invoking any additional  ill-posedness measure. 
In practice, we will need to estimate the debiasing nuisance $q^\dagger=P\xi_0$,  so even with an estimator $\hat\xi$ for $\xi_0$,  we need to additionally approximate the operator $P$. 
But we will show in \Cref{sec: est-nuisance} that this does not change the implication of \Cref{eq: DR-identification}: the estimation bias of the functional does not depend on any ill-posedness measure.  
Moreover, since the functional estimation bias only depends on the \emph{product} of nuisance estimation errors, the bias   
can be asymptotically negligible even if each  nuisance is estimated nonparametrically with a sub-root-$n$ weak-metric-error convergence rate.

The doubly robust nature of  Lemma 1 also shows that Neyman orthogonality holds, meaning that the doubly robust identification formula is insensitive to perturbations to the nuisances.  
Neyman orthogonality plays a pivotal role in the recent debiased machine learning literature that has enabled the use of nuisances estimated by flexible machine learning methods \citep[\eg, ][]{chernozhukov2016locally,chernozhukov2019double,chernozhukov2021automatic,chernozhukov2022automatic}. 
Lemma 1 is  notable in that double robustness and Neyman orthogonality hold for underidentified nuisances.

As a side result of independent interest, we note that \Cref{thm:dr-ident} also implies the following generalization of the result of \cite{rosenbaum1983central}: among the infinity of moment restrictions that $h$ needs to satisfy that are implied by the conditional moment equations, for consistency of the functional, it suffices to estimate an $h$ that satisfies orthogonality of the residual with $q^\dagger$. Testing more moments can increase efficiency but is not required for consistency. This generalizes the result of \cite{rosenbaum1983central} to functionals of endogenous regressions problems:
\begin{lemma}\label{lem:rosenbaum}
Let $h^\dagger$ be any function that satisfies the moment restriction:
\begin{align}
\E[q^\dagger(T)\, (g_2(W) - g_1(W)\,h^\dagger(S))] = 0
\end{align}
Then $\theta^*=\Eb{m(W;h^\dagger)}$.
\end{lemma}
\begin{proof}
First note that by \Cref{assump: new-nuisance}, the first order condition of the minimization problem and the fact that $\Hcal$ is a closed linear space, we have that for any $h\in H$: $\E[m(W;h)]=\E[q^\dagger(T)\, g_1(W)\, h(S)]$. The result follows from the following sequence of identities and \Cref{thm:dr-ident}:
\begin{align}
\Eb{m(W;h^\dagger)} = \Eb{q^\dagger(T)\, g_1(W)\, h^\dagger(S)} = \Eb{q^\dagger(T) g_2(W)} = \theta^*
\end{align}
\end{proof}
Moreover, in the spirit of the Targeted Minimum Loss (TMLE) framework, the above result also suggests a post-processing targeted correction for any initial estimate $h^{(0)}$, by solving the linear IV problem, of estimating the effect $\epsilon$ of $\xi_0(S)$ on the residual $g_2(W)-g_1(W)\,h^0(S)$ with instrument $q^\dagger(T)$, i.e. $\epsilon$ solves the equation: $\Eb{q^\dagger(T)\, (g_2(W) - g_1(W)\, h^0(S) -\epsilon\, \xi_0(S))}$ and setting $h^{(1)} = h^{(0)} + \epsilon\, \xi_0$, we get that the plug-in estimate $\Eb{m(W;h^{(1)})}$ is doubly robust without the need for a correction (i.e. the correction is by definition zero).

In \Cref{sec: est-functional}, we will estimate the target parameter by plugging nuisance estimators into the doubly robust identification formula. 
The robustness properties shown in \Cref{thm:dr-ident} enable the $\sqrt{n}$-consistency and asymptotic normality of the resulting functional estimator, under generic high-level conditions that permit flexible nonparametric nuisance estimators.
Interestingly, the doubly robust identification formula with the debiasing nuisance $q^\dagger = P\xi_0$ recovers the influence function derived in \cite{ichimura2022influence} when specialized to our setup (see \Cref{sec: ichimura}). 

We note that the results so far are based on  a closed linear class $\Hcal$ for the primary nuisance $h^\star$. 
It turns out that we can extend them  to a closed and convex class  $\Hcal$.  
See  \Cref{sec: convex} for details. 

\subsection{Interpreting the Strong Identification Assumption}\label{sec: interpretation}
In this part, we aim to demystify our strong identification condition in \Cref{assump: new-nuisance} by showing that it implicitly restricts the Riesz representer $\alpha$ of the target linear functional. We also compare \Cref{assump: new-nuisance} with many other assumptions in the existing literature. 

\subsubsection{Restrictions on Riesz Representers}
\begin{theorem}\label{thm: new-nuisance}
\Cref{assump: new-nuisance} holds if and only if the Riesz representer $\alpha$ in \Cref{eq: Riesz} satisfies that $\alpha \in \Rcal(P^\star P)$. Moreover, $\braces{\xi\in\Hcal: P^\star P\xi = \alpha}$ is equal to the solution set $\Xi_0$ in \Cref{assump: new-nuisance}. 
\end{theorem}

\Cref{thm: new-nuisance} shows that \Cref{assump: new-nuisance} requires the Riesz representer $\alpha$ to lie in the range space of operator $P^\star P$, where $P^\star: \Lcal_2(T) \to \Hcal$ is the adjoint of $\operator$. It can be shown that $P^\star$ is given by  $[\operator^\star q](S) =  \Pi_\Hcal\bracks{\Eb{g_1(W)q(T) \mid S} \mid S} = \Pi_{\Hcal}\bracks{g_1(W)q(T) \mid S}$ for any $q \in \Lcal_2(T)$, where $\Pi_{\Hcal}(\cdot \mid S)$ is the projection operator onto $\Hcal$.
Since $\Hcal$ is a closed linear space, this projection is well-defined \citep[Theorem 1 of Section 3.12 in][]{luenberger1997optimization}.

Although \Cref{assump: new-nuisance} and \Cref{thm: new-nuisance} impose equivalent restrictions, the formulation in \Cref{assump: new-nuisance} is more amenable to estimation, as it  involves only the functional $\Eb{m(W; h)}$ but not the Riesz representer $\alpha(S)$. This obviates the need to derive the form of the Riesz representer $\alpha(S)$, which is in line with the spirit of the recent automatic debiased machine learning methods for functionals of exogenous regression functions \citep{chernozhukov2020adversarial,chernozhukov2021automatic,chernozhukov2022automatic,chernozhukov2022riesznet}. These methods propose ways to learn Riesz representers solely based on the functionals of interest, without needing to derive the form of the Riesz representers on a case-by-case basis. 
Our formulation in \Cref{assump: new-nuisance} has the same advantages.

\begin{example}\label{ex: svd}
To understand \Cref{assump: new-nuisance} and \Cref{thm: new-nuisance}, it is instructive to consider a compact linear operator $P$
(but it is worth noting that compactness is not needed in our identification and estimation theory).
Let $\braces{\sigma_i, u_i, v_i}_{i=1}^\infty$ denote singular value decomposition of the compact operator $P$, where $\braces{u_i}_{i=1}^\infty, \braces{v_i}_{i=1}^\infty$ are orthonormal bases in $\Lcal_2(T)$ and $\Hcal \subseteq \Lcal_2(S)$ respectively, and $\sigma_1 \ge \sigma_2 \ge \dots$ are singular values.
Then the adjoint operator $P^\star$ 
 has the decomposition $\{\sigma_i, v_i, u_i\}_{i=1}^\infty$ and $P^\star P$ has the decomposition $\{\sigma_i^2, v_i, v_i\}_{i=1}^\infty$.
The Riesz representer $\alpha \in \Hcal$ in \Cref{eq: Riesz} can be represented as $\alpha=\sum_{i=1}^{\infty} \gamma_i v_i$ with $\sum_{i=1}^{\infty} \gamma_i^2 < \infty$ and any $\xi\in\Hcal$ can be represented as  $\xi=\sum_{i=1}^{\infty} \beta_i v_i$ with  $\sum_{i=1}^{\infty} \beta_i^2 < \infty$. 

Note that  
$P\xi = \sum_{i=1}^{\infty} \beta_i P v_i = \sum_{i=1}^{\infty}\beta_i \sigma_i  u_i$, and $\Eb{m(W; \xi)} = \Eb{\alpha(S)\xi(S)} = \sum_{i=1}^\infty \gamma_i\beta_i$. 
 Thus the optimization problem in \Cref{assump: new-nuisance} can be equivalently written as follows:
\begin{align}\label{eq: nuisance-svd}
\min_{\beta_1, \beta_2, \dots}\frac{1}{2}\sum_{i=1}^\infty \sigma_i^2 \beta_i^2 - \sum_{i=1}^\infty \gamma_i \beta_i, ~~ \text{subject to } \sum_{i=1}^{\infty} \beta_i^2 < \infty.
\end{align}
The optimal interior solution is given by $\beta_{0, i} = \gamma_i/\sigma_i^2$, which corresponds to the  solution $\xi_0 = \sum_{i=1}^{\infty} \gamma_iv_i/\sigma_i^2$. Then, \Cref{assump: new-nuisance} requires that $\sum_{i=1}^\infty\beta^{2}_{0, i} =  \sum_{i=1}^\infty\gamma_i^2/\sigma_i^4 < \infty$. 
This condition requires that $\gamma_i$ must be zero on eigenfunctions for which $\sigma_i = 0$ (\ie, basis functions for $\Ncal(P)$).
This means that the Riesz representer $\alpha$ must belong to $\Ncal(\operator)^\perp$.
Moreover, the condition imposes that the Riesz representer cannot be supported heavily on the lower part of the right eigenfunctions of the conditional expectation operator $P$. 
This means that the Riesz representer has to be \emph{smooth} enough relative to the spectrum of the conditional expectation operator $P$.
\Cref{thm: new-nuisance} states the solution $\xi_0$ to \Cref{eq: nuisance-svd} can be also characterized as a root to  $\alpha = P^\star P\xi_0$. Indeed, for  $\xi_0 = \sum_{i=1}^{\infty} \beta_{0, i} v_i$ with $\sum_{i=1}^{\infty} \beta_{0, i}^2 < \infty$, we have $P^\star P\xi_0 = \sum_{i=1}^{\infty}\sigma_i^2\beta_{0, i} v_i$. Equating it to $\alpha = \sum_{i=1}^{\infty} \gamma_i v_i$ gives $\beta_{0, i} = \gamma_i/\sigma_i^2$ for all $i$, which recovers the solution to \Cref{eq: nuisance-svd}. 
This verifies the conclusion  of \Cref{thm: new-nuisance} in the special case of a compact linear operator $P$.
Note also that a similar understanding of \Cref{assump: new-nuisance} can be made for more general compact operators $P$, in terms of a more general (possibly continuous) spectral decomposition; see \emph{e.g.} \citet{cavalier2011inverse} for details.
In addition, note that \Cref{thm: new-nuisance} holds more generally for non-compact linear operators, such as the operators in \Cref{ex: proximal,ex: shadow}.

 \end{example} 

\subsubsection{Strong Identification versus Identification}
\Cref{ex: svd} shows that, for compact $P$, \Cref{assump: new-nuisance} automatically restricts the Riesz representer $\alpha$ to the orthogonal complement of the null space of operator $P$. That is, it restricts to $\alpha \in \Ncal(\operator)^\perp$. 
This can be shown to be the sufficient and necessary condition for the identification of the target parameter for general $P$, by slightly generalizing the analysis in \cite{severini2006some}.

\begin{lemma}\label{lemma: identifiability}
The parameter $\theta^\star$ is identifiable, i.e., $\theta^\star = \Eb{m(W; h_0)}$ for any $h_0 \in \Hcal_0$,  if and only if $\alpha \in \Ncal(\operator)^\perp = \cl\prns{\Rcal\prns{P^\star}}$. 
\end{lemma}

\Cref{lemma: identifiability} gives the weakest condition for the identification of the target parameter $\theta^\star$ when the primary nuisance $h^\star$ can be underidentified. 
This condition trivially holds if the nuisance function $h^\star$ is identified to begin with, so that $\Ncal(\operator) = \braces{0}$ and $\cl\prns{\Rcal\prns{P^\star}} = \Hcal$. 
In fact, the converse is also true: $h^\star$ is identifiable if and only if every continuous linear functional of $h^\star$ is identifiable (see \Cref{lemma: nuisance-id} in \Cref{sec: support}).
But for a given parameter of interest, \Cref{lemma: identifiability} shows that the identifiability of $h^\star$ is not neccessary for the identifiability of $\theta^\star$, since  $\theta^\star$ may capture identifiable parts of the nuisance $h^\star$,  provided that the corresponding Riesz representer is orthogonal to the null space of $P$.

Although the condition in \Cref{lemma: identifiability}  is enough for the identification of $\theta^\star$, it alone does not suffice for inference on $\theta^\star$. 
Thus, we impose the strong identification condition in \Cref{assump: new-nuisance}. 
\Cref{assump: new-nuisance} is certainly always stronger than the identification condition in \Cref{lemma: identifiability}, since by \Cref{thm: new-nuisance} we know that \Cref{assump: new-nuisance} is equivalent to $\alpha \in \Rcal(P^\star P)$, and $\Rcal(P^\star P) \subseteq \cl(\Rcal(P^\star))$.
Thus our \Cref{assump: new-nuisance} requires the target functional to be not only identifiable, but also regular enough so that inference on it is possible.

\subsubsection{Strong Identification versus Existence of Debiasing Nuisances}
Our \Cref{assump: new-nuisance} is also closely related   to the  condition $\alpha\in\Rcal(P^\star)$ proposed in \citet{severini2012efficiency}.
This condition is equivalent to the existence of $q_0 \in \Lcal_2(T)$ that solves 
\begin{align}\label{eq: cond-moment-q}
[P^\star q](S) = \Pi_\Hcal\bracks{g_1(W)q(T) \mid S} = \alpha(S),
\end{align}
or equivalently, 
\begin{align}\label{eq: Q0}
\Qcal_0 \ne \emptyset,\quad \text{where}\quad \Qcal_0 \coloneqq \braces{q \in \Lcal_2\prns{T}: [P^\star q](S) = \alpha(S)} 
\end{align}

Compared to the identification condition in \Cref{lemma: identifiability}, this condition additionally rules out $\alpha$ on the boundary of $\Rcal(P^\star)$.
In the setting of NPIV regression, \citet{severini2012efficiency} shows that this is a \emph{necessary} condition for the $\sqrt{n}$-estimability of IV functionals.
\cite{deaner2019nonparametric} also shows that this condition is sufficient and necessary for the robust estimation of IV functionals when the IV exclusion restriction is misspecified.
Furthermore, below we show that this condition is the minimal condition for the existence of debiasing nuisances $q_0$ such that the corresponding doubly robust identification formula has the  robustness properties akin to \Cref{thm:dr-ident}.

\begin{theorem}\label{lemma: bias-product-general}
If $\alpha\in\Rcal(P^\star)$, then the conclusions in \Cref{thm:dr-ident} hold for any $q_0 \in \Qcal_0$. In particular, 
$\theta^\star = \Eb{\psi(W; h_0, q_0)}$ for any $h_0 \in \Hcal_0$ and $q_0 \in \Qcal_0$.  
Moreover, for two fixed functions $h_0 \in \Hcal$ and $q_0 \in \Lcal_2(T)$, we have 
$h_0 \in \Hcal_0$ and $q_0 \in \Qcal_0$ if and only if, for for any $h \in \Hcal$, $q \in \Lcal_2(T)$,
\begin{align}\label{eq: DR-identification-general}
 \abs{\Eb{\psi(W; h, q)} - \theta^\star} 
    &= \abs{\langle{P\prns{h-h_0}, q-q_0\rangle}} = \abs{\langle{{h-h_0}, P^\star\prns{q-q_0}\rangle}} \nonumber \\
    &\le \min\braces{ \|P\prns{h-h_0}\|_2\|q-q_0\|_2, \|h-h_0\|_2\|P^\star\prns{q-q_0}\|_2}. 
 \end{align}  
\end{theorem}

\Cref{lemma: bias-product-general} shows that $h_0 \in \Hcal_0$ and $q_0 \in \Qcal_0$ is the sufficient and necessary condition for the mixed bias property of the doubly robust identification formula. This also implies that $h_0 \in \Hcal_0$ and $q_0 \in \Qcal_0$ is the minimal condition for the double robustness and  Neyman orthogonality of the doubly robust identification formula. 
The results in \Cref{thm:dr-ident} under our \Cref{assump: new-nuisance} can be viewed as a corollary of \Cref{lemma: bias-product-general}, with $q_0$ specialized to the particular debiasing nuisance function $q^\dagger = P\xi_0$ for $\xi_0 \in \Xi_0$. 
Below we show that our specialized debiasing nuisance function $q^\dagger = P\xi_0$ is actually the minimum-norm function in $\Qcal_0$ in \Cref{eq: Q0}. This shows that our \Cref{assump: new-nuisance} imposes that the minimum-norm debiasing nuisance in $\Qcal_0$ belongs to $\Rcal(P)$.
\begin{lemma}\label{lemma: minimum-norm}
Let $\xi_0\in\Hcal$.
Then, $\xi_0\in\Xi_0$ if and only if $q^\dagger=P \xi_0$ is the minimum-norm element of $\Qcal_0$, namely $q^\dagger = \argmin_{q \in \Qcal_0}\|q\|_2^2$.
\end{lemma}

Although \citet{severini2012efficiency} shows that the condition $\alpha \in \Rcal\prns{P^\star}$, or equivalently, $\Qcal_0 \ne \emptyset$, is a \emph{necesary} condition for the $\sqrt{n}$-estimability of the target functional, it alone is not sufficient, especially if the inverse problem in \Cref{eq: cond-moment-h} is severely ill-posed \citep{chen2015sieve}. 
To overcome this challenge, we impose   our strong identification condition in \Cref{assump: new-nuisance}. 
Note that our condition $\alpha\in \Rcal(P^\star P)$ strengthens the condition $\alpha\in\Rcal(P^\star)$, since we always have $\Rcal(P^\star P) \subseteq \Rcal(P^\star)$. 
In particular,  $\Rcal(P^\star P)$ is a \emph{strict} subset of  $\Rcal(P^\star)$ unless $\Rcal(P)$ is a closed set and the inverse problem in \Cref{eq: cond-moment-h} is well-posed \citep{Carrasco2007}. 
We do not assume a closed $\Rcal(P)$ and therefore well-posedness. Instead, we restrict the Riesz representer $\alpha$ to $\Rcal(P^\star P)$, This imposes a stronger restriction on the Riesz representer than the condition $\alpha\in\Rcal(P^\star)$, but it allows the inverse problem in \Cref{eq: cond-moment-h} for the primary nuisance function to be arbitrarily ill-posed. 
As we will show later, this kind of restriction will enable us to construct  $\sqrt{n}$-consistent and asymptotically normal estimators for the target functional, even when the primary nuisance is weakly identified. 
See \Cref{sec: challenge} for more discussions.

\subsubsection{Relation to Other Conditions}

\Cref{thm: new-nuisance} shows that our \Cref{assump: new-nuisance} is equivalent to 
$\alpha \in \Rcal(P^\star P)$. 
This can be seen as a so-called ``source condition'' on the Riesz representer $\alpha$, restricting the regularity of $\alpha$ with respect to the linear operator $P$ \citep[\eg, ][]{Carrasco2007,florens2011identification}.

Our  \Cref{assump: new-nuisance} is, however, fundamentally different 
from imposing source conditions on the primary nuisance function, such as the IV regression itself  
\citep[\eg, ][]{Carrasco2007,florens2011identification,babii2017completeness,darolles2011nonparametric,singh2019kernel}. 
For example, \citet{darolles2011nonparametric} assumes that the true IV  regression $h^\star$ lies in the space $\Rcal((P^\star P)^{\beta/2})$ for some exponent $\beta > 0$.
This source condition directly restricts the smoothness of the IV regression and the degree of ill-posedness of the IV conditional moment restriction.

Alternatively, some other literature defines and bounds so-called ``ill-posedness measures" of the inverse problem defining $h^\star$
relative to a function class \citep[\eg,][]{chen2012estimation,chen2015sieve,chen2018optimal,DikkalaNishanth2020MEoC,kallus2021causal}.
These ill-posedness measures bound the ratio between weak-metric and strong-metric errors, so that bounds on the former yield bounds on the latter, making $h^\star$ itself strongly identified, that is, unique (in the function class) and well-estimable. This would allow us, in particular, to do inference on the functional by bounding the strong-metric error in the second branch of \cref{eq: DR-identification-general}. For inference on the functional, we could alternatively impose similar ill-posedness measures on $q_0\in\Qcal_0$, so as to instead control the strong-metric error in the first branch \citep[\eg,][]{kallus2021causal}. In either case, we are effectively imposing strong identification of \emph{functions}.
Our \cref{assump: new-nuisance} is fundamentally different
 and complements this literature.

We remark that our \Cref{assump: new-nuisance} is also deeply connected to several ostensibly different conditions in some previous literature. 
In \Cref{sec: ichimura}, we equivalently characterize $\xi_0 \in \Xi_0$ in terms of a certain projection of $q_0\in\Qcal_0$, thereby showing that our \Cref{assump: new-nuisance}  is related to a condition in \cite{ichimura2022influence}.
Moreover, in \Cref{sec: chen},  we show  that a key condition in \cite{ai2007estimation}, when specialized to our setting, is actually equivalent to our \Cref{assump: new-nuisance}. 
In \Cref{sec: partial-linear-iv}, we show that a condition in \cite{chen2021robust} for partially linear IV models also implicitly imposes our \Cref{assump: new-nuisance}. Therefore, our \Cref{assump: new-nuisance} also provides new and generalized interpretations for the assumptions in these existing works.

\subsection{Revisiting the Examples}\label{sec: examples2}

We now revisit the examples from \cref{sec: examples} to instantiate the conditions discussed above.

\begin{continuance}[Functionals of  NPIV Regression]{\ref{ex: IV-reg}}
Consider the parameter $\theta^\star$ given in
\Cref{eq: best-linear-iv}. 
Then \Cref{eq: Q0} posits the existence of functions $q_{0, 1},\dots, q_{0, d} \in \Lcal_2(Z)$ for $d = d_X$, such that  $q_0 = \prns{q_{0, 1},\dots, q_{0, d}}$ solves  
\begin{align}\label{eq: IV-q}
\Eb{q(Z) \mid X} =  \alpha(X) =
X.
\end{align}

\Cref{assump: new-nuisance} further requires the existence of $\xi_0 = (\xi_{0, 1}, \dots, \xi_{0, d})$ such that $\xi_{0, i} \in \Lcal_2(X)$ and $q^\dagger =\prns{\Eb{\xi_{0, 1}(X) \mid Z},\dots, \Eb{\xi_{0, d}(X)\mid Z}}^\top$ satisfies \Cref{eq: IV-q}. 
Alternatively, any such $\xi_0$ is given by 
\begin{align}\label{eq: IV-delta}
\xi_{0, i} \in \argmin_{\xi_i\in\Lcal_2(X)} \Eb{\prns{\Eb{\xi_i(X) \mid Z}}^2} - \Eb{\alpha_i(X)\xi_i(X)},
\end{align}
where $\alpha_i$ is the $i$th coordinate of $\alpha$ in \Cref{eq: IV-q}. According to \Cref{lemma: bias-product-general,thm:dr-ident},  even when the NPIV regression $h^\star$ is unidentifiable, the parameter $\theta^\star$ is still identifiable by any $h_0$ solving \Cref{eq: IV-h}, any $q_0$ solving \Cref{eq: IV-q}, or any $\xi_0$ solving \Cref{eq: IV-delta}.

\cite{escanciano2021optimal} also study the estimation of and inference on the best linear approximation coefficient $\Eb{XX^\top}^{-1}\theta^\star$, allowing the NPIV  regression $h^\star$ to be  unidentifiable. Assumption 3 in their paper is equivalent to the existence of a function $q_0$ that solves \Cref{eq: IV-q} (although $q_0$ is not necessarily in the range space $\Rcal(P)$). 
They propose a penalized linear sieve estimator that can converge to a particular solution $q_0$ to \Cref{eq: IV-q}, and then use it to construct their  estimator for $\theta^\star$. 
They restrict the ill-posedness of the NPIV regression by imposing a source condition \citep[assumption A4]{escanciano2021optimal} and prove that their resulting estimator can achieve desirable asymptotic properties.

In our paper, we will accommodate general flexible function classes. This permits going beyond sieve estimation and its involved technical assumptions (Assumptions A.2--A.6 in \citealp{escanciano2021optimal}) and allows us to rely instead on high-level conditions for approximation by general hypothesis classes. To enable this, we instead incorporate penalization into the general minimax estimation framework with general hypothesis classes \citep{DikkalaNishanth2020MEoC,kallus2021causal} and we employ the doubly robust identification formula in \Cref{thm:dr-ident} to cancel out estimation errors in these nuisances so that we do not need strong assumptions to characterize their behavior.
This leverages highly flexible machine learning nuisance estimators and the resulting functional estimator still has desirable asymptotic  properties. 
Moreover, under \Cref{assump: new-nuisance}, these can be achieved even without restricting the ill-posedness of the NPIV problem.
\end{continuance}

\begin{continuance}[Proximal Causal Inference]{\ref{ex: proximal}}
For the average treatment effect $\theta^\star$ identified via 
\Cref{eq: proximal-parameter},
\Cref{eq: Q0} corresponds to the existence of another nuisance function 
$q_0(Z, X, A)$ solving
 \begin{align}\label{eq: proximal-q}
 \Eb{q(Z, X, A) \mid V, X, A} = {\alpha}(V, X, A) = \frac{A-\Prb{A=1\mid  V, X}}{\Prb{A=1\mid  V, X}(1-\Prb{A=1\mid  V, X})}.
 \end{align}
We note that the $q_0$ here corresponds to a treatment bridge function, which is a second type of  bridge function in proximal causal inference \citep{cui2020semiparametric,kallus2021causal}.

 Our \Cref{assump: new-nuisance} further requires the existence of $\xi_0 \in \Lcal_2(V, X, A)$ such that $q^\dagger(Z, X, A) = [P\xi_0](Z, X, A) =  \Eb{\xi_0(V, X, A) \mid Z, X, A}$ satisfies \Cref{eq: proximal-q}, \ie, there exists a treatment bridge function in the range space $\Rcal(P)$. Any such $\xi_0$ is also given by
\begin{align}\label{eq: proximal-delta}
\xi_{0} \in \argmin_{\xi\in\Lcal_2(V, X, A)} \Eb{\prns{\Eb{\xi(V, X, A) \mid Z, X}}^2} - \Eb{\alpha(V, X, A)\xi(V, X, A)}.
\end{align}
 Then \Cref{lemma: bias-product-general,thm:dr-ident} imply that $\theta^\star$ can be identified by \textit{any} $h_{0}$ solving \Cref{eq: proximal-h},  \textit{any} $q_{0}$ solving \Cref{eq: proximal-q}, and any $\xi_0$ solving \Cref{eq: proximal-delta}.

Although any solutions to \Cref{eq: proximal-q,eq: proximal-h,eq: proximal-delta} identify $\theta^\star$, multiplicity of solutions raises significant challenges for statistical inference. 
Indeed, even if uniqueness is not assumed for identification, the existing proximal causal inference literature largely assumes uniqueness for statistical inference
\citep[\eg,][]{cui2020semiparametric,kallus2021causal,GhassamiAmirEmad2021MKML,mastouri2021proximal,singh2020kernel,miao2018a}. 
One exception is \cite{imbens2021controlling}, which handles the nonunique nuisances by a penalized generalized method of moment estimator, but their approach only applies in their specific panel-data setting where the nuisance is linearly parameterized. 
In this paper, we will develop new estimators and inferential procedures that are robust to the nonuniqueness of general nonparametric nuisance functions.  
Another exception is \cite{deaner2018proxy}, which focuses on the conditional average treatment effect, and establishes identification and well-posedness without requiring identification of the bridge function.
Moreover, existing literature on nonparametric proximal causal inference  restricts the ill-posedness of the inverse problem in \Cref{eq: proximal-h} for the primary bridge function, by either assuming source conditions on the bridge function \citep{singh2020kernel,mastouri2021proximal} or relying on ill-posedness measures \citep{GhassamiAmirEmad2021MKML,kallus2021causal}. 
Our paper shows that these are not necessary if the linear functional is regular enough in the sense that there exist solutions to \Cref{eq: proximal-delta}.
Given \Cref{thm: new-nuisance}, this follows if the observable propensity function $P(A=1 \mid V,X)$ is sufficiently regular.
 \end{continuance}

\begin{continuance}[Missing-Not-at-Random Data with Shadow Variables]{\ref{ex: shadow}}
For the parameter $\theta^\star$ in \Cref{eq: shadow-parameter}, 
\Cref{assump: new-nuisance} requires 
 the existence of $\xi_{0}\in \Lcal_2(X, V)$  such that  
\begin{align}\label{eq: shadow-delta}
\xi_{0} \in \argmin_{\xi\in\Lcal_2(X, V)} \Eb{\prns{\Eb{\xi(X, V) \mid X, Z}}^2} - \Eb{\alpha(X, V)\xi(X, V)}.
\end{align}
Then \Cref{thm:dr-ident} implies that $\theta^\star$ can be identified by {any} $h_{0}$ solving \Cref{eq: shadow-h} or any $\xi_0$ solving \Cref{eq: shadow-delta}. 
 Moreover, \Cref{thm: new-nuisance} means that any such $\xi_{0}$  can be equivalently characterized by 
\begin{align}\label{eq: shadow-q}
 &\Eb{Aq^\dagger(X, Z)\mid X, V} = \alpha(X, V) = V, \\
 \text{where } & q^\dagger(X, Z) = [P\xi_0](X, Z) = \Eb{\xi_0(X, V) \mid X, Z}.   \nonumber 
 \end{align}

\cite{LiMiao2022} also assumes a condition that  requires the existence of a $q_0$ that satisfies \Cref{eq: shadow-q} (although their $q_0$ function is not necessarily in the range space $\Rcal(P)$). 
Then they develop estimation and inferential methods robust to nonunique nuisances. 
Their method extends that in \cite{santos2011instrumental}: they 
first use a linear sieve estimator proposed by \cite{chernozhukov2007estimation} to estimate the \emph{set} of solutions to \Cref{eq: shadow-q} (namely the set $\Qcal_0$ defined in \Cref{eq: Q0}), and then pick a unique element therein that maximizes a certain criterion.  
In contrast, the methods we will propose can accommodate flexible hypothesis classes, rely on high-level conditions about these classes, and avoid the challenging task of estimating  solution sets to conditional moment restrictions.
 \end{continuance}

\subsection{Challenges with Ill-posed  Nuisance Estimation}\label{sec: challenge}
In \Cref{lemma: bias-product-general}, we show that the doubly robust identification formula satisfies the Neyman orthogonality property.
Following the recent literature on debiased machine learning cited above, it can therefore be hoped we can simply plug in any flexible nuisance estimators and use the debiased machine learning inference algorithm.

However, because of the \emph{ill-posedness} of the nuisance estimation problem, statistical inference on the target parameter based on asymptotic normality can be very challenging. To illustrate the challenge, consider some generic nuisance estimators $\hat h, \hat q$ for certain $h_0 \in \Hcal_0, q_0 \in \Qcal_0$, and the corresponding doubly robust estimator for the target parameter:
\begin{align*}
\tilde\theta = \frac{1}{n}\sum_{i=1}^n \psi(W_i; \hat h, \hat q).
\end{align*}
The estimation error of this estimator is decomposed as follows: for any $h_0 \in \Hcal_0, q_0 \in \Qcal_0$, 
\begin{equation}\label{eq: error-decompose}
\begin{aligned}
\sqrt{n}({\tilde \theta - \theta^\star})
    &= \hG_n\prns{\psi(W; h_0,  q_0) - \theta^\star}  +
    \hG_n\prns{\psi(W; \hat h, \hat q) - \psi(W; h_0,  q_0)} \\
    &\qquad\qquad\qquad\qquad\qquad\qquad\qquad\qquad + \sqrt{n}\hP\prns{\psi(W; \hat h, \hat q) - \psi(W; h_0,  q_0)}.
\end{aligned}
\end{equation}

While the first term in \Cref{eq: error-decompose} can be proved to be asymptotically normal by the central limit theorem, the second and third terms depend on the estimation errors of $\hat h, \hat q$ and they 
are particularly challenging to handle because of the ill-posedness of the nuisance estimation. The second term suffers from issues of non-uniqueness while the third term suffers from issues of discontinuity of inverse problems.

The second term in \Cref{eq: error-decompose} is a stochastic equicontinuity term. 
To make this term negligible, we  typically need to require that the nuisance estimators $\hat h, \hat q$ converge to fixed $h_0 \in \Hcal_0, q_0 \in \Qcal_0$,  in terms of strong metrics like the $L_2$ norm \citep[\eg, Lemma 19.24 in][]{van2000asymptotic}.
 This remains the case even if we employ cross fitting as described in \Cref{def: theta-est} below \citep[\eg, see discussions below Assumption 3.2 in][]{chernozhukov2018double}.
In this paper, we study ill-posed inverse problems where nuisances $h_0, q_0$ can be non-unique. 
In this case, 
common nuisance estimators $\hat h, \hat q$ typically do not converge to any fixed asymptotic limits.
As a result, the stochastic equicontinuity term  is generally not negligible, and the resulting functional estimator can easily have an intractable asymptotic distribution (see section 3.1 in \citealp{chen2021robust} for a concrete example in the IV setting). 
To overcome this challenge, we will develop penalized nuisance estimators that converge to fixed asymptotic limits even when the nuisances are non-unique, so that the second term in \Cref{eq: error-decompose} is negligible. 

The third term in \Cref{eq: error-decompose} quantifies the bias due to the estimation errors of $\hat h$ and $\hat q$.
According to \Cref{lemma: bias-product-general}, this term can be bounded in terms of either $\|\hat h - h_0\|_2\|P^\star[\hat q-q_0]\|_2$ or $\|P[\hat h-h_0]\|_2\|\hat q - q_0\|_2$. 
If we estimate $\hat h$ and $\hat q$ by directly solving empirical analogues of \Cref{eq: cond-moment-q,eq: cond-moment-h}, then we can establish  convergence rates of $\hat h, \hat q$ in terms of the weak projected metrics  $\|P[\hat h-h_0]\|_2$ and $\|P^\star[\hat q-q_0]\|_2$ \citep[\eg, ][]{chen2012estimation,DikkalaNishanth2020MEoC,kallus2021causal}. 
However, the strong-metric estimation errors, $\|\hat h - h_0\|_2$ or $\|\hat q - q_0\|_2$, may converge \emph{arbitrarily} slowly (or even not at all), depending on the degrees of  ill-posedness of the inverse problems associated with $h_0$ and $q_0$. 
Consequently, when these inverse problems are severely ill-posed,
the resulting functional estimator may converge slowly and not be asymptotically normal.
To overcome this challenge, we impose our \Cref{assump: new-nuisance}, which restricts the inverse problem associated with $q_0$ (see discussions around 
\Cref{lemma: minimum-norm}). 
Under this assumption, we can instead estimate $q^\dagger = \argmin_{q \in \Qcal_0} \|q\|_2$, which by \Cref{lemma: minimum-norm} is equal to $P\xi_0$ for any $\xi_0 \in \Xi_0$.
In the next section, we will propose a minimax estimator for $q^\dagger = P\xi_0$ based on the formulation of $\xi_0$ in \Cref{assump: new-nuisance}, and provide a strong-metric  convergence rate of this estimator that does not involve any additional ill-posedness measure. 
The intuition behind this strong-metric result for $q^\dagger$  stems from the fact that, under \Cref{assump: new-nuisance}, it essentially corresponds to a weak-metric convergence for $\xi_0$, for which we can provide ill-posedness-free rates.
With the strong-metric convergence rate for $\hat q$, we only need a  weak-metric convergence rate of $\hat h$ to make the third term in  \Cref{eq: error-decompose} vanish. 
As a result, our functional estimator can be asymptotically normal even without restricting the ill-posedness of \Cref{eq: cond-moment-h} for the primary nuisance function.

\section{Minimax Estimation of Nuisances}\label{sec: est-nuisance}

Here we present and analyze our minimax estimators of the primary and debiasing nuisances. In this section we use the notation $h_0$ to denote an arbitrary element of $\Hcal_0$, and $h^\dagger$ to denote the minimum-norm element of $\Hcal_0$, to which we will establish consistency. Note again that this element is always unique, since $\Hcal_0$ is a closed linear subspace.
Similarly, we let $q^\dagger = P \xi_0$ for any $\xi_0 \in \Xi_0$, which per \Cref{lemma: minimum-norm} is the minimum-norm element of $\Qcal_0$.

\subsection{Penalized Estimation of the Primary Nuisance Function}\label{sec: primary-nuisance-est}

We first consider estimation of the minimum-norm primary nuisance function $h^\dagger$.
We will consider estimators of the form
\begin{equation}
\label{eq:h-estimator}
    \hat h_n = \argmin_{h \in \Hcal_n} \sup_{q \in \Qcal_n} \EE_n \Big[ \Big(g_1(W) h(S) - g_2(W)\Big) q(T) - \frac{1}{2} q(T)^2 + \mu_n h(S)^2 \Big] - \gamma_n^q \|q\|_\Qcal^2 + \gamma_n^h \|h\|_\Hcal^2 \,,
\end{equation}
where $\Hcal_n$ and $\Qcal_n$ are function classes for empirical minimax estimation, and $\mu_n,\gamma_n^h,\gamma_n^q$ are regularization hyperparameters for the penalized estimation. Importantly, we assume that $\Qcal_n \subseteq \bar\Qcal$ and $\Hcal_n \subseteq \bar\Hcal$ for some fixed normed function sets $\bar\Hcal \subseteq \Hcal$ and $\bar\Qcal \subseteq \Qcal$ that do not depend on $n$, with norms $\|\cdot\|_\Hcal$ and $\|\cdot\|_\Qcal$ respectively. We optimally allow for regularization using these norms, with hyperparameters $\gamma_n^q \geq 0$ and $\gamma_n^h \geq 0$.

Before we provide finite-sample bounds for this class of estimators, we must establish some technical conditions.
First, we assume $g_1$, $g_2$ are bounded. (We use 1 as the bound without loss of generality, since they appear linearly in the conditional moment restriction in \cref{eq: cond-moment-h}.) 
\begin{assumption}
\label{assum:boundedness}
    We have that: (1) $\|g_2\|_2 \leq 1$; and (2) $\|g_1\|_\infty \leq 1$; 
\end{assumption}

Next, we assume that $\Hcal_n$ and $\Qcal_n$ are uniformly bounded, as is $h^\dagger$.
\begin{assumption}
\label{assum:bounded-estimation-h}
    We have that: (1) $\|h\|_\infty \leq 1$ for all $h \in \Hcal_n$; (2) $\|q\|_\infty \leq 1$ for all $q \in \Qcal_n$; and (3) $\|h^\dagger\|_\infty \leq 1$.
\end{assumption}

Next, we require that $\Hcal_n$ can approximate $h^\dagger$, and $\Qcal_n$ can approximate the projections of $\Hcal_n$

\begin{assumption}
\label{assum:universal-approximation-h}
    There exists some $\delta_n < \infty$ such that: (1) there exists $\Pi_n h^\dagger \in \Hcal_n$ such that $\|\Pi_n h^\dagger - h^\dagger\|_2 \leq \delta_n$; and (2) for every $q \in \{P(h - h^\dagger) : h \in \Hcal_n\}$, there exists $\Pi_n q \in \Qcal_n$ such that $\|\Pi_n q - q\|_2 \leq \delta_n$.
\end{assumption}
This kind of condition is standard in the sieve literature (see \emph{e.g.} \citet{chen2007large}). In addition, this condition can be guaranteed by standard univeral approximation results, such as \emph{e.g.} \citet{yarotsky2017error} when $\Hcal_n$ and $\Qcal_n$ are neural net classes, as long as the range of $P$ is sufficiently smooth such that \emph{e.g.} $\{q \in \{P(h - h^\dagger) : h \in \Hcal_n\}$ lies within a Sobolev ball.

Third, we require that some particular function classes defined in terms of $\Hcal_n$ and $\Qcal_n$ have well-behaved critical radii.

\begin{assumption}
\label{assum:complexity-h}
    There exists some $r_n$ that upper bounds the critical radii of the function classes $\{g_1(W) h(S) q(T) : h \in \starcls(\Hcal_n - h^\dagger), q \in \starcls(\Qcal_n), \|h\|_\Hcal \leq 1, \|q\|_\Qcal \leq 1\}$ and $\{q \in \starcls(\Qcal_n) : \|q\|_\Qcal \leq 1\}$. 
\end{assumption}
Examples of bounds on critical radii of such functions classes are considered in \citet{DikkalaNishanth2020MEoC,kallus2021causal} for a variety of choices for $\Hcal_n,\Qcal_n$, such as  H\"older and Sobolev balls, RKHS balls, linear sieves, neural networks, \emph{etc}.
We note that these critical radii are defined in terms of unit norm-bounded subsets of the respective function classes, and therefore do not depend on the actual complexity of $h^\dagger$ (\emph{i.e.} the size of $\|h^\dagger\|_\Hcal$).

Finally, we require one of the two following assumptions, which either completely bounds the complexity of all functions in $\Hcal_n$ or $\Qcal_n$, or ensures that the regularization coefficients $\gamma_n^q$ and $\gamma_n^h$ are sufficiently large to ensure that we can automatically adapt to the complexity of $h^\dagger$.

\begin{assumption*}
\label{assum:bounded-complexity-h}
    There exists some constant $M \geq 1$ such that: (1) $\|h\|_\Hcal, \|h-h^\dagger\|_\Hcal \leq M$ for all $h \in \Hcal_n$; (2) $\|q\|_\Qcal \leq M$ for all $q \in \Qcal_n$; and (3) $\|h^\dagger\|_\Hcal, \|\Pi_n h^\dagger\|_\Hcal \leq M$.
\end{assumption*}

\begin{assumption*}
\label{assum:regularization-h}
    There exists some constants $M \geq 1$ and $L \geq 1$ such that: (1) $\|h^\dagger\|_\Hcal, \|\Pi_n h^\dagger\|_\Hcal \leq M$; (2) $\|\Pi_n P(h - h^\dagger)\|_\Qcal \leq L \|h - h^\dagger\|_\Hcal$ for all $ \in \Hcal_n$; and (3) for some universal constants $c_1$, $c_2$, and $c_3$, we have
    \begin{align*}
        \gamma_n^q &\geq c_2 \Big(r_n + \sqrt{\log(c_1/\zeta)/n} \Big)^2 \\
        \text{and} \quad \gamma_n^h &\geq c_3 L^2 \Big(\gamma_n^q + \Big(r_n + \sqrt{\log(c_1/\zeta)/n} \Big)^2 \Big) \,.
    \end{align*}
\end{assumption*}

Under the above assumptions, we can provide the finite-sample bound for the estimation error of the minimax estimator $\hat h_n$. The bound is derived from a novel analysis of the minimax estimation problem in \Cref{eq:h-estimator} that differs substantially from the analysis in the seminal work \cite{DikkalaNishanth2020MEoC}.

\begin{theorem}
\label{thm:h-estimator-bound}
    Suppose \Cref{assum:universal-approximation-h,assum:complexity-h,assum:boundedness} hold, as well as either \Cref{assum:bounded-complexity-h} or \Cref{assum:regularization-h}. Then, given some universal constant $c_0$, we have that, for $\zeta\in(0,1/3)$, with probability at least $1-3\zeta$,
    \begin{equation*}
          \|P(\hat h_n - h_0)\|_2 \leq c_0 \Big( M r_n + M \sqrt{\log(c_1/\zeta)/n} + \delta_n + M (\gamma_n^q)^{1/2} + M (\gamma_n^h)^{1/2} +  \mu_n^{1/2} \Big) \,,
    \end{equation*}
    for any $h_0 \in \Hcal_0$, where $c_1$ is the same universal constant as in \Cref{assum:regularization-h}.

    Furthermore, suppose that either of the above sets of assumptions hold, and in addition that: (1) $\mu_n = o(1)$; (2) $\mu_n = \omega(\max(r_n^2,\delta_n^2,\gamma_n^q,\gamma_n^h,1/n))$; (3) $\{h \in \starcls(\Hcal_n - h^\dagger) : \|h\|_\Hcal \leq 1\}$ has critical radius at most $r_n$; (4) $\{h \in \bar\Hcal : \|h\|_\Hcal \leq U\}$ is compact under $\|\cdot\|_\Hcal$ for every $U < \infty$; and (5) $\|h\|_2 \leq K \|h\|_\Hcal$ for all $h \in \bar\Hcal$ and some constant $K < \infty$. Then, we have
    \begin{equation*}
        \|\hat h_n - h^\dagger\|_2 = o_p(1) \,.
    \end{equation*}
    
\end{theorem}

Note that our second bound in \Cref{thm:h-estimator-bound} allows for $\Hcal_n$ and $\Qcal_n$ to be infinite-complexity classes with no well-defined critical radii. It is sufficient that the classes are well-behaved when restricted to radius $\|h^\dagger\|_\Hcal$. Importantly, the algorithm \emph{does not} require knowledge of $\|h^\dagger\|_\Hcal$, and our bound is automatically adaptive to this value. It is important to keep in mind, though, this second bound requires an additional condition that the regularization hyperparameters $\gamma_n^q$ and $\gamma_n^h$ are sufficiently large, although the required size of these hyperparameters does not depend on the unknown $\|h^\dagger\|_\Hcal$.
Conversely, under our first set of assumptions, where $\Hcal_n$ and $\Qcal_n$ have finite total complexity, there is no such restriction, and we are free to set $\gamma_n^q = \gamma_n^h = 0$

Note also that our result being adaptive to the norm of $h^\dagger$ is very similar to the corresponding result in \citet{DikkalaNishanth2020MEoC}. However, our result also ensures strong norm consistency of our estimate $\hat h_n$.

\begin{remark}[Clever Instrument Approach]\label{rem: clever instrument}
We note that within our minimax estimation framework for $h$, we can incorporate such a TMLE constraint, within the estimation of $h$, in a manner similar to the clever covariate adjustment of \cite{scharfstein1999adjusting}. In particular, as long as the test function that we will use when training $h$ is of the form $q + \epsilon\, \hat{q}^\dagger$, where $\hat{q}^\dagger$ is an estimate of $q^\dagger$ (as we describe in the next section), and we do not penalize $\epsilon$, then note that the first order condition for $\epsilon$, implies that at any saddle of the min-max problem the crucial moment is zero, i.e. $\E[q^\dagger(T)(g_2(W) - g_1(W) h(S))]=0$. Thus the plug-in estimate will be doubly robust. The function $q^\dagger(T)$ can be seen as a ``clever instrument'' analogous to how the Riesz representer $a(S)$ is used as a ``clever covariate'' when adjusting for observed confounding.
\end{remark}

\subsection{Estimation of the Debiasing Nuisance Function}\label{sec: debias-nuisance-est}

Next, we consider the estimation of the debiasing nuisance function $q^\dagger$. Here, we will consider estimators of the form
\begin{align}
    \hat q_n &=
    \argmin_{q \in \widetilde\Qcal_n} \EE_n \Big[ 
    \prns{g_1(W) \hat\xi_n(S)-q(T)}^2 
    \Big] + \tilde\gamma_n^q \|q\|_{\widetilde\Qcal}^2 \nonumber \\
    &=     
    \argmax_{q\in\widetilde\Qcal_n} \EE_n \Big[ g_1(W) \hat\xi_n(S) q(T) - \frac{1}{2} q(T)^2 \Big] - \tilde\gamma_n^q \|q\|_{\widetilde\Qcal}^2
     \,, \label{eq:q-estimator} \\ 
    \text{where}\quad\hat \xi_n &= \argmin_{\xi \in \Xi_n} \sup_{q \in \Qcal_n} \EE_n \Big[ g_1(W) \xi(S) q(T) - \frac{1}{2} q(T)^2 - m(W;\xi)  \Big] - \gamma_n^q \|q\|_{\Qcal}^2 + \gamma_n^\xi \|\xi\|_{\Xi}^2 \,. \label{eq:xi-estimator}
\end{align}
Note that in the case that $\widetilde\Qcal_n = \Qcal_n$ we have that $\hat q_n$ is the corresponding interior supremum solution in the minimax estimation of $\hat\xi_n$, so they could be solved for together. However, we allow for the possibility of separate classes for practical empirical reasons. For example, we may wish to use a kernel estimator where the inner maximization is performed analytically for $\hat\xi_n$, then use a different class based on, \eg, neural nets or random forests for the corresponding $\hat q_n$ estimate. Similar to the previous section, we use normed function classes $\Qcal_n \subseteq \bar\Qcal \subseteq \Qcal$, $\widetilde\Qcal_n \subseteq \widetilde\Qcal \subseteq \Qcal$, and $\Xi_n \subseteq \bar\Xi \subseteq \Hcal$, for some fixed normed function sets $\bar\Qcal$, $\widetilde\Qcal$, and $\bar\Xi$ with norms $\|\cdot\|_\Qcal$, $\|\cdot\|_{\widetilde\Qcal}$, and $\|\cdot\|_\Xi$ respectively, and we allow for regularization using these norms, with coefficients $\gamma_n^q$, $\tilde\gamma_n^q$, and $\gamma_n^\xi$.

We do not need to worry about uniqueness of the estimation for estimating $q^\dagger$, since $q^\dagger = P \xi_0$ is unique for any $\xi_0 \in \Xi_0$ according to \Cref{lemma: minimum-norm}
and 
we will be able to obtain rates for the estimation of $q^\dagger$ under $L_2$ norm. However, in order to provide a finite-sample estimation result, we require analogues of \cref{assum:bounded-estimation-h,assum:universal-approximation-h,assum:complexity-h}, as follows. In particular, we require these assumptions to hold for some arbitrary fixed $\xi^\dagger \in \Xi_0$.

\begin{assumption}
\label{assum:bounded-estimation-q}

We have that: (1) $\|\xi\|_\infty \leq 1$ for every $\xi \in \Xi_n$; (2) $\|q\|_\infty \leq 1$ for every $q \in \Qcal_n$; (3) $\|q\|_\infty \leq 1$ for every $q \in \widetilde\Qcal_n$; (4) $\|q^\dagger\|_\infty \leq 1$; and (5) $\|\xi^\dagger\|_\infty \leq 1$ 
    
\end{assumption}

\begin{assumption}
\label{assum:universal-approximation-q}
    There exists some $\delta_n < \infty$ such that: (1) there exists some $\Pi_n \xi^\dagger \in \Xi_n$ such that $\|\Pi_n \xi^\dagger - \xi^\dagger\|_2 \leq \delta_n$; (2) for every $q \in \{P \xi : \xi \in \Xi_n\}$ there exists $\Pi_n q \in \Qcal_n$ such that $\|q - \Pi_n q\|_2 \leq \delta_n$; and (3) there exists $\Pi_n q^\dagger \in \widetilde \Qcal_n$ such that $\|q^\dagger - \Pi_n q^\dagger\|_2 \leq \delta_n$.
\end{assumption}

\begin{assumption}
\label{assum:complexity-q}
    There exists some $r_n$ that bounds the critical radii of the star-shaped closures of the function classes: (1) $\{g_1(W)\xi(S) q(T) : \xi \in \starcls(\Xi_n - \xi^\dagger), q \in \starcls(\widetilde\Qcal_n - q^\dagger), \|\xi\|_\Xi \leq 1, \|q\|_{\widetilde\Qcal} \leq 1\}$; (2) $\{q \in \starcls(\Qcal_n - q^\dagger) : \|q\|_\Qcal \leq 1\}$; (3) $\{q \in \starcls(\widetilde\Qcal_n - q^\dagger) : \|q\|_{\widetilde\Qcal} \leq 1\}$; (4) $\{\xi \in \starcls(\xi_n - \xi^\dagger) : \|\xi\|_{\Xi} \leq 1\}$; and (5) $\{g_1(W)\xi(S) q(T) : \xi \in \starcls(\Xi_n - \xi^\dagger), q \in \starcls(\Qcal_n - q^\dagger), \|\xi\|_\Xi \leq 1, \|q\|_\Qcal \leq 1\}$.
\end{assumption}

We note that the required assumptions on $\Qcal_n$ are significantly stricter than those on $\widetilde\Qcal_n$. In particular, $\widetilde\Qcal_n$ only needs to be able to approximate the single function $q^\dagger$, rather than all functions of the form $\EE[g_1(W) \xi(S) \mid T]$ for $\xi \in \bar\Xi$. We also note that most parts of \Cref{assum:bounded-estimation-q,assum:universal-approximation-q,assum:complexity-q} are identical to \Cref{assum:bounded-estimation-h,assum:universal-approximation-h,assum:complexity-h} when $\Xi_n=\Hcal_n$ and $\bar\Xi=\bar\Hcal$.
We also note that in the case that $\widetilde\Qcal_n = \Qcal_n$ then the first and second function classes in \cref{assum:complexity-q} are identical, as are the fourth and fifth function classes. 

In addition, for this estimation we further impose a boundedness assumption on $m$. 
\begin{assumption}
\label{assum:boundedness-q}
    We have that $\|m(W;h)\|_2 \leq \|h\|_2$ for all $h \in \Hcal$.
\end{assumption}

Finally, as in the previous section, we require either a condition on the maximum functional complexity of the classes $\Qcal_n$, $\widetilde\Qcal_n$, and $\Xi_n$, or a condition that the regularization coefficients are set sufficiently large.

\setcounter{subassumption}{0}
\begin{assumption*}
\label{assum:bounded-complexity-q}
    There exists some constant $M \geq 1$ such that: (1) $\|\xi\|_\Xi, \|\xi-\xi^\dagger\|_\Xi \leq M$ for all $\xi \in \Xi_n$; (2) $\|q\|_\Qcal, \|q-q^\dagger\|_\Qcal \leq M$ for all $q \in \Qcal_n$; (3) $\|q\|_{\widetilde\Qcal}, \|q-q^\dagger\|_{\widetilde\Qcal} \leq M$ for all $q \in \widetilde\Qcal_n$; (4) $\|\xi^\dagger\|_\Xi, \|\Pi_n \xi^\dagger\|_\Xi \leq M$; (5) $\|q^\dagger\|_\Qcal, \|\Pi_n q^\dagger\|_\Qcal \leq M$; and (6) $\|q^\dagger\|_{\widetilde\Qcal,} \|\Pi_n q^\dagger\|_{\widetilde\Qcal} \leq M$.
\end{assumption*}

\begin{assumption*}
\label{assum:regularization-q}
    There exists some constants $M \geq 1$ and $L \geq 1$ such that: (1) $\|\xi^\dagger\|_\Xi, \|\Pi_n \xi^\dagger\|_\Xi \leq M$; (2) $\|q^\dagger\|_\Qcal, \|\Pi_n q^\dagger\|_\Qcal \leq M$; (3) $\|\Pi_n P \xi\|_\Qcal \leq L \|\xi\|_\Xi$ for all $\xi \in \Xi_n$; and (4) for some universal constants $c_1$, $c_2$, $c_3$, and $c_4$ we have
    \begin{align*}
        \gamma_n^q &\geq c_2 \Big(r_n + \sqrt{\log(c_1/\zeta)/n} \Big)^2 \\
        \tilde\gamma_n^q &\geq c_3 \Big(r_n + \sqrt{\log(c_1/\zeta)/n} \Big)^2 \\
        \text{and} \quad \gamma_n^h &\geq c_4 L^2 \Big(\gamma_n^q + \Big(r_n + \sqrt{\log(c_1/\zeta)/n} \Big)^2 \Big) \,.
    \end{align*}
\end{assumption*}

Under the above assumptions, we can provide the following finite-sample bound. 
\begin{theorem}
\label{thm:q-estimator-bound}
    Suppose \Cref{assum:bounded-estimation-q,assum:complexity-q,assum:boundedness,assum:boundedness-q} hold, as well as either \Cref{assum:bounded-complexity-q} or \Cref{assum:regularization-q}. Then, given some universal constant $c_0$, we have that, for $\zeta\in(0,1/7)$, with probability at least $1-7\zeta$,
    \begin{align*}
          \|\hat q_n - q^\dagger\|_2 &\leq c_0 \Big( r_n^{1/2} + \prns{\log(c_1/\zeta)/n}^{1/4} + M r_n + M \sqrt{\log(c_1/\zeta)/n} \\
          &\qquad + \delta_n + M (\gamma_n^q)^{1/2} + M (\gamma_n^h)^{1/2} + M (\tilde \gamma_n^q)^{1/2}  \Big) \,,
    \end{align*}
    where $c_1$ is the same universal constant as in \Cref{assum:regularization-q}.
\end{theorem}

Compared with \cref{thm:h-estimator-bound}, this bound only introduces additional slow rate terms the order of $\sqrt{r_n}$ and $n^{-1/4}$. However, these terms are independent of the unknown function complexity $M$, which only impacts the corresponding fast rate terms of order $r_n$ and $n^{-1/2}$, as well as the regularization coefficient terms.
In addition, the rate here is in terms of the strong $L_2$ norm, rather than a weak projected norm. 
Note as well that this result is based on our novel analysis of the estimation problems in \Cref{eq:xi-estimator,eq:q-estimator}. 
They are different from the canonical forms of minimax problems in \cite{DikkalaNishanth2020MEoC}.

\section{Debiased Inference on the Linear Functional}\label{sec: est-functional}

Given the finite sample bounds from the previous section, and the discussion in \Cref{sec: challenge}, we can now present our main results on the estimation and inference of $\theta^\star$. First we define the $K$-fold cross-fitting estimator, for some fixed $K$ that does not depend on $n$, as follows.

\begin{definition}[Debiased Machine Learning Estimator]\label{def: theta-est}
Fix an integer $K \ge 2$. 
\begin{enumerate}
\item Randomly split the $n$ observations into $K$ (approximately) even folds, whose index sets are denoted by $\Ical_1, \dots, \Ical_K$, respectively. 
\item For $k = 1, \dots, K$, use all data except that in $\Ical_k$ to construct nuisance estimators $\hat h^{(k)}$ and $\hat q^{(k)}$  as described in \Cref{eq:h-estimator,eq:q-estimator}, respectively. 
\item Construct the final debiased machine learning estimator:
\begin{align*}
\hat\theta_n = \frac{1}{K}\sum_{k=1}^K \frac{1}{\abs{\Ical_k}}\sum_{i \in \Ical_k}\psi\prns{W_i; \hat h^{(k)}, \hat q^{(k)}}, ~~ \psi(W; h, q) = {m(W; h) + q(T)(g_2(W) -  g_1(W)h(S))}.
\end{align*}
\end{enumerate}
\end{definition}

Then, we can obtain the following result.

\begin{theorem}
\label{thm:dml-asymp}
    Let the estimator $\hat\theta_n$ be defined as in \Cref{def: theta-est}, and suppose the full conditions of \Cref{thm:h-estimator-bound,thm:q-estimator-bound} hold. Then, as long as $r_n = o(n^{-1/3}), \delta_n = o(n^{-1/4}), \delta_nr_n^{1/2} = o(n^{-1/2}), \mu_n r_n = o(n^{-1})$, and $\mu_n \delta_n^2 = o(n^{-1})$, we have $\|P(\hat h_n - h_0)\|_2\|\hat q_n - q^\dagger\|_2 = o_p(n^{-1/2})$,  and  that as $n\to\infty$,
        \begin{align*}
\sqrt{n}\prns{\hat\theta_n - \theta^\star} = \frac{1}{\sqrt{n}}\sum_{i=1}^n \prns{\psi(W_i;  h^\dagger,  q^\dagger) -\theta^\star} + o_p(1) \rightsquigarrow \mathcal{N}\prns{0, \sigma_0^2} \,,
\end{align*}
    where $\Ncal(0,\sigma_0^2)$ denotes a Gaussian distribution with mean $0$ and variance 
    \begin{equation}\label{eq: asymp-var}
        \sigma_0^2 = \EE\Big[ \Big( \theta^\star - \psi(W;h^\dagger,q^\dagger) \Big)^2 \Big] \,.
    \end{equation}
\end{theorem}

We note that for the full conditions of \Cref{thm:h-estimator-bound} to hold, the penalization hyperparameter $\mu_n$ only needs to converge arbitrarily slower than the critical radii bound $r_n^2$ and the approximation error bound $\delta_n^2$,
\ie, $\mu_n = \omega(\max(r_n^2, \delta_n^2))$. 
At the same time, the condition $\mu_n r_n = o(n^{-1})$ and $\mu_n \delta_n^2 = o(n^{-1})$ requires $\mu_n = o(\min(n^{-1}r_n^{-1}, n^{-1}\delta_n^{-2}))$.
Moreover, when  $r_n = o(n^{-1/3})$, the condition $\delta_nr_n^{1/2} = o(n^{-1/2})$ automatically holds if further $\delta_n = o(n^{-1/3})$.  
The conditions on the critical radii $r_n = o(n^{-1/3})$ and and approximation error $\delta_n = o(n^{-1/3})$ can be justified for many  commonly used machine learning function classes with appropriate structure; see for instance the references on existing results for approximation errors and critical radii cited below \Cref{assum:universal-approximation-h,assum:complexity-h}.

We also note that in the case that $\Hcal_0$ and $\Qcal_0$ are singletons, then \Cref{eq: asymp-var} is known to be the semiparametrically efficient w.r.t. the nonparametric model for $q_0$ and $h_0$; 
for example, \citet{cui2020semiparametric,kallus2021causal} derive this for the problem of proximal causal inference, while the general result easily follows from  \citet{ai2012semiparametric}.
However, when $h_0$ and $q_0$ are non-unique, semiparametric efficiency is unclear. The variance in \Cref{eq: asymp-var} can be shown to arise as the semiparametric efficiency bound when restricting to certain submodels that keep a particular choice of $h_0$ and $q_0$ fixed (see \emph{e.g.} \citet{severini2012efficiency}.) Unfortunately, the meaning of this as a best-achievable variance is unclear as it depends on this choice.

Furthermore, we can estimate the asymptotic variance using the cross-fitting nuisance estimators: 
    \begin{equation}\label{eq: variance-estimator}
        \hat\sigma_n^2 = \frac{1}{K}\sum_{k=1}^K \frac{1}{\abs{\Ical_k}}\sum_{i \in \Ical_k} \Big(\hat\theta_n - \psi(W; \hat h^{(k)}, \hat q^{(k)}) \Big)^2
    \end{equation}
We can further use the  variance estimator above to construct a confidence interval: 
\begin{align}\label{eq: CI}
\op{CI} = \left[\hat\theta_n - \Phi^{-1}\prns{1-\alpha/2}\hat \sigma_n^2,~\hat\theta_n + \Phi^{-1}\prns{1-\alpha/2}\hat \sigma_n^2\right],
\end{align}
where $\hat\theta_n$ is the debiased machine learning estimator in \Cref{def: theta-est}, and $\Phi^{-1}\prns{1-\alpha/2}$ is the $1-\alpha/2$ quantile of  the standard normal distribution.

In the following theorem, we show that the variance estimator  and the confidence interval above are asymptotically valid. 

\begin{theorem}
\label{lem:inference}
    Let $\sigma_0^2$ be the asymptotic variance in \Cref{eq: asymp-var},  $\hat\sigma_n^2$ be the variance estimator in \Cref{eq: variance-estimator}, 
    and $\op{CI}$ be the confidence interval in \Cref{eq: CI}. 
    If the conditions of \Cref{thm:dml-asymp} hold, 
    then as $n\to\infty$, 
    $\hat\sigma_n^2$ converges to $\sigma_0^2$  in probability, and $\Prb{\theta^\star \in \op{CI}} \to 1-\alpha$.  
\end{theorem}

\section{Application to  Partially Linear Models}\label{sec: partial-linear}
In \Cref{ex: IV-reg,ex: proximal}, we consider IV estimation and proximal causal inference with general nonparametric nuisance functions. 
In this section, we apply our methods to  partially linear models in these settings. 
 Partially linear model is a semiparametric model that has been widely used in exogenous regressions because it 
  retains both the flexibility of nonparametric models and the ease of interpretation of linear models \citep{hardle2000partially}.
  Our analyses in this section extend the existing literature on partially linear IV models, and also broaden the scope of proximal causal inference literature by studying partially linear models for the first time.  

\subsection{Partially Linear IV Estimation}
\label{sec: partial-linear-iv}
In \Cref{ex: IV-reg}, we considered a general NPIV regression model with $\Hcal = \Lcal_2(X)$, where $X$ are endogenous variables.  
Now we consider $X = \prns{X_a, X_b}\in\R{d_a}\times \R{d_b}$, and 
focus on  a partially linear IV regression model: 
$$\Hcal = \Hcal_{\op{PL}} \coloneqq \braces{\theta^\top X_a + g(X_b): \theta\in\R{d_a}, g\in\Lcal_2(X_b)}.$$
With some slight abuse of notation, we will alternatively refer to any element of $\Hcal_{\op{PL}}$ as $h$ or as the corresponding tuple $(\theta,g)$.
Accordingly, the true IV regression can be denoted as $h^\star(X) = \theta^{*\top} X_a + g^\star(X_b) \in \Hcal_{\op{PL}}$, such that 
\begin{align}\label{eq: partial-iv}
[Ph^\star](Z) = \Eb{\theta^{*\top} X_a + g^\star(X_b) \mid Z} =  \Eb{Y \mid Z}.
\end{align}
As in the partially linear regression model, we are interested in the coefficient parameter $\theta^\star$ and view $g^\star$ as a nonparametric nuisance function.

Partially linear IV models have already been studied by a few previous works \citep{florens2012instrumental,chen2021robust,chernozhukov2018double}.
\cite{chernozhukov2018double,florens2012instrumental} assume that the IV regression $h^\star$ is uniquely identified by the conditional moment equation in \Cref{eq: partial-iv}.
While \cite{chernozhukov2018double} consider a simpler setting where the nuisance $g^\star$ is a function of exogenous random variables (\ie, $X_b$ is part of $Z$), \cite{florens2012instrumental} allow all variables in IV regression (both $X_a, X_b$) to be endogenous and develop a Tikhonov regularized estimator with strong theoretical guarantees under a source condition on the IV regression. 
\cite{chen2021robust} also allows both $X_a, X_b$ to be potentially endogenous, and additionally allows the nonparameteric nuisance $g^\star$ to be underidentified. 
This is also the setting we consider in this part. However, unlike the penalized sieve-based estimators developed in \cite{chen2021robust}, we will employ our penalized minimax estimator that can leverage more flexible function classes.

We first note that the target parameter $\theta^\star$ can be written as a linear functional:
\begin{align}
&\theta^\star = \Eb{\alpha(X)h^\star(X)},\nonumber \\
&\text{where } \alpha(X) = M^{-1}\prns{X_a - \Eb{X_a \mid X_b}},\,M=\Eb{\prns{X_a - \Eb{X_a \mid X_b}}\prns{X_a - \Eb{X_a \mid X_b}}^\top}\label{eq: PL-iv-tilde-alpha}
\end{align}
where we assume the matrix $M$ is invertible.

For $i = 1, \dots, d_a$, and arbitrary $h = (\theta,g) \in \Hcal_{\op{PL}}$, we define $m_i(W;h) = \theta_i$. Then, for each such $i$, we can define linear functional $h \mapsto \Eb{m_i(W; h)} = \Eb{\alpha_i(X)h(X)}$ where $\alpha_i(X)$ is the $i$th coordinate of $\alpha$ in \Cref{eq: PL-iv-tilde-alpha}. 
Moreover, since $\alpha_i(X)=\theta_{\alpha,i}^\top X_a+g_{\alpha,i}(X_b)$ with $\theta_{\alpha,i}=[M^{-1}]_{i,:},\,g_{\alpha,i}(X_b)=-\theta_{\alpha,i}^\top\EE[X_a\mid X_b]$, it also belongs to $\Hcal_{\op{PL}}$, so it is the unique Riesz representer for the linear functional $h \mapsto \Eb{m_i(W; h)}$. 

In this setting, the condition proposed by \cite{severini2012efficiency} and described in 
\Cref{eq: Q0}  requires the existence of $q_0 = (q_{0, 1}, \dots, q_{0, d_a})$ such that for $i=1, \dots, d_a$, we have $q_{0, i} \in \Lcal_2(Z)$ and 
\begin{align}\label{eq: partial-iv-q}
 \Eb{\prns{q_{0, i}(Z) -  \alpha_i(X)}h(X)} = 0, ~~ \forall h \in \Hcal. 
 \end{align} 
Any such $q_0$ can be alternatively characterized by the following proposition.  
\begin{proposition}\label{prop: partial-iv-q2}
Let $q_0 = (q_{0, 1}, \dots, q_{0, d_a})^\top$ 
where $q_{0, i} \in \Lcal_2(Z)$ for $i = 1, \dots, d_a$. 
Then $q_{0, i}$ satisfies \Cref{eq: partial-iv-q} for all $i = 1,\dots, d_a$ if and only if 
\begin{align}\label{eq: partial-iv-q2}
\Eb{q_0(Z)\mid X_b} = \mathbf{0}_{d_a}, ~~ \Eb{q_0(Z)X_a^\top} = I_{d_a},
\end{align}
where $\mathbf{0}_{d_a}$ is an all-zero vector of length $d_a$ and $I_{d_a}$ is the $d_a \times d_a$ identity matrix. 
\end{proposition}

According to \Cref{thm: new-nuisance}, our  \Cref{assump: new-nuisance} strengthens the condition in \Cref{eq: partial-iv-q2}, by requiring the existence of $\xi_{0, i}(A, X) = \theta_{0, i}^\top X_a + g_{0, i}(X_b) \in \Hcal$ for $i = 1, \dots, d_a$, such that $q^\dagger = (\Eb{\xi_{0, 1}(X)\mid Z}, \dots, \Eb{\xi_{0, d_a}(X)\mid Z})^\top$ satisfies \Cref{eq: partial-iv-q2}.
Such $\xi_{0, i}$ is  given by 
\begin{align}
\label{eq: PL-iv-delta}
\xi_{0, i} \in \argmin_{\xi_i \in \Hcal}\Eb{\prns{\Eb{\xi_i(X) \mid Z}}^2} - 2\Eb{m_i(W; \xi_i)}.
\end{align}

In order to estimate the nuisance functions, we need to first specify partially linear function classes $\Hcal_n \subseteq \Hcal_{\op{PL}}, \Xi_n \subseteq \Hcal_{\op{PL}}$ as hypothesis classes for IV regression $h^\star$ and a solution $\xi_{0, i}$ to \Cref{eq: PL-iv-delta} respectively, a function class $\Qcal_n\subseteq \Lcal_2(Z)$  used to form the inner maximization problem in minimax objectives, and a function class $\widetilde\Qcal_n \subseteq \Lcal_2(Z)$ as the hypothesis class for $P\xi_{0, i}$. 
Note that again we alternatively refer to elements of  $\Hcal_n$ or $\Xi_n$ by the overall functions $h$ or $\xi$, or as tuples $(\theta,g)$.
Then we can follow \Cref{eq:h-estimator} to construct an estimator for the partially IV regression:  
\begin{align*}
    &\hat h_n(X) = \tilde\theta^\top X_a + \hat g_n(X_b), \\
    \text{where }&(\tilde\theta, \hat g_n)= \argmin_{(\theta, g) \in \Hcal_n} \sup_{q \in \Qcal_n} \EE_n \Big[ \Big(\theta^\top X_a + g(X_b)- Y\Big) q(Z) - \frac{1}{2} q(Z)^2 + \mu_n \Big(\theta^\top X_a + g(X_b)\Big)^2 \Big] \,.
\end{align*}

Although the above already gives a coefficient estimator $\tilde\theta$, the estimator $\tilde\theta$ is generally not $\sqrt{n}$-consistent or asymptotically normal because of the estimation bias of $\hat g_n$, 
 especially when $\hat g_n$ is constructed by black-box machine learning methods. 

 To debias the initial coefficient estimator $\tilde\theta$, we further construct the  nuisance estimator $\hat q_n = (\hat q_{1, n}, \dots, \hat q_{d_a, n})^\top$:
 \begin{equation*}
    \hat q_{i, n} = \argmax_{q \in \widetilde\Qcal_n} \EE_n \Big[\bar\xi_i(X) q(Z) - \frac{1}{2} q(Z)^2 \Big] \,,
\end{equation*}
where $\bar\xi_i(X) = {\bar\theta_i^\top X_a +  \bar g_i(X_b)}$ with 
\begin{equation*}
    (\bar\theta_i, \bar g_i) = \argmin_{(\theta, g) \in \Xi_n} \sup_{q \in \Qcal_n} \EE_n \Big[\prns{\theta^\top X_a + g(X_b)} q(Z) - \theta_i - \frac{1}{2} q(Z)^2 \Big] \,.
\end{equation*}

Finally, we construct the debiased coefficient estimator as follows:
\begin{align*}
\hat\theta_n = \tilde\theta + \E_n\bracks{\Big(Y - \tilde\theta^\top X_a - \hat g_n(X_b)\Big)\hat q_n(Z)}. 
\end{align*}
Here to obtain the final estimator $\hat\theta_n$,
the initial coefficient estimator $\tilde\theta$ is debiased by the second augmented term that involves the additional debiasing nuisance estimator $\hat q_n$. 
Above we use the same data to construct nuisance estimators and the final coefficient estimator for simplicity. 
But we can easily  incorporate the cross-fitting described in \Cref{def: theta-est}.

\subsubsection{Connection to \cite{chen2021robust}}
\label{remark:chen-connection}
\cite{chen2021robust} also studies the estimation of partially linear IV regression when the nonparametric component is under-identified. 
A key condition in \cite{chen2021robust} 
is actually closely related to
our existence assumption of functions $\xi_0$. 
\cite{chen2021robust} implicitly assumes the existence of $\rho_0(X_b) = \prns{\rho_{0, 1}(X_b), \dots, \rho_{0, d_a}(X_b)}$ such that for $X^{(i)}_a$, the $i$th component of $X_a$, we have 
\begin{align}\label{eq: chen-nuisance-1}
\rho_{0, i} \in \argmin_{\rho\in\Lcal_2(W)} \Eb{\prns{\Eb{X^{(i)}_a - \rho(X_b) \mid Z}}^2}, ~~ i = 1, \dots, d_a \,.
\end{align}
In addition, he explicitly assumes that the resulting $\rho_0$ satisfies that  
\begin{align}
\label{eq: chen-nuisance-2}
\Gamma = \Eb{\Eb{X_a - \rho_0(X_b) \mid Z}\prns{\Eb{X_a - \rho_0(X_b) \mid Z}}^\top} \text{ is invertible}.
\end{align}
In the following proposition, we show that these assumptions are actually sufficient conditions for the existence of $\xi_0 = (\xi_{0, 1}, \dots, \xi_{0, d_a})$ characterized by \Cref{eq: PL-iv-delta}. 
\begin{proposition}\label{prop: chen-nuisance}
Let $\rho_0 = (\rho_{0, 1}, \dots, \rho_{0, d_a})$ and $\Gamma$ be given in \Cref{eq: chen-nuisance-1,eq: chen-nuisance-2} respectively, and define $\tilde\xi_0 = \Gamma^{-1}(X_a-\rho_0(X_b))$.
Then for each $i = 1, \dots, d_a$, the $i$th coordinate of $\tilde\xi_0$, namely $\tilde\xi_{0, i}$,   
  is a solution to \Cref{eq: PL-iv-delta}. 
\end{proposition}
\Cref{prop: chen-nuisance} provides another perspective to understand the conditions assumed in \cite{chen2021robust}. It shows that \cite{chen2021robust} also implicitly assume our \Cref{assump: new-nuisance} for the partially linear IV model.
\cite{chen2021robust} does not estimate the functions $\rho_{0, i}$ or $\Gamma$ defined in \Cref{eq: chen-nuisance-1,eq: chen-nuisance-2}, instead they use them to analyze the asymptotic distribution of his penalized linear sieve estimator. 
In contrast, we use \Cref{eq: PL-iv-delta} to estimate the nuisance $\xi_0$ in \Cref{eq: PL-iv-delta}, so that we can leverage the doubly robust identification formula in \Cref{thm:dr-ident}. 
This allows us to employ general minimax nuisance estimators beyond linear sieve estimation. 
We also note that \Cref{prop: chen-nuisance} implies that we could consider a different method for estimating $q^\dagger$, based on first directly estimating $\tilde\xi_0$ by solving the minimum projected distance problem of \Cref{eq: chen-nuisance-1}, and then estimating the projection of this function via least squares. However, we do not theoretically analyze this approach.

\subsection{Partially Linear Proximal Causal Inference}\label{sec: partial-linear-proximal}
In \Cref{ex: proximal}, we introduced proximal causal inference with a binary treatment and a general nonparametric class of bridge functions $\Hcal = \Lcal_2(V, X, A)$. 
When the treatment is more complex (\eg, continuous), we may be interested in restricting  bridge functions to some structured classes. 
In this part, we consider the following class of partially linear bridge functions:
$$\Hcal = \tilde\Hcal_{\op{PL}} \coloneqq \braces{\theta^\top A + g(V, X): \theta \in \R{d_A}, g \in \Lcal_2(V, X)}\,.$$
In addition, let us denote the total set of outcome bridge functions as
\begin{equation*}
    \tilde\Hcal_{\op{OB}} \coloneqq \Big\{ h \in L_2(V,X,A) : \EE[Y - h(V,X,Z) \mid U,X,A] = 0 \Big\} \,.
\end{equation*}

To the best of our knowledge, this partially linear model has not been applied to proximal causal inference yet. Previous literature on proximal causal inference  focus on either parametric estimation or nonparametric estimation of the bridge functions \citep[\eg, ][]{cui2020semiparametric,miao2018a,kallus2021causal,singh2020kernel,GhassamiAmirEmad2021MKML,mastouri2021proximal}. 
In the following proposition, we give a justification for the partially linear bridge function model.
\begin{proposition}\label{prop: partial-linear-proximal}
Suppose that $\Eb{Y(a) \mid U, X} = \theta^{\star\top} a + \phi^\star(U,X)$, for some vector $\theta^\star$ and function $\phi^\star\in\Lcal_2(U, X)$, and that $\tilde\Hcal_{\op{OB}}$ is non-empty. Then, we have that $\tilde\Hcal_{\op{PL}} \cap \tilde\Hcal_{\op{OB}}$ is non-empty; that is, there exists a partially-linear outcome bridge function.

Furthermore, suppose in addition that $\Gamma \coloneqq \EE[(A - \EE[A \mid V,X])(A - \EE[A \mid V,X])^\top]$ is invertible, and that for each $i \in [d_A]$ there exists $q_{0,i} \in L_2(Z,X,A)$ such that $\EE[q_{0,i}(Z,X,A) (\theta^\top A + g(V,X)) ] = \theta_i$ for all $(\theta,g) \in \tilde\Hcal_{\op{PL}}$.
Then, for any partially linear bridge function $\theta^\top A + g(V,X) \in \tilde\Hcal_{\op{OB}}$, we have $\theta=\theta^\star$; that is, the partially-linear coefficients are unique.
\end{proposition}
In \Cref{prop: partial-linear-proximal}, we show that if the conditional expectation of potential outcome $Y(a)$ given unobserved confounders $U$ and covariates $X$ is partially linear in the treatment $a$, then there exists a partially linear bridge function. 
In particular, given the additional conditions in the second part of the proposition, the linear coefficients $\theta^\star$ of any such bridge function characterizes the treatment effects.

Now, given the conditions of \Cref{prop: partial-linear-proximal}, this implies that estimating $\theta^\star$ is a special case of the partially-linear IV problem considered in \Cref{sec: partial-linear-iv}, with the variables $X_a$, $X_b$, and $Z$ there corresponding to $A$, $(V,X)$, and $(Z,X,A)$  here respectively. In particular, all of the results from that section immediately follow here, given with these variable substitutions. That is, we can again estimate $\theta^\star$ using the same de-biased estimator, with $q^\dagger$ estimated following either the estimator $(\hat q_{1,n},\ldots,\hat q_{d_a,n})$ proposed in that section, or the alternative approach based on \Cref{prop: chen-nuisance} discussed in \Cref{remark:chen-connection}.

Furthermore, by the definition of the Riesz representer, it is clear that the functions $q_{0,i}$ in \Cref{prop: partial-linear-proximal} satisfy
\begin{equation*}
    \Pi_{\tilde\Hcal_{\op{PL}}}[q_{0,i}(Z,X,A) \mid V,X,A] = \alpha_i(V,X,A) \quad \forall i \in [d_A] \,,
\end{equation*}
where $\alpha_i$ is the Riesz representer for the partially linear IV functional $m(W;(\theta,g)) = \theta_i$ as in \Cref{sec: partial-linear-iv}. That is, the extra condition in \Cref{prop: partial-linear-proximal} used to justify the uniqueness of $\theta^\star$ in partially linear bridge functions is equivalent to the  condition $\alpha \in \Rcal(P^\star)$, 
which as discussed in \Cref{sec: interpretation} is guaranteed by \Cref{assump: new-nuisance}. Alternatively, the existence of $q_0 = (q_{0,1},\ldots,q_{0,d_A})^\top$ could be interpreted as the existence of a treatment-style bridge function for the marginal treatment effect of varying the vector-valued $A$.
\footnote{Technically it is not exactly the same as treatment bridge functions in the standard proximal causal inference literature as in, \emph{e.g.}, \citet{cui2020semiparametric}. There, where $A \in \{0,1\}$, the treatment bridge functions (up to multiplication by $(-1)^{1-a}$) are defined according to $\EE[q_0(Z,X,a) \mid V,X,A=a] = (-1)^{1-a} \PP(A=a \mid V,X)^{-1}$, where 
$(-1)^{1-a} \PP(A=a \mid V,X)^{-1}$ is the Riesz representer of the linear functional $h \mapsto \Eb{h(V, X, 1) - h(V, X, 0)}$. 
This  is equivalent to requiring that $\EE[q_0(Z,X,A) h(V,X,A)] = \EE[h(V,X,1) - h(V,X,0)]$ for all $h \in L_2(V,X,A)$. Instead, here we require the analogous condition that $q_0$ satisfies $\EE[q_0(Z,X,A)(\theta^\top A + g(V,X))] = \theta$ for all $(\theta,g) \in \tilde\Hcal_{\op{PL}}$.}

\section{Simulation Study}
\label{sec:experiments}

We consider a simple simulation study to examine the finite sample performance of our estimation and inference approach. The data are generated by the following simple data generating process:
\begin{align}
T\sim~& N(0, \sigma_T=2) \tag{instrument}\\
U\sim~& N(0, \sigma_S=2) \tag{unobserved confounder}\\
S =~& \rho\, T + (1-\rho)\, U + \zeta, ~~ \zeta\sim N(0, \sigma=.1)\tag{endogenous treatment}\\
y =~& h_0(S) + U + \nu, ~~ \nu\sim N(0, \sigma=.1) \tag{outcome}
\end{align}
The function $h_0(\cdot)$ is some non-linear function among a pre-specified set of non-linearities, covering a range of behaviors. The parameter $\rho$ controls the instrument strength. Our goal is to estimate the functional that corresponds to the average finite difference:
\begin{align}
m(W;h) := \frac{h(S+\epsilon) - h(S-\epsilon)}{2\epsilon}
\end{align}
with $\epsilon=0.1$; which is an arithmetic approximation to the average derivative.

Note that in this setting, if we let $f(S)$ denote the density function of $S$, then the Riesz representer for the average derivative is of the form:
\begin{align}
a_0(S) = -(\log(f(S)))'
\end{align}
Since $S$ is the sum of three independent normal mean zero r.v.'s $S$ is also a normal r.v. with mean $0$ and variance $\sigma_S^2$. Thus the Riesz representer takes the simple linear form:
\begin{align}
a_0(S) = - (-S^2/(2\sigma_S^2))' = S / \sigma_S^2
\end{align}
Moreover, $q_0(T)$ is the outcome of an IV regression with outcome $a_0(S)$, treatment $T$ and instrument $S$, i.e. it has to satisfy the solution:
\begin{align}
\E[S/\sigma_S^2 - q(T)\mid S] = 0
\end{align}
We can always take a linear such solution:
\begin{align}
\E[S/\sigma_S^2 - \gamma T \mid S] =0
\end{align}
Note that $\E[T\mid S] = \rho \frac{\sigma_T^2}{\sigma_S^2} S$. Thus setting $\gamma = \frac{\sigma_S^2}{\rho\sigma_T^2 \sigma_S^2} = \frac{1}{\rho \sigma_T^2}$, we get that $q_0(T) = \gamma T = \frac{1}{\rho \sigma_T^2} T$, satisfies the set of conditional moment restrictions. Moreover, $\xi$ has to satisfy the conditional moment restrictions:
\begin{align}
\E[q_0(T) - \xi(S) \mid T] = \E[T/(\rho \sigma_T^2) - \xi(S) \mid T] = 0
\end{align}
Since $\E[S\mid T] = \rho T$, we find that setting $\xi(S) = \delta S$ with $\delta = 1/(\rho^2 \sigma_T^2)$, satisfies the above equations. Moreover, we can see that this is also the solution to our optimization problem over $\xi$, which is given by objective function
\begin{align}
V :=~& \frac{1}{2}\E[\E[\xi(S)\mid T]^2] - \E[m(W;\xi)]
= \frac{1}{2}\E[\E[\xi(S)\mid T]^2] - \E[S \xi(S) / \sigma_S^2]
\end{align}
For a linear $\xi$ specification it takes the form:
\begin{align}
V :=~& \E\left[\frac{1}{2} \delta^2 \E[S|T]^2 - \delta\, S^2/\sigma_S^2\right] =
\E\left[\frac{1}{2} \delta^2 \rho^2 T^2 - \delta\, S^2/\sigma_S^2\right]
= \frac{1}{2} \delta^2 \rho^2  \sigma_T^2 - \delta
\end{align}
which yields a first order optimal solution $\delta = 1/(\rho^2 \sigma_T^2)$ with value $-\frac{1}{2} \frac{1}{\rho^2 \sigma_T^2}$. Note that in this setting \Cref{lem:rosenbaum} states that for consistency it suffices to find a function $h$ that satisfies the moment condition $\E[T (y - h(S))]=0$. If we restrict $h(S)$ to be linear, i.e. $h(S)=\theta'S$, then the solution $\theta$ to the above moment condition is exactly the solution of ordinary two stage least squares (2SLS). Thus in this specific data generating process we can recover the average derivative by simply running 2SLS; this property is also remarked in \cite[Theorem B.2]{chernozhukov2016locally}.

We implemented the minimax estimation procedure for both the IV function $\hat{h}$ and the nuisance $\hat{\xi}$, using adversarial neural network training; using simultaneous stochastic gradient descend-ascend with the Optimistic Adam algorithm \citep{daskalakis2017training} and batch size of $100$ samples. Subsequently, we estimate $\hat{q}$ by regressing $\hat{\xi}(S)$ on $T$, using automated model selection via cross validation among many random and boosted forest based models and regularized linear models with polynomial feature expansions. The penalty parameter $\mu_n$ was chosen as $.1\cdot n^{-.9}$, so as to satisfy the required bound on the theoretical specification, viewing neural networks as a VC class. The neural networks were trained with early stopping, using an approximation of the inner max loss as described in \cite{bennett2019deep}: a set of test functions for the inner max were chosen by performing adversarial training for $100$ epochs. The test function at the end of each epoch was kept in a collection. Then training was re-initiated and at the end of each epoch the maximum loss over the set of pre-defined test functions was calculated on a held-out sample. The model that achieved the smallest out-of-sample approximate maximum loss was chosen. The learning rate for both the min and the max optimizer were chosen as $1e-4$ and the $\ell_2$-penalty on the weights as $1e-3$. The neural network architecture for both the min and max neural network was a fully connected architecture with one hidden layer of $100$ neurons, leaky RELU activation and dropout with probability $.1$.

We report coverage of 95\% confidence intervals (cov), root-mean-squared-error (rmse) and bias of four methods; the doubly robust method (dr), the tmle variant of the doubly robust method (as described in Section~\ref{sec: identification}), the method that uses only the nuisance $\hat{q}$, i.e. $\hat{\theta}=\E_n[\hat{q}(T) Y]$ (ipw) and the direct method, which plugs the neural network estimate of $\hat{h}$ in the moment, i.e $\hat{\theta}=\E_n[m(W;\hat{h})]$. Coverage is reported only for the dr and tmle methods which offer simple plug-in standard error calculation and construction of confidence intervals with asymptotically nominal coverage. We used simple sample-splitting: the functions $\hat{h}, \hat{\xi}, \hat{q}$ are estimated on a random half of the data and the average moment and standard error are estimated on the other half. The results for a varying sample size $n$, instrument strength $\rho$ and non-linearity $h_0$ are presented in Figure~\ref{fig:experiments}. In Figure~\ref{fig:experiments-clever}, we report analogous results when the idea of the "clever instrument" from \Cref{rem: clever instrument} is incorporated during the training process of $\hat{h}$. We see that in this case the performance of the direct estimate in terms of rmse and bias is comparable to the dr and tmle estimates. Moreover, overall performance slightly improves, as compared to the variant without the clever instrument, especially in small samples. Overall we find that coverage is approximately nominal across all sample sizes and instrument strengths. The only exception is the 2dpoly function when $n=500$ and $\rho=0.2$. Hence, we see that even with very small samples and with very weak instruments, our method provides valid confidence intervals for the average finite difference. Finally, we see that when we add the clever instrument, as expected, the plug-in direct estimator has almost identical rmse performance with the doubly robust and tmle estimates, showcasing that the approach of incorporating "debiasing" during the training process of $h$ worked as expected. The only exception seems to be when instruments are rather weak ($\rho=.2$) where the $q$ function takes too large of values and makes adversarial training with an un-penalized coefficient in front of $q$ a bit un-stable to train when learning the function $h$. Potentially, further heuristic practical improvements on the training process could fix the issue.
Finally, in Figure~\ref{fig:experiments-very-weak} we examine the limits of instrument weakness under which our method provides correct coverage and accurate estimates.

\begin{figure}[H]
\scriptsize
\begin{tabular}{lll||lll|lll|lll|lll|}
\toprule
    &          &            & \multicolumn{3}{c|}{dr} & \multicolumn{3}{c|}{tmle} & \multicolumn{3}{c|}{ipw} & \multicolumn{3}{c|}{direct} \\
    &          &            & cov &   rmse &   bias &  cov &   rmse &   bias & cov &   rmse &   bias &    cov &   rmse &   bias \\
\midrule
\multirow{9}{*}{\textbf{abs}} & \multirow{3}{*}{\textbf{$n=500$}} & \textbf{$\rho=0.2$} &  92 &  0.129 &  0.008 &   79 &  0.266 &  0.057 &  NA &  0.102 &  0.011 &     NA &  0.122 &  0.020 \\
    &          & \textbf{$\rho=0.5$} &  94 &  0.066 &  0.018 &   94 &  0.069 &  0.018 &  NA &  0.063 &  0.014 &     NA &  0.083 &  0.009 \\
    &          & \textbf{$\rho=0.7$} &  95 &  0.051 &  0.012 &   94 &  0.051 &  0.011 &  NA &  0.049 &  0.011 &     NA &  0.055 &  0.002 \\
\cline{2-15}
    & \multirow{3}{*}{\textbf{$n=1000$}} & \textbf{$\rho=0.2$} &  94 &  0.095 &  0.006 &   87 &  0.137 &  0.022 &  NA &  0.089 &  0.007 &     NA &  0.082 &  0.016 \\
    &          & \textbf{$\rho=0.5$} &  98 &  0.042 &  0.009 &   96 &  0.044 &  0.008 &  NA &  0.046 &  0.007 &     NA &  0.042 &  0.002 \\
    &          & \textbf{$\rho=0.7$} &  97 &  0.032 &  0.006 &   97 &  0.031 &  0.005 &  NA &  0.037 &  0.005 &     NA &  0.032 &  0.001 \\
\cline{2-15}
    & \multirow{3}{*}{\textbf{$n=2000$}} & \textbf{$\rho=0.2$} &  90 &  0.077 &  0.000 &   82 &  0.089 &  0.000 &  NA &  0.076 &  0.002 &     NA &  0.072 &  0.015 \\
    &          & \textbf{$\rho=0.5$} &  94 &  0.032 &  0.001 &   92 &  0.033 &  0.000 &  NA &  0.038 &  0.001 &     NA &  0.027 &  0.004 \\
    &          & \textbf{$\rho=0.7$} &  95 &  0.026 &  0.000 &   94 &  0.026 &  0.000 &  NA &  0.031 &  0.003 &     NA &  0.021 &  0.003 \\
\cline{1-15}
\cline{2-15}
\multirow{9}{*}{\textbf{2dpoly}} & \multirow{3}{*}{\textbf{$n=500$}} & \textbf{$\rho=0.2$} &  73 &  0.191 &  0.046 &   77 &  0.240 &  0.021 &  NA &  0.206 &  0.155 &     NA &  0.222 &  0.124 \\
    &          & \textbf{$\rho=0.5$} &  94 &  0.090 &  0.013 &   94 &  0.092 &  0.014 &  NA &  0.097 &  0.009 &     NA &  0.113 &  0.015 \\
    &          & \textbf{$\rho=0.7$} &  92 &  0.067 &  0.008 &   92 &  0.066 &  0.008 &  NA &  0.090 &  0.004 &     NA &  0.098 &  0.015 \\
\cline{2-15}
    & \multirow{3}{*}{\textbf{$n=1000$}} & \textbf{$\rho=0.2$} &  85 &  0.127 &  0.019 &   82 &  0.150 &  0.006 &  NA &  0.151 &  0.092 &     NA &  0.145 &  0.078 \\
    &          & \textbf{$\rho=0.5$} &  95 &  0.054 &  0.009 &   96 &  0.054 &  0.009 &  NA &  0.073 &  0.010 &     NA &  0.054 &  0.004 \\
    &          & \textbf{$\rho=0.7$} &  96 &  0.039 &  0.008 &   96 &  0.039 &  0.007 &  NA &  0.056 &  0.003 &     NA &  0.041 &  0.007 \\
\cline{2-15}
    & \multirow{3}{*}{\textbf{$n=2000$}} & \textbf{$\rho=0.2$} &  92 &  0.092 &  0.010 &   87 &  0.099 &  0.005 &  NA &  0.119 &  0.058 &     NA &  0.094 &  0.047 \\
    &          & \textbf{$\rho=0.5$} &  90 &  0.044 &  0.002 &   90 &  0.044 &  0.003 &  NA &  0.069 &  0.012 &     NA &  0.037 &  0.000 \\
    &          & \textbf{$\rho=0.7$} &  94 &  0.031 &  0.000 &   94 &  0.030 &  0.000 &  NA &  0.053 &  0.007 &     NA &  0.028 &  0.002 \\
\cline{1-15}
\cline{2-15}
\multirow{9}{*}{\textbf{sigmoid}} & \multirow{3}{*}{\textbf{$n=500$}} & \textbf{$\rho=0.2$} &  92 &  0.107 &  0.029 &   81 &  0.226 &  0.053 &  NA &  0.115 &  0.069 &     NA &  0.117 &  0.034 \\
    &          & \textbf{$\rho=0.5$} &  95 &  0.047 &  0.014 &   93 &  0.047 &  0.013 &  NA &  0.047 &  0.020 &     NA &  0.064 &  0.017 \\
    &          & \textbf{$\rho=0.7$} &  95 &  0.035 &  0.010 &   94 &  0.035 &  0.009 &  NA &  0.033 &  0.010 &     NA &  0.044 &  0.008 \\
\cline{2-15}
    & \multirow{3}{*}{\textbf{$n=1000$}} & \textbf{$\rho=0.2$} &  94 &  0.077 &  0.014 &   90 &  0.104 &  0.021 &  NA &  0.088 &  0.043 &     NA &  0.070 &  0.007 \\
    &          & \textbf{$\rho=0.5$} &  97 &  0.034 &  0.004 &   96 &  0.034 &  0.004 &  NA &  0.033 &  0.010 &     NA &  0.035 &  0.002 \\
    &          & \textbf{$\rho=0.7$} &  96 &  0.026 &  0.004 &   97 &  0.026 &  0.003 &  NA &  0.024 &  0.003 &     NA &  0.025 &  0.003 \\
\cline{2-15}
    & \multirow{3}{*}{\textbf{$n=2000$}} & \textbf{$\rho=0.2$} &  92 &  0.057 &  0.006 &   86 &  0.065 &  0.006 &  NA &  0.061 &  0.026 &     NA &  0.056 &  0.005 \\
    &          & \textbf{$\rho=0.5$} &  96 &  0.024 &  0.001 &   93 &  0.024 &  0.000 &  NA &  0.021 &  0.002 &     NA &  0.022 &  0.004 \\
    &          & \textbf{$\rho=0.7$} &  93 &  0.019 &  0.001 &   93 &  0.019 &  0.001 &  NA &  0.019 &  0.004 &     NA &  0.014 &  0.005 \\
\cline{1-15}
\cline{2-15}
\multirow{9}{*}{\textbf{sin}} & \multirow{3}{*}{\textbf{$n=500$}} & \textbf{$\rho=0.2$} &  87 &  0.120 &  0.013 &   80 &  0.201 &  0.035 &  NA &  0.100 &  0.036 &     NA &  0.102 &  0.012 \\
    &          & \textbf{$\rho=0.5$} &  97 &  0.044 &  0.015 &   95 &  0.046 &  0.014 &  NA &  0.043 &  0.018 &     NA &  0.052 &  0.015 \\
    &          & \textbf{$\rho=0.7$} &  98 &  0.036 &  0.010 &   98 &  0.036 &  0.010 &  NA &  0.033 &  0.010 &     NA &  0.045 &  0.000 \\
\cline{2-15}
    & \multirow{3}{*}{\textbf{$n=1000$}} & \textbf{$\rho=0.2$} &  94 &  0.087 &  0.014 &   89 &  0.117 &  0.022 &  NA &  0.086 &  0.034 &     NA &  0.078 &  0.001 \\
    &          & \textbf{$\rho=0.5$} &  95 &  0.035 &  0.004 &   95 &  0.035 &  0.004 &  NA &  0.033 &  0.008 &     NA &  0.040 &  0.001 \\
    &          & \textbf{$\rho=0.7$} &  96 &  0.026 &  0.004 &   96 &  0.026 &  0.003 &  NA &  0.026 &  0.003 &     NA &  0.033 &  0.007 \\
\cline{2-15}
    & \multirow{3}{*}{\textbf{$n=2000$}} & \textbf{$\rho=0.2$} &  90 &  0.063 &  0.007 &   87 &  0.070 &  0.008 &  NA &  0.062 &  0.021 &     NA &  0.058 &  0.010 \\
    &          & \textbf{$\rho=0.5$} &  94 &  0.024 &  0.001 &   92 &  0.024 &  0.000 &  NA &  0.022 &  0.002 &     NA &  0.023 &  0.007 \\
    &          & \textbf{$\rho=0.7$} &  94 &  0.019 &  0.001 &   94 &  0.019 &  0.001 &  NA &  0.018 &  0.002 &     NA &  0.017 &  0.007 \\
\bottomrule
\end{tabular}
\caption{Experimental results from simulations study over $100$ experiment repetitions.}\label{fig:experiments}
\end{figure}

\begin{figure}[H]
\scriptsize
\begin{tabular}{lll||lll|lll|lll|lll|}
\toprule
    &          &            & \multicolumn{3}{c|}{dr} & \multicolumn{3}{c|}{tmle} & \multicolumn{3}{c|}{ipw} & \multicolumn{3}{c|}{direct} \\
    &          &            & cov &   rmse &   bias &  cov &   rmse &   bias & cov &   rmse &   bias &    cov &   rmse &   bias \\
\midrule
\multirow{6}{*}{\textbf{abs}} & \multirow{2}{*}{\textbf{$n=500$}} & \textbf{$\rho=0.2$} &  94 &  0.138 &  0.016 &   80 &  0.227 &  0.052 &  NA &  0.102 &  0.010 &     NA &  0.124 &  0.007 \\
    &          & \textbf{$\rho=0.5$} &  95 &  0.066 &  0.018 &   94 &  0.069 &  0.018 &  NA &  0.063 &  0.016 &     NA &  0.070 &  0.019 \\
\cline{2-15}
    & \multirow{2}{*}{\textbf{$n=1000$}} & \textbf{$\rho=0.2$} &  90 &  0.096 &  0.001 &   87 &  0.131 &  0.015 &  NA &  0.093 &  0.003 &     NA &  0.083 &  0.019 \\
    &          & \textbf{$\rho=0.5$} &  98 &  0.042 &  0.008 &   98 &  0.043 &  0.007 &  NA &  0.045 &  0.006 &     NA &  0.037 &  0.001 \\
\cline{2-15}
    & \multirow{2}{*}{\textbf{$n=2000$}} & \textbf{$\rho=0.2$} &  90 &  0.076 &  0.000 &   83 &  0.089 &  0.000 &  NA &  0.074 &  0.001 &     NA &  0.066 &  0.015 \\
    &          & \textbf{$\rho=0.5$} &  93 &  0.032 &  0.001 &   91 &  0.033 &  0.000 &  NA &  0.037 &  0.001 &     NA &  0.025 &  0.002 \\
\cline{1-15}
\cline{2-15}
\multirow{6}{*}{\textbf{2dpoly}} & \multirow{2}{*}{\textbf{$n=500$}} & \textbf{$\rho=0.2$} &  81 &  0.175 &  0.045 &   73 &  0.247 &  0.022 &  NA &  0.206 &  0.156 &     NA &  0.202 &  0.117 \\
    &          & \textbf{$\rho=0.5$} &  90 &  0.095 &  0.008 &   92 &  0.092 &  0.010 &  NA &  0.130 &  0.019 &     NA &  0.117 &  0.015 \\
\cline{2-15}
    & \multirow{2}{*}{\textbf{$n=1000$}} & \textbf{$\rho=0.2$} &  87 &  0.125 &  0.015 &   86 &  0.149 &  0.002 &  NA &  0.151 &  0.096 &     NA &  0.119 &  0.056 \\
    &          & \textbf{$\rho=0.5$} &  94 &  0.057 &  0.010 &   94 &  0.057 &  0.010 &  NA &  0.073 &  0.008 &     NA &  0.057 &  0.006 \\
\cline{2-15}
    & \multirow{2}{*}{\textbf{$n=2000$}} & \textbf{$\rho=0.2$} &  92 &  0.090 &  0.010 &   87 &  0.100 &  0.008 &  NA &  0.119 &  0.058 &     NA &  0.085 &  0.039 \\
    &          & \textbf{$\rho=0.5$} &  90 &  0.043 &  0.002 &   90 &  0.044 &  0.003 &  NA &  0.070 &  0.011 &     NA &  0.036 &  0.000 \\
\cline{1-15}
\cline{2-15}
\multirow{6}{*}{\textbf{sigmoid}} & \multirow{2}{*}{\textbf{$n=500$}} & \textbf{$\rho=0.2$} &  89 &  0.101 &  0.024 &   79 &  0.185 &  0.039 &  NA &  0.105 &  0.062 &     NA &  0.110 &  0.033 \\
    &          & \textbf{$\rho=0.5$} &  96 &  0.052 &  0.017 &   94 &  0.054 &  0.017 &  NA &  0.046 &  0.019 &     NA &  0.056 &  0.014 \\
\cline{2-15}
    & \multirow{2}{*}{\textbf{$n=1000$}} & \textbf{$\rho=0.2$} &  95 &  0.072 &  0.010 &   90 &  0.099 &  0.018 &  NA &  0.084 &  0.045 &     NA &  0.065 &  0.000 \\
    &          & \textbf{$\rho=0.5$} &  97 &  0.034 &  0.006 &   97 &  0.034 &  0.006 &  NA &  0.034 &  0.010 &     NA &  0.032 &  0.003 \\
\cline{2-15}
    & \multirow{2}{*}{\textbf{$n=2000$}} & \textbf{$\rho=0.2$} &  91 &  0.055 &  0.005 &   86 &  0.064 &  0.005 &  NA &  0.060 &  0.027 &     NA &  0.052 &  0.006 \\
    &          & \textbf{$\rho=0.5$} &  93 &  0.024 &  0.001 &   94 &  0.024 &  0.000 &  NA &  0.021 &  0.003 &     NA &  0.022 &  0.003 \\
\cline{1-15}
\cline{2-15}
\multirow{6}{*}{\textbf{sin}} & \multirow{2}{*}{\textbf{$n=500$}} & \textbf{$\rho=0.2$} &  88 &  0.109 &  0.024 &   76 &  0.195 &  0.046 &  NA &  0.097 &  0.045 &     NA &  0.101 &  0.018 \\
    &          & \textbf{$\rho=0.5$} &  98 &  0.046 &  0.011 &   96 &  0.047 &  0.011 &  NA &  0.046 &  0.019 &     NA &  0.047 &  0.009 \\
\cline{2-15}
    & \multirow{2}{*}{\textbf{$n=1000$}} & \textbf{$\rho=0.2$} &  93 &  0.083 &  0.015 &   89 &  0.113 &  0.023 &  NA &  0.085 &  0.035 &     NA &  0.075 &  0.004 \\
    &          & \textbf{$\rho=0.5$} &  93 &  0.035 &  0.004 &   93 &  0.036 &  0.004 &  NA &  0.034 &  0.010 &     NA &  0.032 &  0.000 \\
\cline{2-15}
    & \multirow{2}{*}{\textbf{$n=2000$}} & \textbf{$\rho=0.2$} &  92 &  0.062 &  0.007 &   86 &  0.070 &  0.008 &  NA &  0.062 &  0.021 &     NA &  0.059 &  0.009 \\
    &          & \textbf{$\rho=0.5$} &  93 &  0.024 &  0.000 &   93 &  0.024 &  0.000 &  NA &  0.021 &  0.002 &     NA &  0.023 &  0.011 \\
\bottomrule
\end{tabular}
\caption{Experimental results with clever instrument approach from simulations study over $100$ experiment repetitions.}\label{fig:experiments-clever}
\end{figure}

\begin{figure}[H]
\begin{center}
\scriptsize
\begin{tabular}{ll||lll|lll|lll|lll|}
\toprule
          &            & \multicolumn{3}{c|}{dr} & \multicolumn{3}{c|}{tmle} & \multicolumn{3}{c|}{ipw} & \multicolumn{3}{c|}{direct} \\
          &            & cov & $\frac{\textbf{rmse}}{|\theta^*|}$ & $\frac{\textbf{bias}}{|\theta^*|}$ &  cov & $\frac{\textbf{rmse}}{|\theta^*|}$ & $\frac{\textbf{bias}}{|\theta^*|}$ & cov & $\frac{\textbf{rmse}}{|\theta^*|}$ & $\frac{\textbf{bias}}{|\theta^*|}$ &    cov & $\frac{\textbf{rmse}}{|\theta^*|}$ & $\frac{\textbf{bias}}{|\theta^*|}$ \\
\midrule
\multirow{2}{*}{\textbf{$n=2000$}} & \textbf{$\rho=0.05$} &  49 &                              0.650 &                              0.436 &   41 &                              6.038 &                              0.929 &  NA &                              0.825 &                              0.717 &     NA &                              0.687 &                              0.581 \\
          & \textbf{$\rho=0.1$} &  83 &                              0.411 &                              0.126 &   77 &                              0.578 &                              0.032 &  NA &                              0.521 &                              0.358 &     NA &                              0.452 &                              0.328 \\
\cline{1-14}
\multirow{2}{*}{\textbf{$n=20000$}} & \textbf{$\rho=0.05$} &  91 &                              0.316 &                              0.062 &   86 &                              0.373 &                              0.008 &  NA &                              0.419 &                              0.236 &     NA &                              0.329 &                              0.222 \\
          & \textbf{$\rho=0.1$} &  94 &                              0.160 &                              0.006 &   90 &                              0.169 &                              0.006 &  NA &                              0.221 &                              0.090 &     NA &                              0.134 &                              0.047 \\
\bottomrule
\end{tabular}
\end{center}
\caption{Experimental results for $\textbf{2dpoly}$ non-linearity, for very weak instrument, without clever instrument, from simulations study over $100$ experiment repetitions.}\label{fig:experiments-very-weak}
\end{figure}

\section{Conclusions and Future Directions}\label{sec: conclusion}
In this paper, we study the estimation of and inference on strongly identified linear functionals of weakly identified nuisance functions.  This challenge arises in a variety of applications in causal inference and missing data, where the primary nuisance (\eg, NPIV regression) is defined by an ill-posed conditional moment equation.
Side-stepping conditions that directly control the estimability of the primary nuisance function(s), we propose a novel assumption for the strong identification of just the functional of it. Mere identification of the functional (\ie, uniqueness) is equivalent to restricting the Riesz representer of the functional to a certain subspace. Our assumption, which posits the existence of a solution to a certain optimization problem, can be seen as further restricts it to a smaller subspace. The optimization formulation of our assumption directly motivates us to propose new minimax estimators for the unknown nuisances. 
Via a novel analysis, we show our nuisance estimators can converge to fixed limits in terms of the $L_2$ norm error even when the nuisances are underidentified. Moreover, our nuisance estimators can accommodate a wide variety of flexible machine learning methods like RKHS methods or neural networks. We then use our estimators to construct debiased estimators for the functionals of interest. Put together, we obtain high-level conditions under which our functional estimator is asymptotically normal and under which our estimated-variance Wald confidence intervals are valid.

There are several interesting future directions of research. Our paper currently focuses on \emph{linear} functionals of nuisance functions defined by \emph{linear} conditional moment equations. It would be interesting to explore more general \emph{nonlinear} functionals and/or \emph{nonlinear} conditional moment equations without requiring point identification of the nuisances. 
{As an example of the former, we may consider inference on consumer surplus and deadweight loss as functionals of a demand function estimated using an IV for an endogenous price \citep[\eg, ][]{chen2018optimal}.}
As an example of the latter, we may consider functionals of NPIV \emph{quantile} regressions \citep[\eg, Example 3.3 in][]{ai2009semiparametric}. 
One important challenge in this direction is to establish the identifiability of the functionals of interest \citep{chen2014local}. 
Moreover, our paper studies point-identified functionals (though the nuisances are not identified). 
We may relax the identification restriction on the functionals, and instead study partial identification bounds of the functionals \citep{escanciano2013identification}. In particular, debiased inference on partial identification bounds is an area of growing interest \citep{dorn2021doubly,kallus2022assessing,kallus2022s,yadlowsky2018bounds} where our new theory may provide new directions.

\section*{Acknowledgements}
This material is based upon work supported by the National Science Foundation under Grants No. 1846210 and 1939704.
Xiaojie Mao acknowledges support from National Natural Science Foundation of China (No. 72201150 and No. 72293561) and National Key R\&D Program of China (2022ZD0116700).

\bibliographystyle{plainnat}
\bibliography{semiparametric}

\newpage 
\appendix 

\begin{center}\LARGE
Appendices
\end{center}

\section{Generalized Framework}
\label{sec:generalized-framework}

We briefly discuss a generalized version of the conditional moment restriction framework laid out in \Cref{sec: setup}, and provide corresponding results for estimation and inference in this more general setting. Here, we let some Hilbert spaces $\Hcal \subseteq L_2(S)^{d_h}$ and $\Qcal \subseteq L_2(T)^{d_q}$ be given under the standard $L_2$ inner products, for some positive integers $d_h$ and $d_q$. 

In this general setting, we are interested in the parameter $\theta^\star = \EE[m(W;h^\star)]$, 
where the primary nuisance function $h^\star \in \Hcal$ solves the following orthogonality condition 
\begin{equation}
\label{eq:general-id}
    \EE\Big[ G(W;h,q) - r(W;q) \Big] = 0 \qquad \forall q \in \Qcal \,,
\end{equation}
for some given $r$ such that $q \mapsto \EE[r(W;q)]$ is a continuous linear functional, that is, there exists $\beta \in \Qcal$ such that
\begin{equation*}
    \EE[r(W;q)] = \EE[\beta(T)^\top q(T)] \qquad \forall q \in \Qcal,
\end{equation*}
and some given $G$ such that there exists $k(S,T) \in L_2(S,T)^{d_h \times d_q}$ such that
\begin{equation*}
    \EE[G(W;h,q)] = \EE[h(S)^\top k(S,T) q(T)] \qquad \forall h \in \Hcal, q \in \Qcal \,.
\end{equation*}
That is, $h, q \mapsto \EE[G(W;h,q)]$ is bi-linear functional and satisfies a bi-linear version of the standard Riesz representation formula for continuous linear operators. 
Again, we allow for \emph{weak identification} of $h^\star$ via the orthogonality condition in \Cref{eq:general-id}. That is, \Cref{eq:general-id} can be severely ill-posed and  may admit multiple solutions. As before, we let $\Hcal_0$ denote the set of all such solutions.

We again focus on $m$ such that $h \mapsto \EE[m(W;h)]$ is a continuous linear functional; that is, there exists a unique $\alpha \in \Hcal$ such that
\begin{equation*}
    \EE[m(W;h)] = \EE[\alpha(S)^\top h(S)] \qquad \forall h \in \Hcal \,.
\end{equation*}

This framework subsumes the simpler framework considered in \Cref{sec: setup}. Indeed, we can recover the setup in \Cref{sec: setup} by setting $d_h=1$, $d_q=1$, $\Qcal=L_2(T)$, $G(W,h,q) = g_1(W) h(S) q(T)$, and $r(W;q) = g_2(W) q(T)$, because then the above equations are satisfied with $k(S,T) = \EE[g_1(W) \mid S,T]$, $\beta(T) = \EE[g_2(W) \mid T]$, and
\begin{equation*}
    \EE[g_1(W) h_0(S) q(T) - g_2(W) q(T)] = 0 \ \ \forall q \in L_2(T) \quad \iff \quad \EE[g_1(W) h_0(S) - g_2(W) \mid T] = 0 \,.
\end{equation*}
However, the framework in this section is much more general. For example, it can encompass multiple conditional moment restrictions with different conditioning variables \citep{ai2012semiparametric,ai2007estimation}, or the treatment bridge function in proximal causal inference as the primary nuisance \citep{cui2020semiparametric}.

Now, define the operator $P : \Hcal \to \Qcal$ and its adjoint $P^\star : \Qcal \to \Hcal$ by 
\begin{align}\label{eq: operator-general}
[P h](T) = \Pi_\Qcal[k(S,T)^\top h(S) \mid T], ~~ [P^\star q](S) = \Pi_\Hcal[k(S,T) q(T) \mid S], ~~ \forall h \in\Hcal, q\in\Qcal. 
\end{align}
These generalize the  corresponding operators $P$ and $P^\star$ in \Cref{sec: setup,sec: identification}.  
We can then verify that \Cref{lemma: identifiability} is still true, namely, $\theta^\star$ is identified by all $h_0$ satisfying the orthogonality condition in \Cref{eq:general-id} if and only if $\alpha \in \Ncal(P)^\perp = \cl(\Rcal(P^\star))$.

To enable the inference of $\theta^\star$, we need to strengthen the identification condition. 
In particular, the condition in \Cref{eq: Q0} can be  generalized to 
\begin{equation*}
    \Qcal_0 \ne \emptyset, ~~ \text{where } \Qcal_0 = \Big\{q_0 \in \Qcal: P^\star q = \alpha\Big\} = \Big\{q_0 \in \Qcal : \EE\Big[G(W;h,q_0) - m(W;h)\Big] = 0, \ \forall h \in \Hcal \Big\} \,.
\end{equation*}
We further strengthen this condition by  imposing our strong identification condition in \Cref{assump: new-nuisance}. This requires the existence of a solution $\xi_0\in\Hcal$ to \Cref{eq: delta-eq-2}, where the optimization objective now depends on the  general operator $P$ given in \Cref{eq: operator-general}. 
We can easily prove that \Cref{thm: new-nuisance} still holds, namely, \Cref{assump: new-nuisance} restricts the Riesz representer $\alpha$ to the range space $\Rcal(P^\star P)$. 
Moreover, like \Cref{lemma: minimum-norm}, the debiasing nuisance $P\xi_0$ for $\xi_0 \in \Xi_0$ is the minimum-norm function in $\Qcal_0$.

Next, we define the doubly robust equation
\begin{equation}\label{eq: dr-formula-general}
    \psi(W;h,q) = m(W;h) + r(W;q) - G(W;h,q) \,.
\end{equation}
Then, it is straightforward to argue that \Cref{thm:dr-ident} still holds with this equation. 
We formally state the results in the following lemma. 
\begin{lemma}
\label{lem:general-dr}
    Let some arbitrary $h_0 \in \Hcal_0$ and $q_0 \in \Qcal_0$ be given and $\psi(W;h,q)$ be the doubly robust equation in \Cref{eq: dr-formula-general}. For every $h \in \Hcal$ and $q \in \Qcal$, we have \begin{equation*}
        \EE[\psi(W;h,q)] - \theta^\star = \EE[G(W;h-h_0,q-q_0)] \,.
    \end{equation*}
    Consequently, it follows that for any $h \in \Hcal$ and $q \in \Qcal$ we have
    \begin{equation*}
        \Big|\EE[\psi(W;h,q)] - \theta^\star\Big| \leq \min\Big\{ \|P(h - h_0)\|_{2,2} \|q - q_0\|_{2,2}, \|h - h_0\|_{2,2} \|P^\star(q - q_0) \|_{2,2}\Big\} \,,
    \end{equation*}
    and 
    \begin{equation*}
        \left. \frac{d}{dt}  \EE\Big[\psi(W;h_0+th, q_0)\Big]\right|_{t=0} = \left. \frac{d}{dt}  \EE\Big[\psi(W;h_0, q_0+tq)\Big]\right|_{t=0}  = 0 \,.
    \end{equation*}
\end{lemma}

Then, we can apply a similar debiased inference analysis in this more general setting, as long as we can get similar rates for $\|P(\hat h_n - h_0)\|_{2,2}$ and $\|q - q_0\|_{2,2}$ as we obtained in \Cref{sec: est-nuisance}. Below, we provide appropriate conditions under which we can obtain such rates by generalizing  \Cref{thm:h-estimator-bound,thm:q-estimator-bound}.

\subsection{Penalized Estimation of the Primary Nuisance Function}

Here, we consider estimators for the minimum-norm function $h^\dagger$ in $\Hcal_0$
as follows:
\begin{equation}
\label{eq:h-estimator-general}
    \hat h_n = \argmin_{h \in \Hcal_n} \sup_{q \in \Qcal_n} \EE_n \Big[ G(W;h,q) - r(W;q) - \frac{1}{2} q(Z)^\top q(Z) + \mu_n h(S)^\top h(S) \Big] - \gamma_n^q \|q\|_\Qcal^2 + \gamma_n^h \|h\|_\Hcal^2 \,,
\end{equation}
where all definitions are analogous to those in \Cref{sec: primary-nuisance-est}. 
 Below we provide natural generalizations for the technical assumptions in \Cref{sec: primary-nuisance-est}.

\begin{assumption}
\label{assum:boundedness-general}
    We have that: (1) $\|G(W;h,q)\|_2 \leq \|q\|_{2,2}$ for all $q \in \bar\Qcal$ and $h \in \bar\Hcal$; (2) $\|r(W;q)\|_2 \leq \|q\|_{2,2}$ for every $q \in \Qcal$; (3) $\|G(W;h,q)\|_\infty \leq 1$ for every $h \in \bar\Hcal$ and $q \in \bar\Qcal$; and (4) $\|P(h - h^\dagger)\|_{2,2} \leq \|h-h^\dagger\|_{2,2}$ for every $h \in \bar\Hcal$.
\end{assumption}

\begin{assumption}
\label{assum:bounded-estimation-h-general}
    We have that: (1) $\|h\|_\infty \leq 1$ for all $h \in \Hcal_n$; (2) $\|q\|_\infty \leq 1$ for all $q \in \Qcal_n$; and (3) $\|h^\dagger\|_\infty \leq 1$.
\end{assumption}

\begin{assumption}
\label{assum:universal-approximation-h-general}
    There exists some $\delta_n < \infty$ such that: (1) there exists $\Pi_n h^\dagger \in \Hcal_n$ such that $\|\Pi_n h^\dagger - h^\dagger\|_{2,2} \leq \delta_n$; and (2) for every $q \in \{P(h - h^\dagger) : h \in \Hcal_n\}$, there exists $\Pi_n q \in \Qcal_n$ such that $\|\Pi_n q - q\|_{2,2} \leq \delta_n$.
\end{assumption}

\begin{assumption}
\label{assum:complexity-h-general}
    There exists some $r_n$ that upper bounds the critical radii of the function classes $\{G(W;h,q) : h \in \starcls(\Hcal_n - h^\dagger), q \in \starcls(\Qcal_n), \|h\|_\Hcal \leq 1, \|q\|_\Qcal \leq 1\}$ and $\{q\in \starcls(\Qcal_n) : \|q\|_\Qcal \leq 1\}$. 
\end{assumption}

\setcounter{subassumption}{0}
\begin{assumption*}
\label{assum:bounded-complexity-h-general}
    There exists some constant $M \geq 1$ such that: (1) $\|h\|_\Hcal, \|h-h^\dagger\|_\Hcal \leq M$ for all $h \in \Hcal_n$; (2) $\|q\|_\Qcal \leq M$ for all $q \in \Qcal_n$; and (3) $\|h^\dagger\|_\Hcal, \|\Pi_n h^\dagger\|_\Hcal \leq M$.
\end{assumption*}

\begin{assumption*}
\label{assum:regularization-h-general}
    There exists some constants $M \geq 1$ and $L \geq 1$ such that: (1) $\|h^\dagger\|_\Hcal, \|\Pi_n h^\dagger\|_\Hcal \leq M$; (2) $\|\Pi_n P(h - h^\dagger)\|_\Qcal \leq L \|h - h^\dagger\|_\Hcal$ for all $ \in \Hcal_n$; and (3) for some universal constants $c_1$, $c_2$, and $c_3$, we have
    \begin{align*}
        \gamma_n^q &\geq c_2 \Big(r_n + \sqrt{\log(c_1/\zeta)/n} \Big)^2 \\
        \text{and} \quad \gamma_n^h &\geq c_3 L^2 \Big(\gamma_n^q + \Big(r_n + \sqrt{\log(c_1/\zeta)/n} \Big)^2 \Big) \,.
    \end{align*}
\end{assumption*}

It is trivial to verify that these conditions generalize those in \Cref{sec: est-nuisance}. Then, under these assumptions, we can provide the following generalization of \Cref{thm:h-estimator-bound}.

\begin{theorem}
\label{thm:h-estimator-bound-general}
    Suppose \Cref{assum:universal-approximation-h-general,assum:bounded-estimation-h-general,assum:complexity-h-general,assum:boundedness-general} hold, as well as either \Cref{assum:bounded-complexity-h-general} or \Cref{assum:regularization-h-general}. Then, given some universal constant $c_0$, we have that, for $\zeta\in(0,1/3)$, with probability at least $1-3\zeta$,
    \begin{equation*}
          \|P(\hat h_n - h_0)\|_{2,2} \leq c_0 \Big( M r_n + M \sqrt{\log(c_1/\zeta)/n} + \delta_n + M (\gamma_n^q)^{1/2} + M (\gamma_n^h)^{1/2} +  \mu_n^{1/2} \Big) \,,
    \end{equation*}
    for any $h_0 \in \Hcal_0$, where $c_1$ is the same universal constant as in \Cref{assum:regularization-h-general}.

    Furthermore, suppose that either of the above sets of assumptions hold, and in addition that: (1) $\mu_n = o(1)$; (2) $\mu_n = \omega(\max(r_n^2,\delta_n^2,\gamma_n^q,\gamma_n^h,1/n))$; (3) $\{h : h \in \starcls(\Hcal_n - h^\dagger), \|h\|_\Hcal \leq 1\}$ has critical radius at most $r_n$; (4) $\{h \in \bar\Hcal : \|h\|_\Hcal \leq U\}$ is compact under $\|\cdot\|_\Hcal$ for every $U < \infty$; and (5) $\|h\|_{2,2} \leq K \|h\|_\Hcal$ for all $h \in \bar\Hcal$ and some constant $K < \infty$. Then, we have
    \begin{equation*}
        \|\hat h_n - h^\dagger\|_{2,2} = o_p(1) \,.
    \end{equation*}
    
\end{theorem}

\subsection{Estimation of the Debiasing Nuisance Function}

Next, we consider the estimation of the debiasing nuisance function $q^\dagger$ in the general setting. Here, we will consider estimators of the form
\begin{equation}
\label{eq:q-estimator-general}
    \hat q_n = \argmax_{q \in \widetilde\Qcal_n} \EE_n \Big[ G(W;\hat\xi_n,q) - \frac{1}{2} q(T)^\top q(T) \Big] - \tilde\gamma_n^q \|q\|_{\widetilde\Qcal}^2 \,,
\end{equation}
where
\begin{equation}\label{eq:xi-estimator-general}
    \hat \xi_n = \argmin_{\xi \in \Xi_n} \sup_{q \in \Qcal_n} \EE_n \Big[ G(W;\xi,q) - m(W;\xi) - \frac{1}{2} q(T)^\top q(T) \Big] - \gamma_n^q \|q\|_{\Qcal}^2 + \gamma_n^\xi \|\xi\|_{\Xi}^2 \,.
\end{equation}
Again, all definitions here are analogues to those in \Cref{sec: debias-nuisance-est}. Notably, the minimax objective function in \Cref{eq:xi-estimator-general} is very close to the minimax objective in \Cref{eq:h-estimator-general}, except that $m(W;\xi)$ appears in \Cref{eq:xi-estimator-general} while $r(W; q)$ appears in \Cref{eq:h-estimator-general}.

Below, we further generalize our previous technical assumptions in \Cref{sec: debias-nuisance-est}.

\begin{assumption}
\label{assum:bounded-estimation-q-general}
    We have that: (1) $\|\xi\|_\infty \leq 1$ for every $\xi \in \Xi_n$; (2) $\|q\|_\infty \leq 1$ for every $q \in \Qcal_n$; (3) $\|q\|_\infty \leq 1$ for every $q \in \widetilde\Qcal_n$; (4) $\|q^\dagger\|_\infty \leq 1$; and (5) $\|\xi^\dagger\|_\infty \leq 1$ 
\end{assumption}

\begin{assumption}
\label{assum:universal-approximation-q-general}
    There exists some $\delta_n < \infty$ such that: (1) there exists some $\Pi_n \xi^\dagger \in \Xi_n$ such that $\|\Pi_n \xi^\dagger - \xi^\dagger\|_{2,2} \leq \delta_n$; (2) for every $q \in \{P \xi : \xi \in \Xi_n\}$ there exists $\Pi_n q \in \Qcal_n$ such that $\|q - \Pi_n q\|_{2,2} \leq \delta_n$; and (3) there exists $\Pi_n q^\dagger \in \widetilde \Qcal_n$ such that $\|q^\dagger - \Pi_n q^\dagger\|_{2,2} \leq \delta_n$.
\end{assumption}

\begin{assumption}
\label{assum:complexity-q-general}
    There exists some $r_n$ that bounds the critical radii of the star-shaped closures of the function classes: (1) $\{g_1(W)\xi(S) q(T) : \xi \in \starcls(\Xi_n - \xi^\dagger), q \in \starcls(\widetilde\Qcal_n - q^\dagger), \|\xi\|_\Xi \leq 1, \|q\|_{\widetilde\Qcal} \leq 1\}$; (2) $\{q \in \starcls(\Qcal_n - q^\dagger) : \|q\|_\Qcal \leq 1\}$; (3) $\{q \in \starcls(\widetilde\Qcal_n - q^\dagger) : \|q\|_{\widetilde\Qcal} \leq 1\}$; (4) $\{\xi \in \starcls(\xi_n - \xi^\dagger) : \|\xi\|_{\Xi} \leq 1\}$; and (5) $\{g_1(W)\xi(S) q(T) : \xi \in \starcls(\Xi_n - \xi^\dagger), q \in \starcls(\Qcal_n - q^\dagger), \|\xi\|_\Xi \leq 1, \|q\|_\Qcal \leq 1\}$.
\end{assumption}

\begin{assumption}
\label{assum:boundedness-q-general}
    We have that $\|m(W;q)\|_2 \leq \|q\|_{2,2}$ for all $q \in \bar\Qcal$.
\end{assumption}

\setcounter{subassumption}{0}
\begin{assumption*}
\label{assum:bounded-complexity-q-general}
    There exists some constant $M \geq 1$ such that: (1) $\|\xi\|_\Xi, \|\xi-\xi^\dagger\|_\Xi \leq M$ for all $\xi \in \Xi_n$; (2) $\|q\|_\Qcal, \|q-q^\dagger\|_\Qcal \leq M$ for all $q \in \Qcal_n$; (3) $\|q\|_{\widetilde\Qcal}, \|q-q^\dagger\|_{\widetilde\Qcal} \leq M$ for all $q \in \widetilde\Qcal_n$; (4) $\|\xi^\dagger\|_\Xi, \|\Pi_n \xi^\dagger\|_\Xi \leq M$; (5) $\|q^\dagger\|_\Qcal, \|\Pi_n q^\dagger\|_\Qcal \leq M$; and (6) $\|q^\dagger\|_{\widetilde\Qcal,} \|\Pi_n q^\dagger\|_{\widetilde\Qcal} \leq M$.
\end{assumption*}

\begin{assumption*}
\label{assum:regularization-q-general}
    There exists some constants $M \geq 1$ and $L \geq 1$ such that: (1) $\|\xi^\dagger\|_\Xi, \|\Pi_n \xi^\dagger\|_\Xi \leq M$; (2) $\|q^\dagger\|_\Qcal, \|\Pi_n q^\dagger\|_\Qcal \leq M$; (3) $\|\Pi_n P \xi\|_\Qcal \leq L \|\xi\|_\Xi$ for all $\xi \in \Xi_n$; and (4) for some universal constants $c_1$, $c_2$, $c_3$, and $c_4$ we have
    \begin{align*}
        \gamma_n^q &\geq c_2 \Big(r_n + \sqrt{\log(c_1/\zeta)/n} \Big)^2 \\
        \tilde\gamma_n^q &\geq c_3 \Big(r_n + \sqrt{\log(c_1/\zeta)/n} \Big)^2 \\
        \text{and} \quad \gamma_n^h &\geq c_4 L^2 \Big(\gamma_n^q + \Big(r_n + \sqrt{\log(c_1/\zeta)/n} \Big)^2 \Big) \,.
    \end{align*}
\end{assumption*}

Then, we provide the following generalization of \Cref{thm:q-estimator-bound}.

\begin{theorem}
   \label{thm:q-estimator-bound-general}
    Suppose \Cref{assum:bounded-estimation-q-general,assum:universal-approximation-q-general,assum:complexity-q-general,assum:boundedness-general,assum:boundedness-q-general} hold, as well as either \Cref{assum:bounded-complexity-q-general} or \Cref{assum:regularization-q-general}. Then, given some universal constant $c_0$, we have that, for $\zeta\in(0,1/7)$, with probability at least $1-7\zeta$,
    \begin{align*}
          \|\hat q_n - q^\dagger\|_{2,2} &\leq c_0 \Big( r_n^{1/2} + \prns{\log(c_1/\zeta)/n}^{1/4} + M r_n + M \sqrt{\log(c_1/\zeta)/n} \\
          &\qquad + \delta_n + M (\gamma_n^q)^{1/2} + M (\gamma_n^h)^{1/2} + M (\tilde \gamma_n^q)^{1/2}  \Big) \,,
    \end{align*}
    where $c_1$ is the same universal constant as in \Cref{assum:regularization-q}.
\end{theorem}

\subsection{Asymptotic Normality and Inference}

Finally, we provide analogues of the results form \Cref{sec: est-functional} for the more general setting.

\begin{theorem}
\label{thm:dml-asymp-general}
    Let the estimator $\hat\theta_n$ be defined following the cross-fitting procedure as in \Cref{def: theta-est}, with $\hat q$ and $\hat h$ defined following \Cref{eq:h-estimator-general,eq:q-estimator-general} respectively, 
    and $\psi(W;h,q)$ defined according to  
    \Cref{eq: dr-formula-general}.
    Suppose that the full conditions of \Cref{thm:h-estimator-bound-general,thm:q-estimator-bound-general} hold. Then, as long as $r_n = o(n^{-1/3}), \delta_n = o(n^{-1/4}), \delta_nr_n^{1/2} = o(n^{-1/2}), \mu_n r_n = o(n^{-1})$ and $\mu_n \delta_n^2 = o(n^{-1})$, we have $\|P(\hat h_n - h_0)\|_{2,2} \|\hat q_n - q^\dagger\|_{2,2} = o_p(n^{-1/2})$,  and  that as $n\to\infty$,
        \begin{align*}
\sqrt{n}\prns{\hat\theta_n - \theta^\star} = \frac{1}{\sqrt{n}}\sum_{i=1}^n \prns{\psi(W_i;  h^\dagger,  q^\dagger) -\theta^\star} + o_p(1) \rightsquigarrow \mathcal{N}\prns{0, \sigma_0^2} 
\end{align*}
    where $\Ncal(0,\sigma_0^2)$ denotes a Gaussian distribution with mean $0$ and variance 
    \begin{equation*}
        \sigma_0^2 = \EE\Big[ \Big( \theta^\star - \psi(W;h^\dagger,q^\dagger) \Big)^2 \Big]
    \end{equation*}
\end{theorem}

Similarly, we can construct a variance estimator and confidence interval by following \Cref{eq: variance-estimator,eq: CI} respectively, and prove that they are asymptotically valid.

\begin{theorem}
\label{lem:inference-general}
    Let $\hat\sigma_n^2$ and $\op{CI}$ be the variance estimator and confidence interval constructed from \Cref{eq: variance-estimator,eq: CI}, with the $\psi(W; h, q)$ function defined according to  
    \Cref{eq: dr-formula-general}.
    If the conditions of \Cref{thm:dml-asymp-general} hold, 
    then as $n\to\infty$, 
    $\hat\sigma_n^2$ converges to $\sigma_0^2$  in probability, and $\Prb{\theta^\star \in \op{CI}} \to 1-\alpha$.  
\end{theorem}

\section{More on Connections to Some Previous Literature}\label{sec: connection}

In this section, we connect our paper to \citet{ichimura2022influence,ai2007estimation}. These two papers study   functionals of some functions defined by conditional moment restrictions. 
They can handle general nonlinear functionals and nonlinear conditional moment restrictions, but they focus on point-identified conditional moment restrictions.
In contrast, our paper focuses on linear functionals and linear conditional moment restrictions but we allow for under-identified conditional moment restrictions. 
In this section, we specialize the results in \citet{ichimura2022influence,ai2007estimation} to the linear setting studied in our paper, aiming  to connect our strong identification condition in \Cref{assump: new-nuisance} to conditions in these two existing papers.

\subsection{Connection to \citet{ichimura2022influence}}\label{sec: ichimura}
\citet{ichimura2022influence} study the influence function of a two-stage sieve estimator for the parameter $\theta^\star = \Eb{m(W; h^\star)}$ where $h \mapsto \Eb{m(W; h)}$ defines a functional of $h \in \Hcal \subseteq \Lcal_2(S)$, and $h^\star$ is \emph{uniquely} identified by the orthogonality condition 
\begin{align*}
\Eb{q(T)\rho(W; h)} = 0, ~~ q \in \Qcal \subseteq \Lcal_2(T), h \in \Hcal \subseteq \Lcal_2(S).
\end{align*}
for  a generalized residual function $\rho(W; h)$. 
\citet{ichimura2022influence}  study a general nonlinear residual function and hopes to estimate a general nonlinear functional.
In contrast, our paper focuses on a linear residual function and targets a linear functional $h \mapsto \Eb{m(W; h)}$.
In particular, the conditional moment formulation in our \Cref{sec: setup} can be viewed as a special example with $\Qcal = \Lcal_2(T)$ and $\rho(W; h) = g_2(W) - g_1(W)h(S)$.
Our general formulation in \Cref{sec:generalized-framework} further 
 allows for a general $\Qcal$ function class, although the corresponding residual function is still linear.

In this subsection, we relate our \Cref{sec:generalized-framework} to \citet{ichimura2022influence}. 
For simplicity, we focus on the conditional-moment formulation with $\Qcal = \Lcal_2(T)$,  $\rho(W; h) = g_2(W) - g_1(W)h(S)$ and $h\mapsto \Eb{m(W; h)}$ being a linear functional. In the notations here, \citet{ichimura2022influence} require the functional $h \mapsto \Eb{m(W; h)}$ to be continuous with a Riesz representer $\alpha\in\Hcal$.
Moreover, they assume the existence of  $q_0$ such that 
\begin{align}
\label{eq: cond-moment-q-2}
[P^\star q_0](S) = \Pi_\Hcal\bracks{g_1(W)q_0(T) \mid S} = \alpha(S).
\end{align}
This is identical to \Cref{eq: cond-moment-q}. 
Under these two assumptions, their Proposition 3 derives the influence function of a certain two-stage sieve estimator for $\theta^\star$:
\begin{align}\label{eq: IF-newey}
\varphi(W; h^\star, q^\star) = m(W; h^\star) - \theta^\star+ q^\star(T)(g_2(W) - g_1(W)h^\star(S)),
\end{align}
where $q^\star$ is the least squares projection of the $q_0$ in \Cref{eq: cond-moment-q-2} onto $\bar{\mathcal{R}}(P)$,  
the mean square closure of the range space $\mathcal{R}(P)$, namely, 
\begin{align}\label{eq: q-projection}
q^\star = \argmin_{q \in {\bar{\mathcal{R}}(P)}}\Eb{\prns{q_0(T) - q(T)}^2}.
\end{align}

In the following proposition, we show that the debiasing  nuisance derived from   our \Cref{assump: new-nuisance} is closely related to the nuisance defined in \Cref{eq: q-projection}. 
\begin{proposition}\label{prop: newey-interpret} 
For any $\xi_0 \in \Xi_0$ defined in  \Cref{assump: new-nuisance}, we have 
\begin{align}\label{eq: q-projection2}
q^\dagger = P\xi_0 = \argmin_{q \in {{\mathcal{R}}(P)}}\Eb{\prns{q_0(T) - q(T)}^2}.
\end{align}
\end{proposition}

\Cref{prop: newey-interpret}  shows that our debiasing nuisance  $q^\dagger = P\xi_0$ for $\xi_0$ given in our \Cref{assump: new-nuisance} is very similar to the function $q^\star$ in \Cref{eq: q-projection}. 
The difference is that the former is a projection onto the range space of $P$, while the latter is the projection onto the closure of the range space.  
According to the projection theorem \citep[Theorem 2, Section 3.3]{luenberger1997optimization},  the nuisance $q^\star$ defined by \citet{ichimura2022influence} in \Cref{eq: q-projection} always exists without any extra condition, since  $\bar{\mathcal{R}}(P)$  is a closed subspace of $\Lcal_2(T)$. 
In contrast, the existence of our nuisance in \Cref{assump: new-nuisance} needs to impose extra restrictions on the Riesz representer, as we discussed in \Cref{sec: identification}.
We note that when our \Cref{assump: new-nuisance} indeed holds, the influence function in \Cref{eq: IF-newey} with $q^\star = P\xi_0$ is actually identical to the influence function of our proposed functional estimator (see \Cref{thm:dml-asymp}). 

Although \citet{ichimura2022influence} do not impose restrictions like our \Cref{assump: new-nuisance}, their least squares nuisance $q^\star$ in \Cref{eq: q-projection} is not amenable to direct estimation, as it involves another unknown function $q_0$. 
Moreover, \citet{ichimura2022influence} focus on 
deriving the candidate influence function $\varphi(W; h^\star, q^\star)$ in 
 \Cref{eq: IF-newey}, but not on estimation and inference details. 
 For example,
  they do not provide conditions for when their two-stage sieve estimator is indeed asymptotically linear with $\varphi(W; h^\star, q^\star)$ as its actual influence function.
Establishing these  asymptotic guarantees will need extra conditions. 
In contrast, our paper aims to propose practical  estimation and inference methods for the  functionals of interest.
In particular, the  characterization of our  nuisance in \Cref{assump: new-nuisance} \Cref{eq: delta-eq-2} does not involve any unknown function, so it is particularly convenient for estimation. 
 This allows us to establish that our proposed estimator is asymptotically linear with the stated influence function,  under generic high-level conditions that can accommodate flexible machine learning nuisance estimators.
Importantly, our asymptotic guarantees are robust to the weak identification of the primary nuisance.

\subsection{Connection to \citet{ai2007estimation}}
\label{sec: chen}
\citet{ai2007estimation} study the semiparametric estimation and inference of possibly misspecified but point identified conditional moment equations. 
In this subsection, we specialize their results to linear functionals and  linear conditional moment restrictions considered in our paper. 
In this setting, the sieve minimum distance estimator
proposed in \citet{ai2007estimation} can be also applied to our problem, and 
the consistency and asymptotic distribution of this sieve estimator  can be established following their theory. 

In our notations, the asymptotic distribution analysis of the sieve estimator in \citet{ai2007estimation} requires the existence of solution 
$\nu^\star(S) \in \Hcal$
 to the following minimization problem:  
\begin{align*}
\nu^\star \in \argmin_{v \in \cl\prns{\Hcal - \braces{h^\star}}} \Eb{\prns{[P\nu](T)}^2} + \prns{1 + \Eb{m(W; \nu)}}^2. 
\end{align*}
Since in our setting $\Hcal$ is a closed linear space and $h^\star \in \Hcal$, we have $\cl\prns{\Hcal - \braces{h^\star}} = \Hcal$, so the problem above can be also written as 
\begin{align}\label{eq: Ai-chen-nuisance}
\nu^\star \in \argmin_{v \in \Hcal} \Eb{\prns{[P\nu](T)}^2} + \prns{1 + \Eb{m(W; \nu)}}^2. 
\end{align}
We denote the minimum objective value in \Cref{eq: Ai-chen-nuisance} as $V^\star$, \ie, the objective value attained by $\nu^\star$.

Then we can follow  the theory in \citet{ai2007estimation} to show that under suitable technical conditions, the sieve estimator for $\theta^\star$   is asymptotically linear with influence function given below: 
\begin{align*}
\varphi(W; h^\star, \nu^\star) = {V^\star}^{-1}\prns{1+\Eb{m(W; \nu^\star)}}\prns{m(W; h^\star) - \theta^\star} + {V^\star}^{-1}[P{\nu^\star}](T)(g_1(W)h^\star(S)-g_2(W)).
\end{align*}

In the following proposition, we show that the existence of an optimal solution to \Cref{eq: Ai-chen-nuisance} is actually equivalent to our \Cref{assump: new-nuisance}.
This proposition provides a new perspective for a key assumption in \citet{ai2007estimation}: their assumption implicitly restricts the functional of interest like our \Cref{assump: new-nuisance}.

\begin{proposition}\label{prop: ai-chen}
If there exists a solution $\nu^\star \in \Hcal$ to the minimization problem in \Cref{eq: Ai-chen-nuisance}, then $V^\star = 1 + \Eb{m(W; v^\star)}$,  
 $\alpha = P^\star P\xi^\star \in \Rcal(P^\star P)$ for $\xi^\star \coloneqq -\nu^\star/V^\star$, and  
\begin{align*}
\varphi(W; h^\star, \nu^\star) =\prns{m(W; h^\star) - \theta^\star} + [P{\xi^\star}](T)(g_2(W) - g_1(W)h^\star(S)).
\end{align*}
Conversely, if there exists $\xi^\star$ such that $\alpha = P^\star P\xi^\star \in \Rcal(P^\star P)$, then $v^\star \coloneqq -\xi^\star/(1 + \E[\prns{\bracks{P\xi^\star}(T)}^2])$ solves the minimization problem in \Cref{eq: Ai-chen-nuisance}. 
\end{proposition}

Although here $\nu^\star$ and $V^\star$ correspond to our debiasing nuisance function, \citet{ai2007estimation}  do not estimate them when constructing the estimator of $\theta^\star$ (although they propose to estimate $\nu^\star$ and $V^\star$ when estimating the asymptotic variance of their sieve estimator). 
In contrast, we propose to directly estimate the debiasing nuisance based on the formulation in \Cref{assump: new-nuisance} \Cref{eq: delta-eq-2}, and use it in the estimation of $\theta^\star$. 
This allows us to move beyond sieve estimation and leverage  flexible machine learning nuisance estimators.
Moreover, our paper does not require the nuisance functions to be uniquely idenfied by the conditional moment restrictions.

\section{Extension to Convex Classes}\label{sec: convex}

In the main text, we focus on a closed linear function class.
In this section, we extend the results to convex classes. In this section, we let $\Hcal$ be a closed convex function class.

\begin{assumption}\label{assump: convex}
There exists  $\xi_0 \in \op{int}\prns{\Hcal}$ such that 
\begin{align}\label{eq: delta-eq-3}
    \xi_0 \in \argmin_{\xi\in\Hcal}\Eb{\prns{[P\xi](T)}^2} - 2\Eb{m(W; \xi)} \,.
\end{align}
\end{assumption}

Here \Cref{assump: convex} is an analogue of \Cref{assump: new-nuisance}. 
Indeed, when $\Hcal$ is a closed and linear space, $\op{int}(\Hcal) = \Hcal$ so the interior restriction in \Cref{assump: new-nuisance} is vacuous. 
Here for a closed and convex class, we additionally restrict $\xi_0$ to the interior of $\Hcal$.

In the following theorem, we show an analogue of \Cref{thm: new-nuisance} that characterizes interior solution to \Cref{eq: delta-eq-3} in trems of  the  the Riesz representer. Note that the projection operator $\Pi_\Hcal$ is also well defined for a closed and convex class $\Hcal$.   

\begin{theorem}\label{theorem: convex}
If there exists $\xi_0 \in \op{int}(\Hcal)$ that solves \Cref{eq: delta-eq-3} in \Cref{assump: convex}, then 
\begin{align}\label{eq: new-nuisance-convex}
 \Pi_\Hcal\EE\bracks{g_1(W)[P\xi_0](T) \mid S} = \alpha(S).
\end{align}
\end{theorem}

We again consider the doubly robust identification given by $q^\dagger = P\xi_0$:
\begin{align*}
\theta^\star = \Eb{\psi(W; h_0, q^\dagger)},
\end{align*}
where $h_0$ is any function in $\Hcal_0$ and $\xi_0$ is any interior solution to \Cref{eq: delta-eq-3}.
In the lemma below, we show that this doubly robust identification formula still has the bias-product property  and Neyman orthogonality property. Therefore, our debiased inference theory in \Cref{sec: est-functional} still works through. 

\begin{lemma}\label{lemma: bias-product-convex}
Suppose \Cref{assump: convex} holds. Then for any $h \in \Hcal, q \in \Lcal_2(T), h_0 \in \Hcal_0$, any interior solution $\xi_0$ to \Cref{eq: delta-eq-3} and $q^\dagger = P\xi_0$, we have 
\begin{align}\label{eq: DR-identification-convex}
 \abs{\Eb{\psi(W; h, q)} - \theta^\star} 
    &= \abs{\langle{P\prns{h-h_0}, q-q^\dagger\rangle}} \le \|P\prns{h-h_0}\|_2\|q-q^\dagger\|_2.
 \end{align} 
Moreover, we have 
\begin{align}\label{eq: orthogonality-convex}
\frac{\partial}{\partial t}  \Eb{\psi(W; h_0 + t(h-h_0), q^\dagger)}\big\vert_{t = 0} = \frac{\partial}{\partial t}  \Eb{\psi(W; h_0, q^\dagger  + t(q-q_0))}\big\vert_{t = 0} = 0.
\end{align}
\end{lemma}

\section{Supporting Lemmas}
\label{sec: support}
\begin{lemma}\label{lemma: nuisance-id}
The following three conditions are equivalent:
(i) the function $h^\star$ is identifiable,
(ii) $\Ncal(\operator)=\{0\}$, 
(iii) $\E[m(W;h^\star)]$ is identifiable for every function $m$ such that $\E[m(W;\cdot)]$ is continuous linear.
\end{lemma}

\begin{lemma}[Theorem 14.1 in \citet{wainwright2019high}]\label{lem:support1}
Given a star-shaped and $b$-uniformly bounded function class $\mathcal{G}$, let $\eta_n$ be any positive solution of the inequality $\mathcal{R}_n(\mathcal{G}^{\mid \eta})\leq \eta^2/b$. Then there exist universal positive constants $c_1, c_2$, such that for any $t\geq \eta_n$, we have 
\begin{align*}
    \abs{\|g\|^2_n-\|g\|^2_2} \leq \frac{1}{2}\|g\|^2_2+\frac{1}{2} t^2, ~~ \forall g\in \mathcal{G}, 
\end{align*}
with probability at least $1 - c_1 \exp\prns{-c_2\frac{nt^2}{b^2}}$.
\end{lemma}

\begin{lemma}[Lemma 11 in  \cite{foster2019orthogonal}]\label{lem:support2}
Let $\mathcal{F}: \mathcal{X} \rightarrow \mathbb{R}^{d}$ be a $1$-uniformly bounded function class, whose $t$ th coordinate projection is denoted as $\left.\mathcal{F}\right|_{t}$. Let $\ell: \mathbb{R}^{d} \times \mathcal{Z} \rightarrow \mathbb{R}$ be a loss function Lipschitz in its first argument with a Lipschitz constant $L$. We receive an i.i.d. sample set $S=\braces{Z_{1}, \ldots, Z_{n}}$. Let $\mathcal{L}_{f}$ denote the random variable $\ell(f(X), Z)$ and let
$$
\mathbb{P} \mathcal{L}_{f}=\mathbb{E}[\ell(f(X), Z)], \quad \text { and } \quad \mathbb{P}_{n} \mathcal{L}_{f}=\frac{1}{n} \sum_{i=1}^{n} \ell\left(f\left(X_{i}\right), Z_{i}\right).
$$
There exists universal positive constants $c_1, c_2, c_3$ such that for any $\delta_{n}^{2} \geq \frac{4 d \log \left(41 \log \left(2 c_{1} n\right)\right)}{c_{1} n}$ that solves the inequalities $\mathcal{R}\left(\operatorname{star}\left(\left.\mathcal{F}\right|_{t}-f_{t}^{\star}\right), \delta\right) \leq \delta^{2}$ for any $t \in\{1, \ldots, d\}$, we have 
$$
\left|\mathbb{P}_{n}\left(\mathcal{L}_{f}-\mathcal{L}_{f^{\star}}\right)-\mathbb{P}\left(\mathcal{L}_{f}-\mathcal{L}_{f^{\star}}\right)\right| \leq 18 L d \delta_{n}\left\{\sum_{t=1}^d \left\|f_t-f_t^{\star}\right\|_{2}+\delta_{n}\right\}, \quad \forall f \in \mathcal{F},
$$
with probability at least $1-c_{2} \exp \left(-c_{3} n \delta_{n}^{2}\right)$.
\end{lemma}

\begin{lemma}
\label{lem:quadratic-inequality}
    Let some $b,d > 0$ be given, and suppose that
    \begin{equation*}
        \frac{1}{2} X^2 \leq b X + d \,,
    \end{equation*}
    for some $X > 0$. Then, we have
    \begin{equation*}
        X \leq 2 b + \sqrt{2d}
    \end{equation*}
\end{lemma}

\begin{proof}[Proof of \Cref{lem:quadratic-inequality}]

Since $b$, and $d$ are both positive, the quadratic $\frac{1}{2} X^2 - b X - d$ must have a positive and a negative root. Therefore, this quadratic is negative if and only if $X$ is less than the positive root; that is, we have
\begin{align*}
    &\frac{1}{2} X^2 - b X - d \leq 0 \\
    &\iff X \leq b + \sqrt{b^2 + 2 d} \\
    &\implies X \leq 2 b + \sqrt{2d} \,,
\end{align*}
as required.

\end{proof}

\section{Proofs of Results on Identification, Double Robustness, and Newyman Orthogonality}

Here we present the proofs of our theorems relating to the general properties of our influence function $\psi$; \emph{i.e.} those relating to identification, double robustness, and Neyman orthogonality. We present both the results for our simplified setting from \Cref{sec: identification}, as well as our generalized setting from \Cref{sec:generalized-framework}.

\subsection{Results from \Cref{sec: identification}}

\begin{proof}[Proof for \Cref{thm:dr-ident}]
\Cref{thm:dr-ident} directly follows from \Cref{lemma: bias-product-general} so we defer the details to the proof for \Cref{lemma: bias-product-general}. 
\end{proof}

\begin{proof}[Proof for \Cref{lemma: identifiability}]
First consider $\alpha \in \Ncal(\operator)^\perp$. Since any $h \in \Hcal_0$ can be written as $h = h^\star + h^{\perp}$ for some $h^{\perp} \in \Ncal(\operator)$, we have 
\begin{align*}
\Eb{m(W; h)} = \Eb{\alpha(S) h(S)} = \Eb{\alpha(S) h^\star(S)} = \Eb{m(W; h^\star)}, ~~ \forall h \in \Hcal_0. 
\end{align*}
This shows that $\theta^\star$ is identifiable when $\alpha \in \Ncal(\operator)^\perp$.

On the other hand, when $\theta^\star$ is identifiable, we must have $\alpha \in  \Ncal(\operator)^\perp$. Suppose the latter is not true, then there exists $h^\perp \in \Ncal(\operator)$ such that $\Eb{\alpha(S)h^\perp(S)} \ne 0$. Then for $h = h^\star + h^\perp \in \Hcal_0$, we have  
\begin{align*}
\Eb{m(W; h)} = \Eb{\alpha(S) h(S)} = \Eb{\alpha(S) h^\star(S)} + \Eb{\alpha(S) h^\perp(S)} \ne \Eb{\alpha(S) h^\star(S)} = \Eb{m(W; h^\star)}.
\end{align*}
This contradicts the identifiability of $\theta^\star$. Thus if  $\theta^\star$ is identifiable, then $\alpha \in \Ncal(\operator)^\perp$. 
\end{proof}

\begin{proof}[Proof for \Cref{thm: new-nuisance}]
According to the definition of $\Xi_0$ in  \Cref{assump: new-nuisance}, $\xi_0 \in \Xi_0$ if and only if it satisfies the first order condition of the optimization problem in \Cref{eq: delta-eq-2}: 
\begin{align*}
\Eb{[P\xi_0](T)[Ph](T)} - \Eb{\alpha(S)h(S)}  = 0, ~~ \forall h \in \Hcal,
\end{align*}
This is equivalent to 
\begin{align*}
\Eb{\alpha(S)h(S)} = \Eb{h(S)g_1(W)[P\xi_0](T)}, ~~ \forall h \in \Hcal.
\end{align*}
Therefore, 
\begin{align*}
\Pi_{\Hcal}\bracks{g_1(W)[P\xi_0](T) \mid S} = [P^\star P \xi_0](S) =  \alpha(S).  
\end{align*}
The steps above can be all reversed, which proves the asserted conclusion.  
\end{proof}

\begin{proof}[Proof for \Cref{lemma: minimum-norm}]
According to the definition of $\xi_0$ in \Cref{assump: new-nuisance}, we immediately have $P\xi_0 \in \Qcal_0$. 
Moreover, $P\xi_0 \in \Rcal(P)$.
Note the function class $\Qcal_0$ in \Cref{eq: Q0} can be written as $q_0 + \mathcal{N}\prns{P^\star} = q_0 + \Rcal(P)^\perp$ for any $q_0 \in \Qcal_0$. 
So $P\xi_0$ is exactly the minimum-norm element in $\Qcal_0$.
\end{proof}

\begin{proof}[Proof for \Cref{lemma: bias-product-general}]
According to \Cref{lemma: identifiability}, we have $\theta^\star = \Eb{m(W; h_0)}$ for any $h_0\in\Hcal$. 
Moreover, any $q_0 \in \Qcal_0$ satisfies that 
\begin{align*}
\Eb{g_1(W)q_0(T)h(S)} = \Eb{\alpha(S)h(S)} =  \Eb{m(W; h)}, ~~ \text{for any } h \in \Hcal.
\end{align*}
In particular, the above holds for $h = h_0 \in \Hcal$, namely,
\begin{align*}
\Eb{m(W; h_0)} = \Eb{g_1(W)q_0(T)h_0(S)}.
\end{align*}
It follows that 
\begin{align*}
\theta^\star 
   &= \Eb{q_0(T)g_1(W)h_0(S)} \\
   &= \Eb{q_0(T)\Eb{g_1(W)h_0(S) \mid T}} \\
   &= \Eb{q_0(T)\Eb{g_2(W) \mid T}} \\
   &= \Eb{q_0(T){g_2(W)}}.
\end{align*}
Finally, 
\begin{align*}
 &\Eb{m(W; h_0) +q_0(T)(g_2(W) - g_1(W)h_0(S))} \\
=& \Eb{q_0(T)g_2(W)} + \Eb{m(W; h_0) - q_0(T)g_1(W)h_0(S)} \\
=& \Eb{q_0(T)g_2(W)} = \theta^\star.
\end{align*}
The above also holds for $q_0 = P\xi_0$ with $\xi_0 \in \Xi_0$ since such $q_0$ is also in $\Qcal_0$. This recovers the conclusion in \Cref{thm:dr-ident}.

We now need to prove that $h_0 \in \Hcal_0$ and $q_0 \in \Qcal_0$ if and only if, for for any $h \in \Hcal$, $q \in \Lcal_2(T)$,
\begin{align}\label{eq: key-dr}
\Eb{\psi(W; h, q)} - \theta^\star = \Eb{g_1(W)\prns{ q(T) - q_0(T)}\prns{h(S) -  h_0(S)}},
\end{align}
since \Cref{eq: key-dr} gives 
\begin{align*}
\Eb{\psi(W; h, q)} - \theta^\star  
    &= \Eb{\Eb{g_1(W)\prns{h(S) - h_0(S)}\mid T}\prns{ q(T) - q_0(T)}} \\
    &= \Eb{\prns{P[h\prns{S}-h_0\prns{S}]}\prns{ q(T) - q_0(T)}} = \langle P(h-h_0), q-q_0\rangle. 
\end{align*}
and 
\begin{align*}
 \Eb{\psi(W; h, q)} - \theta^\star 
    &= \Eb{\prns{h(S) - h_0(S)}\Eb{g_1(W)\prns{ q(T) - q_0(T)} \mid S}} \\
    &= \Eb{\prns{h(S) - h_0(S)}\prns{P^\star[q(T) - q_0(T)]}} = \langle h-h_0, P^\star(q-q_0)\rangle. 
\end{align*}

We first prove that $h_0 \in \Hcal_0$ and $q_0 \in \Qcal_0$ implies \Cref{eq: key-dr}: for any $h \in \Lcal_2\prns{S}, q \in \Lcal_2\prns{T}$ and $h_0 \in \Hcal_0, q_0 \in \Qcal_0$, 
\begin{align*}
\Eb{\psi(W; h, q)} - \theta^\star 
    &= \Eb{\psi(W; h, q)} - \Eb{\psi(W; h_0, q_0)} \\
    &= \Eb{m(W; h- h_0)} + \Eb{q(T)(g_2(W) - g_1(W)h(S))} \\
    &= \Eb{\alpha(S)(h(S)- h_0(S))} + \Eb{q(T)(\Eb{g_2(W) \mid T} - g_1(W)h(S))} \\
    &= \Eb{g_1(W)q_0(T)(h(S)- h_0(S))} + \Eb{g_1(W)q(T)\prns{h_0(S) - h(S)}} \\
    &= \Eb{g_1(W)\prns{ q(T) - q_0(T)}\prns{h(S) -  h_0(S)}}.
\end{align*}

Next, we prove that if  \Cref{eq: key-dr} holds for any $h \in \Hcal$, $q \in \Lcal_2(T)$, then $h_0 \in \Hcal_0, q_0 \in \Qcal_0$. 
We note that by taking $h = h_0$ or $q = q_0$ in \Cref{eq: key-dr}, we can conclude that if either $h = h_0 \in \Hcal_0$ or $q = q_0 \in \Qcal_0$, then we have 
\begin{align*}
 \theta^\star  = \Eb{\psi(W; h, q)}.
\end{align*}
This means that $\theta^\star = \Eb{\psi(W; h_0, q_0)}$. Therefore, we have 
\begin{align*}
\Eb{\psi(W; h, q)} - \theta^\star  
    &= \Eb{\psi(W; h, q)} -  \Eb{\psi(W; h_0, q_0)} \\
    &= \Eb{m(W; h- h_0)} + \Eb{q(T)(g_2(W) - g_1(W)h(S))} - \Eb{q_0(T)(g_2(W) - g_1(W)h_0(S))}.
\end{align*}
On the one hand, if we set $q = q_0 \in \Qcal_0$, then we have 
\begin{align*}
0 = \Eb{\prns{\prns{\alpha(S)} - g_1(W)q_0(T)}\prns{h(S) - h_0(S)}}, ~~ \forall h \in \Hcal. 
\end{align*}
This means that $\alpha(S) = \Pi_{\Hcal}\bracks{q_0(T) \mid S}$, namely, $q_0 \in \Qcal_0$. 

On the other hand, if we set $h = h_0 \in \Hcal_0$, then we have 
\begin{align*}
\Eb{\prns{q(T)-q_0(T)}\prns{g_2(W) - g_1(W)h_0(S)}} = 0, \forall q \in \Lcal_2(T).
\end{align*}
This means that $\Eb{g_2(W) - g_1(W)h_0(S) \mid T} = 0$, namely, $h_0 \in \Hcal_0$. 

Therefore, $h_0 \in \Hcal_0$ and $q_0 \in \Qcal_0$ if and only if \Cref{eq: key-dr} holds  for any $h \in \Hcal$, $q \in \Lcal_2(T)$.

We can similarly show that $h_0 \in \Hcal_0$ and $q_0 \in \Qcal_0$ is sufficient and necessary for the Neyman orthogonality. That is, $h_0 \in \Hcal_0$ and $q_0 \in \Qcal_0$ if and only if for any $h \in \Hcal$, $q \in \Lcal_2(T)$, 
\begin{align*}
 \frac{\partial}{\partial t}  \Eb{\psi(W; h_0 + th, q_0)}\big\vert_{t = 0}  =  \frac{\partial}{\partial t}  \Eb{\psi(W; h_0, q_0  + tq)}\big\vert_{t = 0}  = 0.
 \end{align*}

To prove this, note that 
\begin{align*}
 \frac{\partial}{\partial t}  \Eb{\psi(W; h_0 + th, q_0)}\big\vert_{t = 0} 
    &= \Eb{\alpha(S)h(S)} - \Eb{g_1(W)q_0(T)h(S)}.
 \end{align*} 
It is equal to $0$ for any $h\in\Hcal$ if and only $\Pi_\Hcal\bracks{g_1(W)q_0(T) \mid S} = \alpha(S)$, namely, $q_0 \in \Qcal_0$. 

Moreover,  
\begin{align*}
 \frac{\partial}{\partial t}  \Eb{\psi(W; h_0, q_0  + tq)}\big\vert_{t = 0}  = \Eb{q(T)\prns{g_2(W) - g_1(S)h_0(S)}}.
 \end{align*} 
 It is equal to $0$ for any $q \in \Lcal_2(T)$ if and only if $\Eb{g_2(W) - g_1(S)h_0(S) \mid T} = 0$, namely, $h_0 \in \Hcal_0$. 
\end{proof}

\subsection{Results from \Cref{sec:generalized-framework}}

\begin{proof}[Proof of \Cref{lem:general-dr}]

First, we have
\begin{align*}
    \EE\Big[ \psi(W;h,q) \Big] &= \EE\Big[ m(W;h) + r(W;q) - G(W;h,q) \Big] \\
    &= \EE[m(W;h_0)] + \EE\Big[ m(W;h-h_0) + r(W;q) - G(W;h,q) \Big] \\
    &= \theta^\star + \EE\Big[ m(W;h-h_0) - G(W;h-h_0,q) + r(W;q)  - G(W;h_0,q) \Big] \\
    &= \theta^\star + \EE\Big[ m(W;h-h_0) - G(W;h-h_0,q) \Big] \\
    &= \theta^\star + \EE\Big[ m(W;h-h_0) - G(W;h-h_0,q_0) - G(W;h-h_0,q-q_0) \Big] \\
    &= \theta^\star - \EE\Big[ G(W;h-h_0,q-q_0) \Big] \,,
\end{align*}
where above we apply the fact that $\theta^\star = \EE[m(W;h_0)]$, as well as the definitions of $\Qcal_0$ and $\Hcal_0$.

For the second part of the lemma, applying the previous equality gives us
\begin{align*}
    \left. \frac{d}{dt} \right|_{t=0} \EE\Big[ \psi(W;h_0 + th, q_0 + tq) \Big] &= - \left. \frac{d}{dt} \right|_{t=0} \EE\Big[ G(W;th, tq)\Big] \\
    &= - \EE\Big[ G(W;h, q)\Big] \left. \frac{d}{dt} \right|_{t=0} t^2  \\
    &= 0 \,.
\end{align*}

Finally, for the third part of the lemma, applying the previous equality again, along with our continuous bi-linear assumption on $G$, gives us
\begin{align*}
    \Big| \EE[\psi(W;h,q)] - \theta^\star\Big| &= \Big| \EE[G(W;h-h_0,q-q_0)] \Big| \\
    &= \Big| \EE\Big[(h-h_0)(S) k(S,T) (q-q_0)(T)\Big] \Big| \,.
\end{align*}
Then, since $q - q_0 \in \Qcal$, and applying the definition of $P$, we can further simply the above by
\begin{align*}
    \Big| \EE\Big[(h-h_0)(S) k(S,T) (q-q_0)(T)\Big] \Big| &= \Big| \EE\Big[ \Pi_\Qcal\Big[ k(S,T)^\top (h - h_0)(S) \mid T \Big]^\top (q - q_0)(T)\Big] \Big| \\
    &= \Big| \EE\Big[ (P(h-h_0))(T)^\top (q - q_0)(T)\Big] \Big| \\
    &\leq \|P(h-h_0)\|_{2,2} \|q - q_0\|_{2,2} \,,
\end{align*}
where in the final step we apply Cauchy Schwartz. Alternatively, if we instead project on to $\Hcal$, symmetrical reasoning gives us
\begin{align*}
    \Big| \EE\Big[(h-h_0)(S) k(S,T) (q-q_0)(T)\Big] \Big| &= \Big| \EE\Big[ \Pi_\Hcal\Big[ k(S,T) (q - q_0)(T) \mid S \Big]^\top (h - h_0)(S)\Big] \Big| \\
    &= \Big| \EE\Big[ (P^\star(q-q_0))(S)^\top (h - h_0)(S)\Big] \Big| \\
    &\leq \|P^\star(q-q_0)\|_{2,2} \|h - h_0\|_{2,2} \,.
\end{align*}
Then, taking the minimum of these two bounds gives us our final required result.

\end{proof}

\section{Proofs for Minimax Estimation Theory}

Here we consider the proof of our minimax estimation theory, specifically \Cref{thm:h-estimator-bound,thm:q-estimator-bound}, and their generalizations \Cref{thm:h-estimator-bound-general,thm:q-estimator-bound-general} (presented in \Cref{sec:generalized-framework}).
Since the former two theorems are just special cases of the latter, we will only present explicit proofs below for the latter, and then the former immediately follow.

\subsection{Proof of \Cref{thm:h-estimator-bound-general}}

Let us define
\begin{equation*}
    \phi(h,q) = G(W;h,q) - r(W;q) - \frac{1}{2}q(T)^\top q(T) \,.
\end{equation*}
Then, the estimator we are studying is given by
\begin{equation*}
    \hat h_n = \argmin_{h \in \Hcal_n} \sup_{q \in \Qcal_n} \EE_n[ \phi(h,q) - \mu_n h(S)^2 ] - \gamma_n^q \|q\|_\Qcal^2 + \gamma_n^h \|h\|_\Hcal^2 \,.
\end{equation*}

First, note that the population version of the minimax objective trivially satisfies
\begin{align*}
    \sup_{q \in \Qcal} \EE[\phi(h,q)] &= \sup_{q \in \Qcal} \EE\Big[ \Big( k(S,T)^\top h(S) - \beta(T) \Big)^\top q(T) - \frac{1}{2} q(T)^\top q(T) \Big] \\
    &= \frac{1}{2} \Big\| \Pi_\Qcal\Big[ k(S,T)^\top h(S) - \beta(T) \mid T \Big] \Big\|_{2,2}^2 \,,
\end{align*}
for any $h \in \Hcal$. Motivated by this, let us define the population projected loss objective
\begin{equation*}
    J(h) = \frac{1}{2} \Big\| \Pi_\Qcal\Big[ k(S,T)^\top h(S) - \beta(T) \mid T \Big] \Big\|_{2,2}^2 = \sup_{q \in \Qcal} \EE[\phi(h,q)] \,.
\end{equation*}
In addition, we define the weak norm of interest
\begin{align*}
    \|h - h'\|_w &= \Big\| \Pi_\Qcal\Big[ k(S,T)^\top (h - h')(S) \mid T \Big] \Big\|_{2,2} \\
    &= \Big\| P(h - h') \Big\|_{2,2} \,,
\end{align*}
for any $h,h' \in \Hcal$. Note that according to these definitions we trivially have $J(h_0) = 0$ for all $h \in h_0$, and $\|h - h_0\|_w = \|h - h'_0\|_w$ for any $h \in \Hcal$ and $h_0,h'_0 \in \Hcal_0$. That is, without loss of generality we can prove the required bound for $h_0 = h^\dagger$.

Then our goal is to provide a high-probability bound on $\|\hat h_n - h^\dagger\|_w$. We will provide a high-level overview of the proof in terms of some intermediate lemmas below, and then relegate the proof of these intermediate lemmas to the end of the subsection.

\subsubsection*{Bounding Weak Norm by Population Objective Sub-optimality}

First, we can provide the following lemma which follows by strong convexity.
\begin{lemma}
\label{lem:primal-strong-convexity-bound}
    For any $h \in \Hcal$ we have
    \begin{equation*}
        \frac{1}{2} \|h - h^\dagger\|_w^2 \leq J(h) - J(h^\dagger)
    \end{equation*}
\end{lemma}

This lemma is the first step for proving fast rates, since it allows us to bound
\begin{equation*}
    \frac{1}{2} \|\hat h_n - h^\dagger\|_w^2 \leq \sup_{q \in \Qcal} \EE[\phi(\hat h_n,q)] - \sup_{q \in \Qcal} \EE[\phi(h^\dagger,q)] \,.
\end{equation*}
This is powerful, since $\hat h_n$ minimizes the empirical version of the RHS, which allows us to apply a empirical risk minimization analysis to this minimax loss.

\subsubsection*{High-Probability Bound}

Next, we establish an important high-probability event, which will be applied in various lemmas in the rest of the proof sketch. For these bounds, and our subsequent theory, we define 
\begin{align*}
    \Omega^q(q') &= \max( \|q\|_\Qcal, M)  \\
    \Omega^h(h') &= \max( \|h\|_\Hcal, M) \,,
\end{align*}
for all $q' \in \Qcal$ and $h' \in \Hcal$ respectively.

\begin{lemma}
\label{lem:primal-high-prob}

    Let $\epsilon_n = r_n + c_2 \sqrt{\log(c_1/\zeta)}$, where the constants $c_1$ and $c_2$ are defined analogously as in \citet[Lemma 11]{foster2019orthogonal}, and let
    \begin{equation*}
        \Phi_n(q) = G(W;h^\dagger, q) - r(W;q) - \frac{1}{2} q(T)^\top q(T) \,,
    \end{equation*}
    for any $q \in \Qcal_n$. Then, with probability at least $1-3\zeta$, we have
    \begin{equation*}
        \Big| (\EE_n - \EE) \Big[ \Phi_n(q) \Big] \Big| \leq 54 \Omega^q(q) \epsilon_n \|q\|_{2,2} + 54 \Omega^q(q)^2 \epsilon_n^2 \,,
    \end{equation*}
    for every $q \in \Qcal_n$, as well as
    \begin{equation*}
        \Big| (\EE_n - \EE) \Big[ G(W;h - h^\dagger, q) \Big] \Big| \leq 36 \epsilon_n \|q\|_{2,2} + 18 \Omega^q(q)^2 \epsilon_n^2 + 18 \Omega^h(h-h^\dagger)^2 \epsilon_n^2 \,,
    \end{equation*}
    for every $h \in \Hcal_n$ and $q \in \Qcal_n$.

\end{lemma}

\subsubsection*{Bounding the Saddle-Point Estimate}

Next, we provide a lemma on the quality of the interior maximization problem, which will be used throughout the proofs of the rest of the lemmas. Let us define
\begin{align*}
    q_n(h) &= \argmax_{q \in \Qcal_n} \EE_n[\phi(h,q)] - \gamma_n^q \|q\|_\Qcal^2 \\
    q_0(h) &= \argmax_{q \in \Qcal} \EE[\phi(h,q)] \,,
\end{align*}
for any $h \in \Hcal$, with tie-breaking performed arbitrarily. Note that it is trivial to verify that $q_0(h)(T) = \Pi_\Qcal[k(S,T)^\top h(S) - \beta(T) \mid T]$ and $\|q_0(h)(T)\| = \|h - h_0\|_w$.
Then, the following lemma ensures the quality of the interior maximization.

\begin{lemma}
\label{lem:primal-interior-maximization}
    Under the same high-probability event of \Cref{lem:primal-high-prob}, for every $h \in \Hcal_n$ we have
    \begin{align*}
        \|q_n(h)\|_{2,2} &\leq \Big(377 \Omega^q(q_n(h)) + 26 \Omega^q(\Pi_n q_0(h)) + 12 \Omega^h(h) \Big) \epsilon_n \\
        &\qquad + 14 \delta_n + 14 \|h-h^\dagger\|_w + 2 (\gamma_n^q)^{1/2} \Omega^q(\Pi_n q_0(h)) \,.
    \end{align*}
\end{lemma}

\subsubsection*{Reduction to Stochastic Equicontinuity Problem}

Given the previous lemmas, we can now provide the following result, which reduces the difference between the population objective for $\hat h_n$ and $h^\dagger$ to a stochastic-equicontinuity term.

\begin{lemma}
\label{lem:primal-minimax-erm-bound}
    Under the high-probability event of \Cref{lem:primal-high-prob}, we have
    \begin{align*}
        J(\hat h_n) - J(h^\dagger) &\leq (\EE - \EE_n)\Big[\phi(\hat h_n,\Pi_n q_0(\hat h_n)) - \phi(h^\dagger,q_n(\Pi_n h^\dagger)\Big] + 2 \delta_n \|\hat h_n - h^\dagger\|_w \\
        &\qquad + \Big( 14775 \Omega^q(q_n(\Pi_n h^\dagger))^2 + 949 \Omega^q(\Pi_n q_0(\Pi_n h^\dagger)) \Big) \epsilon_n^2 + 741 \delta_n^2 \\
        &\qquad + \Big( 37 \Omega^q(\Pi_n q_0(\Pi_n h^\dagger))^2 + 2 \Omega^q(\Pi_n q_0(\hat h_n))^2 \Big) \gamma_n^q  + 2 M^2 \gamma_n^h \\
        &\qquad + \mu_n \EE_n[(\Pi_n h^\dagger)(S)^2 - \hat h(S)^2] - \gamma_n^q \Omega^q(q_n(\Pi_n h^\dagger))^2 - \gamma^h_n \Omega^h(\hat h_n)^2
    \end{align*}
\end{lemma}

\subsubsection*{Bounding the Stochastic Equicontinuity Term using Localization}

Next, we focus on the ``stochastic equicontinuity''-style term in bound from the previous lemma. By some simple algebra, we can see that
\begin{align*}
    &\phi(\hat h_n, \Pi_n q_0(\hat h_n)) - \phi(h^\dagger, q_n(h^\dagger))  \\
    &= G(W;\hat h_n-h^\dagger, q_0(h^\dagger)) + G(W;\hat h_n-h^\dagger,\Pi_n q_0(\hat h_n) - q_0(h^\dagger)) \\
    &\quad + G(W;h^\dagger, \Pi_n q_0(\hat h_n) - q_0(h^\dagger)) - G(W; h^\dagger, q_n(h^\dagger) - q_0(h^\dagger)) \\
    &\quad - r(W;\Pi_n q_0(\hat h_n) - q_0(h^\dagger)) + r(W;q_n(h^\dagger) - q_0(h^\dagger)) \\
    &\quad - \frac{1}{2} \Big( \Pi_n q_0(\hat h_n) - q_0(h^\dagger) \Big)(T)^\top \Big( \Pi_n q_0(\hat h_n) - q_0(h^\dagger) \Big)(T) \\
    &\quad + \frac{1}{2} \Big( q_n(h^\dagger) - q_0(h^\dagger) \Big)(T)^\top \Big( q_n(h^\dagger) - q_0(h^\dagger) \Big)(T)  \\
    &\quad - q_0(h^\dagger)(T)^\top \Big( \Pi_n q_0(\hat h_n) - q_0(h^\dagger) \Big)(T)  + q_0(h^\dagger)(T)^\top \Big( q_n(h^\dagger) - q_0(h^\dagger) \Big)(T) \,.
\end{align*}
Furthermore, since $q_0(h^\dagger) = 0$ by construction, the first and last two terms vanish. Instantiating this reasoning, and applying \Cref{lem:primal-high-prob}, we can easily derive the following lemma.

\begin{lemma}
\label{lem:primal-stochastic-equicontinuity-bound}
    Under the high-probability event of \Cref{lem:primal-high-prob}, we have
    \begin{align*}
        &\Big| (\EE_n - \EE) \Big[\phi(\hat h_n, \Pi_n q_0(\hat h_n)) - \phi(h^\dagger, q_n(\Pi_n h^\dagger)\Big] \Big| \\
        &\leq 90 \epsilon_n \Omega^q(\Pi_n q_0(\hat h_n)) \|\hat h_n - h^\dagger\|_w + 801 \delta_n^2 + 54 \gamma_n^q \Omega^q(\Pi_n q_0(\Pi_n h^\dagger))^2 \\
        &\qquad + \Big( 22572 \Omega^q(q_n(\Pi_n h^\dagger))^2 + 702 \Omega^q(\Pi_n q_0(\Pi_n h^\dagger))^2 + 117 \Omega^q(\Pi_n q_0(\hat h_n))^2 + 18 \Omega^h(\hat h_n - h^\dagger)^2 \Big) \epsilon_n^2
    \end{align*}
\end{lemma}

\subsubsection*{First Projected Norm Bound}

Putting together the results of the above intermediate lemmas, we get
\begin{align}
    &\frac{1}{2} \|\hat h_n - h^\dagger\|_w^2 \nonumber \\
    &\leq \Big( 90 \epsilon_n \Omega^q(\Pi_n q_0(\hat h_n)) + 2 \delta_n \Big) \|\hat h_n - h^\dagger\|_w + 1542 \delta_n^2 \nonumber \\
    &\qquad + \Big( 37347 \Omega^q(q_n(\Pi_n h^\dagger))^2 + 1651 \Omega^q(\Pi_n q_0(\Pi_n h^\dagger))^2 + 117 \Omega^q(\Pi_n q_0(\hat h_n))^2 + 18 \Omega^h(\hat h_n - h^\dagger)^2 \Big) \epsilon_n^2 \nonumber \\
    &\qquad + \Big( 91 \Omega^q(\Pi_n q_0(\Pi_n h^\dagger))^2 + 2 \Omega^q(\Pi_n q_0(\hat h_n))^2 \Big) \gamma_n^q + 2 M^2 \gamma_n^h \nonumber \\
    &\qquad + \mu_n \EE_n[(\Pi_n h^\dagger)(S)^2 - \hat h(S)^2] - \gamma_n^q \Omega^q(q_n(\Pi_n h^\dagger))^2 - \gamma^h_n \Omega^h(\hat h_n)^2
    \label{eq:full-bound-1}
\end{align}

Now, under the assumptions of our first bound where $\|h\|_\Hcal, \|h-h^\dagger\|_\Hcal \leq M$ and $\|q\|_\Qcal \leq M$ for all $h \in \Hcal_n$ and $q \in \Qcal_n$, the above bound can be simplified to
\begin{align*}
    &\frac{1}{2} \|\hat h_n - h^\dagger\|_w^2 \nonumber \\
    &\leq \Big( 90 M \epsilon_n + 2 \delta_n \Big) \|\hat h_n - h^\dagger\|_w + 1542 \delta_n^2 + 39133 M^2 \epsilon_n^2  + 93 M^2 \gamma_n^q + 2 M^2 \gamma_n^h + \mu_n \,.
\end{align*}

Applying \Cref{lem:quadratic-inequality} to this gives us
\begin{align*}
    \|\hat h_n - h^\dagger\|_w &\leq 460 M \epsilon_n + 60 \delta_n + 14 M (\gamma_n^q)^{1/2} + 2 M (\gamma_n^h)^{1/2} + 2 \mu_n^{1/2} \,,
\end{align*}
which is our first promised bounds.

\subsubsection*{Second Projected Norm Bound}

For our latter bound, first note that, following our theorem's premises, we have
\begin{align*}
    \|\Pi_n q_0(h)\|_\Qcal &= \|\Pi_n P(h - h^\dagger)\|_\Qcal \\
    &\leq L \|h - h^\dagger\|_\Hcal \\
    &\leq L \|h\|_\Hcal + M L \,,
\end{align*}
and therefore, given that $M,L \geq 1$, we have
\begin{align*}
    \Omega^q(\Pi_n q_0(h)) &\leq \max(L \|h\|_\Hcal + ML, M) \\
    &= L \|h\|_\Hcal + ML \\
    &\leq 2 L \Omega^h(h) \,,
\end{align*}
which applies for all $h \in \Hcal_n$. In addition, by assumption we have $\|\Pi_n h^\dagger\|_\Hcal \leq M$, so therefore $\Omega^h(\Pi_n h^\dagger) \leq M$, $\|\Pi_n q_0(\Pi_n h^\dagger)\|_\Qcal \leq 2LM$, and $\Omega^q(\Pi_n q_0(\Pi_n h^\dagger)) \leq 2LM$. Then, we easily have from \Cref{eq:full-bound-1} that
\begin{align*}
    &\frac{1}{2} \|\hat h_n - h^\dagger\|_w^2 + \frac{1}{2} \gamma_n^h \Omega^h(\hat h_n)^2 \\
    &\leq \Big( 90 \epsilon_n \Omega^q(\Pi_n q_0(\hat h_n)) + 2 \delta_n \Big) \|\hat h_n - h^\dagger\|_w + 1542 \delta_n^2 \\
    &\qquad + \Big( 37347 \Omega^q(q_n(\Pi_n h^\dagger))^2 + 7144 L^2 \Omega^h(\hat h_n)^2 \Big) \epsilon_n^2 + 186 L^2 \Omega^h(\hat h_n)^2 \gamma_n^q + 2 M^2 \gamma_n^h \\
    &\qquad + \mu_n \EE_n[(\Pi_n h^\dagger)(S)^2 - \hat h(S)^2] - \gamma_n^q \Omega^q(q_n(\Pi_n h^\dagger))^2 - \frac{1}{2} \gamma^h_n \Omega^h(\hat h_n)^2 \,.
\end{align*}
Furthermore, by the AM-GM inequality, we have
\begin{equation*}
    \Big( 90 \epsilon_n \Omega^q(\Pi_n q_0(\hat h_n)) + 2 \delta_n \Big) \|\hat h_n - h^\dagger\|_w \leq \frac{1}{4} \|\hat h_n - h^\dagger\|_w^2 + \Big( 90 \epsilon_n \Omega^q(\Pi_n q_0(\hat h_n)) + 2 \delta_n \Big)^2 \,,
\end{equation*}
so therefore we have
\begin{align*}
    &\frac{1}{4} \|\hat h_n - h^\dagger\|_w^2 + \frac{1}{2} \gamma_n^h \Omega^h(\hat h_n)^2 \\
    &\leq 1542 \delta_n^2 + \Big( 37347 \Omega^q(q_n(\Pi_n h^\dagger))^2 + 71944 L^2 \Omega^h(\hat h_n)^2 \Big) \epsilon_n^2 + 186 L^2 \Omega^h(\hat h_n)^2 \gamma_n^q + 2 M^2 \gamma_n^h \\
    &\qquad + \mu_n \EE_n[(\Pi_n h^\dagger)(S)^2 - \hat h(S)^2] - \gamma_n^q \Omega^q(q_n(\Pi_n h^\dagger))^2 - \frac{1}{2} \gamma^h_n \Omega^h(\hat h_n)^2 \,.
\end{align*}
Now, as long as we have
\begin{equation*}
    \gamma_n^q \geq 37347 \epsilon_n^2 \,,
\end{equation*}
and
\begin{equation*}
    \gamma_n^h \geq L^2 \Big( 143888 \epsilon_n^2 + 372 \gamma_n^q \Big) \,,
\end{equation*}
the above simplifies to
\begin{align}
    &\frac{1}{4} \|\hat h_n - h^\dagger\|_w^2 + \frac{1}{2} \gamma_n^h \Omega^h(\hat h_n)^2 \nonumber \\
    &\leq 1542 \delta_n^2 + 2 M^2 \gamma_n^h + \mu_n \EE_n[(\Pi_n h^\dagger)(S)^2 - \hat h(S)^2] \,.
    \label{eq:full-bound-2}
\end{align}

It immediately follows from the above that
\begin{equation*}
    \|\hat h_n - h^\dagger\|_w \leq 79 \delta_n + 3 M (\gamma_n^h)^{1/2} + 2 \mu_n \,,
\end{equation*}
which gives us our second required bound. Note that the $M \epsilon_n$ and $M (\gamma_n^q)^{1/2}$ terms in the bound are implicit given the minimum size requirement on $\gamma_n^h$.

\subsubsection*{Strong Norm Consistency}

Finally, we will argue that $\hat h_n$ not only converges to $h_0$ under weak norm (for any $h_0 \in \Hcal_0$), but converges to $h^\dagger$ in strong norm. 

Let $\zeta_n$ be the sequence defined in the problem statement, and define $\tilde\epsilon_n = r_n + c_2 \sqrt{\log(c_1 / \zeta_n) / n}$. Recall that our required assumptions on $\gamma_n^q$ and $\gamma_n^h$ hold with $\zeta_n$, and that $\log(c_1 / \zeta) / n = o(\mu_n)$. Now, from the above bounds, under our first set of assumptions we have
\begin{align*}
    \mu_n \Big( \|\hat h_n \|_{n,2}^2 - \|\Pi_n h^\dagger \|_{n,2}^2 \Big) &\preccurlyeq  \|\hat h_n - h^\dagger\|_w \Big( M \tilde\epsilon_n + \delta_n \Big) \\
    &\qquad + M^2 \tilde\epsilon_n^2 + \delta_n^2 + M^2 \gamma_n^q + M^2 \gamma_n^h \\
    &\preccurlyeq M^2 \tilde\epsilon_n^2 + \delta_n^2 + M^2 \gamma_n^q + M^2 \gamma_n^h \,,
\end{align*}
with probability at least $1-\zeta_n$. Similarly, under our second set of assumptions, we have
\begin{align*}
    \mu_n \Big( \|\hat h_n \|_{n,2}^2 - \|\Pi_n h^\dagger \|_{n,2}^2 \Big) &\preccurlyeq  \delta_n^2 + M^2 \gamma_n^h \,,
\end{align*}
with probability at least $1-\zeta_n$. Now, in either case, since by assumption we have that $r_n = o(\mu_n^{1/2})$, $\delta_n = o(\mu_n^{1/2})$, $\gamma_n^q = o(\mu_n)$, and $\gamma_n^h = o(\mu_n)$, and by the above choice of $\zeta_n$, we have
\begin{equation*}
     \|\hat h_n \|_{n,2}^2 \leq \|\Pi_n h^\dagger \|_{n,2}^2 + o(1) \,,
\end{equation*}
with probability at least $1-\zeta_n$. Therefore, since $\zeta_n = o(1)$, this implies that
\begin{equation*}
    \|\hat h_n \|_{n,2}^2 \leq \|\Pi_n h^\dagger \|_{n,2}^2 + o_p(1) \,.
\end{equation*}
This bound then motivates the following lemma, which jutifies that the $L_2$ norm of $\hat h_n$ is asymptotically bounded by the minimum solution $L_2$ norm $\|h^\dagger\|^2_{n,2}$.
\begin{lemma}
\label{lem:h-norm-convergence}
    Under the above assumptions, we have
    \begin{equation*}
        \|\hat h_n\|_{2,2}^2 \leq \| \|h^\dagger\|_{2,2}^2 + o_p(1) \,.
    \end{equation*}
\end{lemma}
That is, this suggests that $\hat h_n$ cannot converge to any element $h_0$ other than the minimum-norm element $h^\dagger$. Specifically, the following lemma ensures that this convergence does in fact occur.

\begin{lemma}
\label{lem:p-convergence}
    Suppose $r_n = o(\mu_n)$ and $\delta_n = o(\mu_n)$. Then, we have
    \begin{equation*}
        \|\hat h_n - h^\dagger\|_2 \to 0 \,,
    \end{equation*}
    in probability.
\end{lemma}

\subsubsection*{Proofs of Intermediate Lemmas}

Finally, we list the proofs of the above intermediate lemmas

\begin{proof}[Proof of \Cref{lem:primal-strong-convexity-bound}]

Consider some arbitrary $h \in \Hcal$, and define
\begin{equation*}
    L(t) = J(h^\dagger + t(h - h^\dagger)) \,.
\end{equation*}
Then, the first two derivatives of $L$ are given by
\begin{equation*}
    L'(t) = \EE\Big[ \Pi_\Qcal \Big[ k(S,T)^\top(h - h^\dagger)(S) \mid T \Big]^\top \Pi_\Qcal \Big[ k(S,T)^\top \Big(h^\dagger + t(h-h^\dagger)\Big)(S) - \beta(T) \mid T \Big] \Big] \,
\end{equation*}
and
\begin{equation*}
    L''(t) = \Big\| \Pi_\Qcal \Big[ k(S,T)(h - h^\dagger)(S) \mid T \Big] \Big\|_{2,2}^2 = \|h - h^\dagger\|_w^2 \,.
\end{equation*}
Note that the first derivative follows since projections onto Hilbert spaces are linear.
Furthermore, we have $L(0) = J(h^\dagger)$, and $L(1) = J(h)$. Also, let $q' = \Pi_\Qcal[ k(S,T)^\top (h - h^\dagger)(S) \mid T ]$, we have
\begin{align*}
    L'(0) &= \EE\Big[ \Pi_\Qcal \Big[ k(S,T)^\top (h - h^\dagger)(S) \mid T \Big] \Pi_\Qcal \Big[ k(S,T)^\top h^\dagger(S) - \beta(T) \mid T \Big] \Big] \\
    &= \EE\Big[ \Big( k(S,T)^\top h^\dagger(S) - \beta(T) \Big) \Pi_\Qcal \Big[ k(S,T)^\top(h - h^\dagger)(S) \mid T \Big] \Big] \\
    &= \EE[G(W;h^\dagger,q') - r(W;q')] \\
    &= 0 \,,
\end{align*}
where  the final two equalities follow since $q' \in \Qcal$ by construction and  $h^\dagger \in \Hcal_0$.
Therefore, noting that $L''(t)$ does not depend on $t$, strong convexity gives us
\begin{align*}
    &L(1) \geq L(0) + L'(0) + \frac{1}{2} \inf_{t \in [0,1]} L''(t) \\
    &\iff \frac{1}{2} \|h - h^\dagger\|_w^2 \leq J(h) - J(h^\dagger) \,,
\end{align*}
as required.

\end{proof}

\begin{proof}[Proof of \Cref{lem:primal-high-prob}]

First, let us establish our bound for $\Psi_n(q)$. Note that for any $q$ we can trivially bound
\begin{equation}
\label{eq:primal-high-prob-psi-bound}
    \Big| (\EE_n - \EE) \Big[ \Phi_n(q) \Big] \Big| \leq \Big| (\EE_n - \EE) \Big[ G(W;h^\dagger,q) - r(W;q) \Big] \Big| + \frac{1}{2} \Big| (\EE_n - \EE) \Big[ q(T)^\top q(T)  \Big] \Big| \,.
\end{equation}
We will bound each of the terms in the right hand side of \Cref{eq:primal-high-prob-psi-bound} separately. For the first term, it easily follows from our boundedness assumptions that $G(W;h^\dagger,q) - r(W;q)$ is $2$-Lipschitz in $q$ under $L_2$ norm, and $\|G(W;h^\dagger,q) - r(W;q)\|_\infty \leq 2$. In addition, note that by construction $q / \Omega^q(q)$ is in $\{q \in \starcls(\Qcal_n) : \|q\|_\Qcal \leq 1\}$ for all $q \in \Qcal_n$, since $\Omega^q(q) \geq 1$, and $\|q/\Omega^q(q)\|_\Qcal \leq 1$. Therefore, applying \citet[Lemma 11]{foster2019orthogonal} with \cref{assum:complexity-h-general} gives us
\begin{align*}
    \Big| (\EE_n - \EE) \Big[ G(W;h^\dagger,q) - r(W;q) \Big] \Big| &= 2 \Omega^q(q) \Big| (\EE_n - \EE) \Big[ \frac{1}{2} \Big( G(W;h^\dagger,q/\Omega^q(q)) - r(W;q/\Omega^q(q)) \Big) \Big] \Big| \\
    &\leq 2 \Omega^q(q) \Big( 18 \epsilon_n \|q/\Omega^q(q)\|_{2,2} + 18 \epsilon_n^2 \Big) \\
    &= 36 \epsilon_n \|q\|_{2,2} + 36 \Omega^q(q) \epsilon_n^2 \,,
\end{align*}
which holds uniformly for all $q \in \Qcal_n$ under an event with probability at least $1-\zeta$.

For the second term in the right hand side of \Cref{eq:primal-high-prob-psi-bound}, we can follow similar reasoning. However, we have to bound slighly differently since $q(T)^\top q(T)$ is quadratic in $q$ rather than linear. By our boundedness assumptions $\frac{1}{2} q(T)^\top q(T)$ is $1$-Lipschitz in $q$ under $L_2$ norm, and is uniformly bounded by $1$. Therefore, applying \citet[Lemma 11]{foster2019orthogonal} with \cref{assum:complexity-h-general} again gives us
\begin{align*}
    \Big| (\EE_n - \EE) \Big[ \frac{1}{2} q(T)^\top q(T) \Big] \Big| &= \Omega^q(q)^2 \Big| (\EE_n - \EE) \Big[ \frac{1}{2} (q(T) / \Omega^q(q))^\top (q(T) / \Omega^q(q)) \Big] \Big| \\
    &\leq \Omega^q(q)^2 \Bigg( 18 \epsilon_n \|q/\Omega^q(q)\|_{2,2} + 18 \epsilon_n^2 \Bigg) \\
    &\leq 18 \Omega^q(q) \epsilon_n \|q\|_{2,2} + 18 \Omega^q(q)^2 \epsilon_n^2 \,,
\end{align*}
which holds uniformly for all $q \in \Qcal_n$ under a separate event with probability at least $1-\zeta$.

Putting the previous two bounds together with a union bound, and noting htat $1 \leq \Omega^q(q) \leq \Omega^q(q)^2$, under an event with probability at least $1-2\zeta$ we have
\begin{equation*}
    \Big| (\EE_n - \EE) \Big[ \Phi_n(q) \Big] \Big| \leq 54 \Omega^q(q) \epsilon_n \|q\|_{2,2} + 54 \Omega^q(q)^2 \epsilon_n^2 \,,
\end{equation*}
uniformly over $q \in \Qcal_n$.

Next, we will deal with bounding the $G(W;h-h^\dagger,q)$ term. Applying similar reasoning as above, we have that $G(W;(h-h^\dagger)/\Omega^h(h-h^\dagger), q/\Omega^q(q)$ is an element of our second critical-radius bounded set defined in \cref{assum:complexity-h-general} for all $h \in \Hcal_n$ and $q \in \Qcal_n$. In addition $G(W;h-h^\dagger,q)$ is uniformly bounded by $2$ for all such $h$ and $q$. Therefore, we can bound
\begin{align*}
    \Big| (\EE_n - \EE) \Big[ G(W;h-h^\dagger,q) \Big] \Big| &\leq 2 \Omega^h(h-h^\dagger) \Omega^q(q) \Big| (\EE_n - \EE) \Big[ \frac{1}{2} G(W;(h-h^\dagger)/\Omega^h(h-h^\dagger),q/\Omega^q(q)) \Big] \Big| \\
    &\leq 2 \Omega^h(h-h^\dagger) \Omega^q(q) \Big( 18 \epsilon_n \|\frac{1}{2} G(W;(h-h^\dagger)/\Omega^h(h-h^\dagger),q/\Omega^q(q))\|_2 + 18 \epsilon_n^2 \Big) \\
    &\leq 18 \epsilon_n \|G(W;h-h^\dagger,q)\|_2 + 36 \Omega^h(h-h^\dagger) \Omega^q(q) \epsilon_n^2 \\
    &\leq 36 \epsilon_n \|q\|_{2,2} + 18 \Omega^q(q)^2 \epsilon_n^2 + 18 \Omega^h(h-h^\dagger)^2 \epsilon_n^2 \,,
\end{align*}
which holds uniformly over $h \in \Hcal_n$ and $q \in \Qcal_n$ under a third event with probability at least $1-\zeta$.
Note that in the final inequality we apply our boundedness assumption on $q \mapsto G(W;h,q)$ over $L_2$, as well as the AM-GM inequality.

Finally, putting everything together, and applying a union bound, we have our two required bounds under an event with probability at least $1-3\zeta$.

\end{proof}

\begin{proof}[Proof of \Cref{lem:primal-interior-maximization}]

First, we note that 
\begin{align*}
    &\frac{1}{2}\|q_n(h) - q_0(h)\|_{2,2}^2 - \frac{1}{2}\|\Pi_n q_0(h) - q_0(h)\|_{2,2}^2 \\
    &= \EE\Big[ \Big( q_0(h)(T)^\top \Pi_n q_0(h)(T) - \frac{1}{2} \Pi_n q_0(h)(T)^\top \Pi_n q_0(h)(T) \Big) \\
    &\qquad - \Big( q_0(h)(T)^\top q_n(h)(T) - \frac{1}{2} q_n(h)(T)^\top q_n(h)(T) \Big) \Big] \\
    &= \EE\Big[ \phi(h, \Pi_n q_0(h)) - \phi(h, q_n(h)) \Big] \,,
\end{align*}
where in the second equality we apply the fact that for any $q \in \Qcal$ and $h \in \Hcal$ we have
\begin{align*}
    \EE\Big[q(T)^\top q_0(h)(T) - \frac{1}{2} q(T)^\top q(T)\Big] &= \EE\Big[q(T)^\top \Pi_\Qcal[k(S,T)^\top h(S) - \beta(T) \mid T] - \frac{1}{2} q(T)^\top q(T)\Big] \\
    &= \EE\Big[[h(S)^\top k(S,T) q(T) - \beta(T)^\top q(T) - \frac{1}{2} q(T)^\top q(T)\Big] \\
    &= \EE\Big[G(W;h,q) - r(W;q) - \frac{1}{2} q(T)^\top q(T)\Big] \\
    &= \EE\Big[\phi(h,q)\Big] \,.
\end{align*}

Furthermore, by the optimality of $q_n(h)$ for the empirical interior supremum problem at $h$ over $\Qcal_n$, and the fact that $\Pi_n q_0(h)\in\Qcal_n$, we have 
\begin{align*}
    \EE_n\Big[ \phi(h, q_n(h)) \Big] - \gamma_n^q \|q_n(h)\|_\Qcal^2 \geq \EE_n\Big[ \phi(h, \Pi_n q_0(h)) \Big] - \gamma_n^q \| \Pi_n q_0(h)\|_\Qcal^2 \,, ~~ h \in \Hcal_n \,.
\end{align*}
Thus we have 
\begin{align*}
    &\frac{1}{2}\|q_n(h) - q_0(h)\|_{2,2}^2 - \frac{1}{2}\|\Pi_n q_0(h) - q_0(h)\|_{2,2}^2 \\
     &= \EE\Big[ \phi(h, \Pi_n q_0(h)) - \phi(h, q_n(h)) \Big] \\
    &\leq (\EE_n - \EE)\Big[ \phi(h, q_n(h)) - \phi(h, \Pi_n q_0(h)) \Big]  + \gamma_n^q \Big( \|\Pi_n q_0(h)\|_\Qcal^2 - \|q_n(h)\|_\Qcal^2 \Big) \\
    &\leq (\EE_n - \EE)\Big[ \Phi_n(q_n(h)) - \Phi_n(\Pi_n q_0(h)) + G(W, h - h^\dagger, q_n(h)) - G(W, h - h^\dagger, \Pi_n q_0(h)) \Big]  \\
    &\qquad + \gamma_n^q \Big( \|\Pi_n q_0(h)\|_\Qcal^2 - \|q_n(h)\|_\Qcal^2 \Big) \,,
\end{align*}
where $\Phi_n$ is defined as in the statement of \Cref{lem:primal-high-prob}.

Now, since by assumption $h \in \Hcal_n$,
under the high probability event of \Cref{lem:primal-high-prob}, along with the AM-GM inequality, we can reduce the above to
\begin{align*}
    &\frac{1}{2}\|q_n(h) - q_0(h)\|_{2,2}^2 - \frac{1}{2}\|\Pi_n q_0(h) - q_0(h)\|_{2,2}^2 \\
    &\leq 54 \epsilon_n \Big( \Omega^q(q_n(h)) \|q_n(h)\|_{2,2} + \Omega^q(\Pi_n q_0(h)) \|\Pi_n q_0(h)\|_{2,2} \Big) + 36 \epsilon_n \Big( \|q_n(h)\|_{2,2} + \|\Pi_n q_0(h)\|_{2,2} \Big) \\
    &\qquad + 54 \epsilon_n^2 \Big( \Omega^q(q_n(h))^2 + \Omega^q(\Pi_n q_0(h))^2  \Big) + 36 \epsilon_n^2 \Omega^h(h-h^\dagger)^2 \\
    &\qquad + 18 \epsilon_n^2 \Omega^q(q_n(h))^2 + 18 \epsilon_n^2 \Omega^q(\Pi_n q_0(h))^2 + \gamma_n^q \Big( \|\Pi_n q_0(h)\|_\Qcal^2 - \|q_n(h)\|_\Qcal^2 \Big) \\
    &\leq \epsilon_n \Big( 90 \Omega^q(q_n(h)) \|q_n(h)\|_{2,2} + 90 \Omega^q(\Pi_n q_0(h)) \|\Pi_n q_0(h)\|_{2,2} \Big) \\
    &\qquad + \epsilon_n^2 \Big( 72 \Omega^q(q_n(h))^2 + 72 \Omega^q(\Pi_n q_0(h))^2 + 36 \Omega^h(h-h^\dagger)^2  \Big) + \gamma_n^q \Big( \|\Pi_n q_0(h)\|_\Qcal^2 - \|q_n(h)\|_\Qcal^2 \Big)  \,.
\end{align*}

Next, noting that $(a+b)^2 \leq 2 a^2 + 2b^2$ for $a,b \geq 0$, under the above high-probability event we have
\begin{align*}
    \frac{1}{2} \|q_n(h)\|_{2,2}^2 &\leq \|q_n(h) - q_0(h) \|_{2,2}^2 + \|q_0(h)\|_{2,2}^2 \\
    &= 2 \Big( \frac{1}{2} \|q_n(h) - q_0(h) \|_{2,2}^2 - \frac{1}{2}\|\Pi_n q_0(h) - q_0(h)\|_{2,2}^2 \Big) \\
    &\qquad + \|\Pi_n q_0(h) - q_0(h) \|_{2,2}^2 + \|q_0(h)\|_{2,2}^2 \\
    &\leq \epsilon_n \Big( 180 \Omega^q(q_n(h)) \|q_n(h)\|_{2,2} + 180 \Omega^q(\Pi_n q_0(h)) \|\Pi_n q_0(h)\|_{2,2} \Big) \\
    &\qquad + \epsilon_n^2 \Big( 144 \Omega^q(q_n(h))^2 + 144 \Omega^q(\Pi_n q_0(h))^2 + 72 \Omega^h(h-h^\dagger)^2 \Big)  \\
    &\qquad + \|\Pi_n q_0(h) - q_0(h) \|_{2,2}^2 + \|q_0(h)\|_{2,2}^2 + \gamma_n^q \Big( \|\Pi_n q_0(h)\|_\Qcal^2 - \|q_n(h)\|_\Qcal^2 \Big)
\end{align*}
Furthermore, noting that $q_0(h) = P(h - h^\dagger)$, and applying \Cref{assum:universal-approximation-h-general}, along with the AM-GM inequality, gives us
\begin{align*}
    \frac{1}{2} \|q_n(h)\|_{2,2}^2 &\leq 180 \epsilon_n \Omega^q(q_n(h)) \|q_n(h)\|_{2,2} + 180 \epsilon_n \delta_n \Omega^q(\Pi_n q_0(h)) + 180 \epsilon_n \|h - h^\dagger\|_w \Omega^q(\Pi_n q_0(h))  \\
    &\qquad + 180 \epsilon_n^2 \Omega^q(q_n(h))^2 + 180 \epsilon_n^2 \Omega^q(\Pi_n q_0(h))^2 + 72 \epsilon_n^2 \Omega^h(h-h^\dagger)^2 + \delta_n^2 + \|h - h^\dagger\|_w^2 \\
    &\qquad + \gamma_n^q \Big( \|\Pi_n q_0(h)\|_\Qcal^2 - \|q_n(h)\|_\Qcal^2 \Big) \\
    &\leq 180 \epsilon_n \Omega^q(q_n(h)) \|q_n(h)\|_{2,2} + 144  \epsilon_n^2 \Omega^q(q_n(h))^2 + 324 \epsilon_n^2 \Omega^q(\Pi_n q_0(h))^2 \\
    &\qquad + 72 \epsilon_n^2 \Omega^h(h-h^\dagger)^2  + 91 \delta_n^2 + 91 \|h - h^\dagger\|_w^2 + \gamma_n^q \Big( \|\Pi_n q_0(h)\|_\Qcal^2 - \|q_n(h)\|_\Qcal^2 \Big) \,.
\end{align*}

Finally, applying \Cref{lem:quadratic-inequality} to the above, along with the fact that $\sqrt{x+y} \leq \sqrt{x} + \sqrt{y}$ for non-negative $x,y$, gives us our final required bound
\begin{align*}
    \|q_n(h)\|_{2,2} &\leq \Big(377 \Omega^q(q_n(h)) + 26 \Omega^q(\Pi_n q_0(h)) + 12 \Omega^h(h-h^\dagger) \Big) \epsilon_n \\
    &\qquad + 14 \delta_n + 14 \|h-h^\dagger\|_w + 2 (\gamma_n^q)^{1/2} \Omega^q(\Pi_n q_0(h)) \,.
\end{align*}

\end{proof}

\begin{proof}[Proof of \Cref{lem:primal-minimax-erm-bound}]

First, by the optimality of $\hat h_n$ for the empirical minimax objective over $h\in\Hcal_n$ and $\Pi_n h^\dagger\in\Hcal_n$, we have 
\begin{align*}
&\sup_{q \in \Qcal_n}\E_n\bracks{\phi(\hat h_n, q) + \mu_n \hat h_n(S)^2 - \gamma_n^q \|q\|_\Qcal^2 + \gamma_n^h \|\hat h_n\|_\Hcal^2} \\
&\le \sup_{q \in \Qcal_n}\E_n\bracks{\phi(\Pi_n h^\dagger, q) + \mu_n \Pi_n h^\dagger(S)^2 - \gamma_n^q \|q\|_\Qcal^2 + \gamma_n^h \|\Pi_n h^\dagger\|_\Hcal^2}.
\end{align*}

Therefore, we have
\begin{align*}
    &J(\hat h_n) - J(h^\dagger) \\
    &= \sup_{q \in \Qcal} \EE[\phi(\hat h_n,q)] - \sup_{q \in \Qcal} \EE[\phi(h^\dagger,q)] \\
    &\leq \EE[\phi(\hat h_n,q_0(\hat h_n))] - \EE[\phi(h^\dagger,q_n(\Pi_n h^\dagger))] \\
    &\leq \EE[\phi(\hat h_n,q_0(\hat h_n))] - \EE[\phi(h^\dagger,q_n(\Pi_n h^\dagger))]  \\
    &\qquad - \sup_{q \in \Qcal_n} \Big( \EE_n[\phi(\hat h_n,q)] + \mu_n \EE_n[\hat h_n(S)^2] - \gamma_n^q \|q\|_\Qcal^2 + \gamma_n^h \|\hat h_n\|_\Hcal^2 \Big) \\
    &\qquad + \sup_{q \in \Qcal_n} \Big( \EE_n[\phi(\Pi_n h^\dagger,q)] + \mu_n \EE_n[ (\Pi_n h^\dagger)(S)^2]  - \gamma_n^q \|q\|_\Qcal^2 + \gamma_n^h \|\Pi_n h^\dagger \|_\Hcal^2 \Big) \\
    &\leq \EE[\phi(\hat h_n,\Pi_n q_0(\hat h_n))] - \EE[\phi(h^\dagger,q_n(\Pi_n h^\dagger))]  \\
    &\qquad - \Big( \EE_n[\phi(\hat h_n, \Pi_n q_0(\hat h_n))] + \mu_n \EE_n[\hat h_n(S)^2] - \gamma_n^q \|\Pi_n q_0(\hat h_n)\|_\Qcal^2 + \gamma_n^h \|\hat h_n\|_\Hcal^2 \Big) \\
    &\qquad + \Big( \EE_n[\phi(h^\dagger,q_n(\Pi_n h^\dagger))] + \mu_n \EE_n[ (\Pi_n h^\dagger)(S)^2]  - \gamma_n^q \|q_n(\Pi_n h^\dagger)\|_\Qcal^2 + \gamma_n^h \|\Pi_n h^\dagger \|_\Hcal^2 \Big) + \Ecal_1 + \Ecal_2 \\
    &= (\EE - \EE_n)\Big[\phi(\hat h_n,\Pi_n q_0(\hat h_n)) - \phi(h^\dagger,q_n(\Pi_n h^\dagger)\Big] + \mu_n \EE_n[(\Pi_n h^\dagger)(S)^2 - \hat h(S)^2] \\
    &\qquad + 2 \gamma_n^q \Omega^q(\Pi_n q_0(\hat h_n))^2 + 2 M^2 \gamma_n^h - \gamma_n^q \Omega^q(q_n(\Pi_n h^\dagger))^2 - \gamma^h_n \Omega^h(\hat h_n)^2 + \Ecal_1 + \Ecal_2 \,,
\end{align*}
where the first inequality above follows from the facts that $\sup_{q \in \Qcal} \EE[\phi(h^\dagger,q)] \geq \EE[\phi(h^\dagger,q_n(\Pi_n h^\dagger))]$, the second inequality follows from the optimality of $\hat h_n$ described above,  the third inequality follows from $-\sup_{q \in \Qcal_n} \EE_n[\phi(\hat h_n,q)] \le \EE_n[\phi(\hat h_n,\Pi_n q_0(\hat h_n))]$, $\sup_{q \in \Qcal_n} \EE_n[\phi(h^\dagger,q)] =  \EE_n[\phi(h^\dagger,q_n(h^\dagger))]$, and where 
\begin{align*}
    \Ecal_1 &= \EE\Big[\phi(\hat h_n, q_0(\hat h_n)) - \phi(\hat h_n, \Pi_n q_0(\hat h_n)) \Big] \\
    \Ecal_2 &= \EE_n\Big[\phi(\Pi_n h^\dagger, q_n(\Pi_n h^\dagger)) \Big] - \EE_n\Big[\phi(h^\dagger, q_n(\Pi_n h^\dagger)) \Big] \,.
\end{align*}

Now, let us bound the error terms $\Ecal_1$ and $\Ecal_2$. For the first, using the shorthand $\Delta = q_0(\hat h_n) - \Pi_n q_0(\hat h_n)$, we have
\begin{align*}
    \Ecal_1 &= \EE\Big[ G(W;\hat h_n, \Delta) - r(W;\Delta) - \frac{1}{2} q_0(\hat h_n)(T)^\top q_0(\hat h_n)(T) + \frac{1}{2} \Pi_n q_0(\hat h_n)(T)^\top \Pi_n q_0(\hat h_n)(T) \Big] \\
    &= \EE\Big[ G(W;\hat h_n - h^\dagger, \Delta) - \frac{1}{2} q_0(\hat h_n)(T)^\top q_0(\hat h_n)(T) + \frac{1}{2} \Pi_n q_0(\hat h_n)(T)^\top \Pi_n q_0(\hat h_n)(T) \Big] \\
    &= \EE\Big[ (\hat h_n - h^\dagger)(S)^\top k(S,T) \Delta(T) + \frac{1}{2} \Delta(T)^\top \Delta(T) - q_0(\hat h_n)^\top \Delta(T) \Big] \\
    &\leq 2 \|\hat h_n - h^\dagger\|_w \delta_n + \frac{1}{2} \delta_n^2 \,,
\end{align*}
where above we apply the fact that $\|q_0(\hat h_n)\|_2 = \|\Pi_\Qcal (\hat h_n - h^\dagger)(S)^\top k(S,T)\|_2 = \|\hat h_n - h^\dagger\|_w$, and the fact that $\|\Delta\|_2 \leq \delta_n$ by \Cref{assum:universal-approximation-h-general}, along with Cauchy Schwartz.

Next, for the second term above, we can bound
\begin{align*}
    \Ecal_2  &= (\EE_n - \EE)\Big[ G(W;\Pi_n h^\dagger - h^\dagger, q_n(\Pi_n h^\dagger)) \Big] + \EE\Big[ G(W; \Pi_n h^\dagger - h^\dagger, q_n(\Pi_n h^\dagger)) \Big]\\
    &\leq 36 \epsilon_n \|q_n(\Pi_n h^\dagger)\|_{2,2} + 18 \Omega^h(\Pi_n h^\dagger - h^\dagger)^2 \epsilon_n^2 + 18 \Omega^q(q_n(\Pi_n h^\dagger))^2 \epsilon_n^2 \\
    &\qquad + \EE\Big[ G(W; \Pi_n h^\dagger - h^\dagger, q_n(\Pi_n h^\dagger)) \Big] \\
    &\leq (\delta_n + 36 \epsilon_n) \|q_n(\Pi_n h^\dagger)\|_{2,2} + 36 \Omega^q(q_n(\Pi_n h^\dagger))^2 \epsilon_n^2 \,,
\end{align*}
where the first inequality follows by relaxing the sup in the negative term, the second inequality follows from the assumed high-probability event of \Cref{lem:primal-high-prob}, and the final inequality follows from our Lipschitz assumptions on $G$ and our universal approximation assumption, along with the fact that $\Omega^h(\Pi_n h^\dagger - h^\dagger) \leq M \leq \Omega^q(q_n(\Pi_n h^\dagger))$, and
\begin{equation*}
    |\EE[G(W;h,q)]| = |\EE[(P h)(T)^\top q(T)]| \leq \|h\|_w \|q\|_{2,2} \,.
\end{equation*}
Furthermore, by \Cref{lem:primal-interior-maximization}, under the same high-probability event, we have
\begin{align*}
    \|q_n(\Pi_n h^\dagger)\|_{2,2} &\leq \Big(389 \Omega^q(q_n(\Pi_n h^\dagger)) + 26 \Omega^q(\Pi_n q_0(\Pi_n h^\dagger)) \Big) \epsilon_n + 28 \delta_n + 2 (\gamma_n^q)^{1/2} \Omega^q(\Pi_n q_0(\Pi_n h^\dagger)) \,.
\end{align*}
Therefore, applying the AM-GM inequality to the above, we have
\begin{align*}
    \Ecal_2 &\leq \Big( 14775 \Omega^q(q_n(\Pi_n h^\dagger))^2 + 949 \Omega^q(\Pi_n q_0(\Pi_n h^\dagger)) \Big) \epsilon_n^2 + 741 \delta_n^2 + 37 \gamma_n^q \Omega^q(\Pi_n q_0(\Pi_n h^\dagger))^2
\end{align*}

Putting all of the above together, we get
\begin{align*}
    J(\hat h_n) - J(h^\dagger) &\leq (\EE - \EE_n)\Big[\phi(\hat h_n,\Pi_n q_0(\hat h_n)) - \phi(h^\dagger,q_n(\Pi_n h^\dagger)\Big] + 2 \delta_n \|\hat h_n - h^\dagger\|_w \\
    &\qquad + \Big( 14775 \Omega^q(q_n(\Pi_n h^\dagger))^2 + 949 \Omega^q(\Pi_n q_0(\Pi_n h^\dagger)) \Big) \epsilon_n^2 + 741 \delta_n^2 \\
    &\qquad + \Big( 37 \Omega^q(\Pi_n q_0(\Pi_n h^\dagger))^2 + 2 \Omega^q(\Pi_n q_0(\hat h_n))^2 \Big) \gamma_n^q  + 2 M^2 \gamma_n^h \\
    &\qquad + \mu_n \EE_n[(\Pi_n h^\dagger)(S)^2 - \hat h(S)^2] - \gamma_n^q \Omega^q(q_n(\Pi_n h^\dagger))^2 - \gamma^h_n \Omega^h(\hat h_n)^2
\end{align*}
which is our required bound

\end{proof}

\begin{proof}[Proof of \Cref{lem:primal-stochastic-equicontinuity-bound}]

First, as argued in the proof overview, and plugging in $q_0(h^\dagger) = 0$, we have
\begin{equation*}
    \phi(\hat h_n, \Pi_n q_0(\hat h_n)) - \phi(h^\dagger, q_n(\Pi_n h^\dagger) = \Phi_n\Big(\Pi_n q_0(\hat h_n)\Big) - \Phi_n\Big(q_n(\Pi_n h^\dagger)\Big) + G\Big(W; \hat h_n - h^\dagger, \Pi_n q_0(\hat h_n)\Big) \,,
\end{equation*}
where $\Phi_n$ is defined as in the statement of \Cref{lem:primal-high-prob}. Then, applying \Cref{lem:primal-high-prob}, under its high probability event we have
\begin{align*}
    &\Big| (\EE_n - \EE) \Big[ \phi(\hat h_n, \Pi_n q_0(\hat h_n)) - \phi(h^\dagger, q_n(\Pi_n h^\dagger) \Big] \Big| \\
    &\leq \Big| (\EE_n - \EE) \Big[\Phi_n\Big(\Pi_n q_0(\hat h_n)\Big) - \Phi_n\Big(q_n(\Pi_n h^\dagger)\Big)\Big]\Big| + \Big| (\EE_n - \EE)\Big[G\Big(W; \hat h_n - h^\dagger, \Pi_n q_0(\hat h_n)\Big)\Big] \Big|\\ 
    &\leq 54 \epsilon_n \Big( \Omega^q(\Pi_n q_0(\hat h_n)) \|\Pi_n q_0(\hat h_n) \|_{2,2} + \Omega^q(q_n(\Pi_n h^\dagger)) \|q_n(\Pi_n h^\dagger)\|_{2,2} \Big) + 36 \epsilon_n \|\Pi_n q_0(\hat h_n)\|_{2,2} \\
    &\qquad + \Big( 54 \Omega^q(\Pi_n q_0(\hat h_n))^2 + 54 \Omega^q(q_n(\Pi_n h^\dagger))^2 + 18 \Omega^h(\hat h_n - h^\dagger)^2 \epsilon_n^2 + 18 \Omega^q(\Pi_n q_0(\hat h_n))^2 \epsilon_n^2 \\
    &\leq 90 \epsilon_n \Omega^q(\Pi_n q_0(\hat h_n)) \|\Pi_n q_0(\hat h_n) \|_{2,2} + 54 \epsilon_n \Omega^q(q_n(\Pi_n h^\dagger)) \|q_n(\Pi_n h^\dagger)\|_{2,2}  \\
    &\qquad + \Big( 72 \Omega^q(\Pi_n q_0(\hat h_n))^2 + 54 \Omega^q(q_n(\Pi_n h^\dagger))^2 + 18 \Omega^h(\hat h_n - h^\dagger)^2 \Big) \epsilon_n^2
\end{align*}

Next, by \Cref{assum:universal-approximation-h-general}, and the definition of $q_0(h)$, we have
\begin{align*}
    \Big\|\Pi_n q_0(\hat h_n)\Big\|_{2,2} &\leq \|q_0(\hat h_n)\|_{2,2} + \Big\|\Pi_n q_0(\hat h_n) - q_0(\hat h_n)\Big\|_{2,2} \\
    &\leq \|\hat h_n - h^\dagger\|_w + \delta_n \,.
\end{align*}
In addition, by \Cref{lem:primal-interior-maximization}, under the same high-probability event as above we have
\begin{align*}
    \|q_n(\Pi_n h^\dagger)\|_{2,2} &\leq \Big(389 \Omega^q(q_n(\Pi_n h^\dagger)) + 26 \Omega^q(\Pi_n q_0(\Pi_n h^\dagger)) \Big) \epsilon_n + 28 \delta_n + 2 (\gamma_n^q)^{1/2} \Omega^q(\Pi_n q_0(\Pi_n h^\dagger)) \,.
\end{align*}

Therefore, combining these bounds, under the previous high probability event, which occurs with probability at least $1-3\zeta$, we have
\begin{align*}
    &\Big| (\EE_n - \EE) \Big[\phi(\hat h_n, \Pi_n q_0(\hat h_n)) - \phi(h^\dagger, q_n(\Pi_n h^\dagger)\Big] \Big| \\
    &\leq 90 \epsilon_n \Omega^q(\Pi_n q_0(\hat h_n)) \|\hat h_n - h^\dagger\|_w + 801 \delta_n^2 + 54 \gamma_n^q \Omega^q(\Pi_n q_0(\Pi_n h^\dagger))^2 \\
    &\qquad + \Big( 22572 \Omega^q(q_n(\Pi_n h^\dagger))^2 + 702 \Omega^q(\Pi_n q_0(\Pi_n h^\dagger))^2 + 117 \Omega^q(\Pi_n q_0(\hat h_n))^2 + 18 \Omega^h(\hat h_n - h^\dagger)^2 \Big) \epsilon_n^2
\end{align*}
which is our promised bound.

\end{proof}

\begin{proof}[Proof of \Cref{lem:h-norm-convergence}]

We first note that
\begin{align*}
    \|\hat h_n\|_{2,2}^2 - \|h^\dagger\|_{2,2}^2 &= (\EE - \EE_n)\Big[ \hat h_n(S)^\top \hat h_n(S) - h^\dagger(S)^\top h^\dagger(S) \Big] \\
    &\qquad + \Big( \|\hat h_n\|_{n,2}^2 - \|\Pi_n h^\dagger\|_{n,2}^2 \Big) + \Big( \|\Pi_n h^\dagger\|_{2,2}^2 - \|h^\dagger\|_{2,2}^2 \Big) \,.
\end{align*}
Now, as already argued in the main proof text, we have 
\begin{equation*}
    \|\hat h_n\|_{n,2}^2 - \|\Pi_n h^\dagger\|_{n,2}^2 \leq o_p(1) \,,
\end{equation*}
and we also have
\begin{align*}
    \|\Pi_n h^\dagger\|_{2,2}^2 - \|h^\dagger\|_{2,2}^2 &= \Big( \|\Pi_n h^\dagger\|_{2,2} + \|h^\dagger\|_{2,2} \Big) \Big( \|\Pi_n h^\dagger\|_{2,2} - \|h^\dagger\|_{2,2} \Big) \\
    &\leq 2 \Big| \|\Pi_n h^\dagger\|_{2,2} - \|h^\dagger\|_{2,2} \Big| \\
    &\leq 2 \|\Pi_n h^\dagger - h^\dagger\|_{2,2} \\
    &\leq 2 \delta_n \,,
\end{align*}
so therefore
\begin{align*}
    \|\hat h_n\|_{2,2}^2 - \|h^\dagger\|_{2,2}^2 = (\EE - \EE_n)\Big[ (\hat h_n - h^\dagger)(S)^\top (\hat h_n - h^\dagger)(S) + 2 h^\dagger(S)^\top (\hat h_n - h^\dagger)(S) \Big] + o_p(1) \,.
\end{align*}
Therefore, we only need to bound the stochastic equicontinuity term above.

Next, following identical arguments as in the proof of \Cref{lem:primal-high-prob}, under our additional critical radius assumption on $\Hcal_n$, we can bound
\begin{align*}
    &\Big| (\EE_n - \EE)\Big[ (\hat h_n - h^\dagger)(S)^\top (\hat h_n - h^\dagger)(S) + 2 h^\dagger(S)^\top (\hat h_n - h^\dagger)(S) \Big] \Big| \\
    &\preccurlyeq \Omega^h(\hat h_n - h^\dagger) \|\hat h_n - h^\dagger\|_{2,2} \tilde\epsilon_n + \Omega^h(\hat h_n - h^\dagger)^2 \tilde\epsilon_n^2 \\
    &\preccurlyeq \Omega^h(\hat h_n) \tilde\epsilon_n + \Omega^h(\hat h_n)^2 \tilde\epsilon_n^2 \,,
\end{align*}
which holds with probability at least $1-2\zeta_n$. Therefore, since $\zeta_n \to 0$, we have our required result as long as $\Omega^h(\hat h_n) \tilde\epsilon_n = o_p(1)$.

Now, under our first set of assumptions, we have $\Omega^h(\hat h_n - h^\dagger) \leq M$, so therefore
\begin{align*}
    \Omega^h(\hat h_n - h^\dagger) \tilde\epsilon_n &\leq  M \tilde\epsilon_n \\
    &= o(1) \,.
\end{align*}
Alternatively, under our second set of assumptions, it follows from \Cref{eq:full-bound-2} that
\begin{equation*}
    \Omega^h(\hat h_n) \preccurlyeq \frac{\delta_n}{(\gamma_n^h)^{1/2}} + L M \frac{(\gamma_n^q)^{1/2}}{(\gamma_n^h)^{1/2}} + M + \frac{\mu_n^{1/2}}{(\gamma_n^h)^{1/2}} \,,
\end{equation*}
which holds with probability at least $1-3\zeta_n$, and furthermore applying our assumptions on the minimum size of $\gamma_n^q$ and $\gamma_n^h$, this can be simplified to
\begin{equation*}
    \Omega^h(\hat h_n) \preccurlyeq M + \tilde\epsilon_n^{-1} (\delta_n + \mu_n^{1/2}) \,.
\end{equation*}
Therefore, we have
\begin{align*}
    \Omega^h(\hat h_n) \tilde\epsilon_n &\leq M \tilde \epsilon_n + \delta_n + \mu_n \\
    &= o(1) \,,
\end{align*}
which holds with probability at least $1-3\zeta_n$. Since $\zeta_n \to 0$, this then gives us $\Omega^h(\hat h_n) \tilde\epsilon_n = o_p(1)$.
Therefore, under either set of assumptions we have $\Omega^h(\hat h_n) \tilde\epsilon_n = o_p(1)$, so we are done.

\end{proof}

\begin{proof}[Proof of \Cref{lem:p-convergence}]

First, let
\begin{equation*}
    D(h) = \inf_{h_0 \in \Hcal_0} \|h - h_0\|_{2, 2} \,.
\end{equation*}
In addition, for any $\eta > 0$, let $U_\eta$ be defined such that $\PP(\|\hat h_n\|_\Hcal > U_\eta) \leq \eta$; note that as argued in the proof of \Cref{lem:h-norm-convergence} $\|\hat h_n\|_\Hcal$ is stochastically bounded, so therefore such a $U_\eta < \infty$ exists for every $\eta>0$.

Now, given $\eta > 0$ and $\epsilon > 0$, we have
\begin{align*}
    \PP\Big( D(\hat h_n) \geq \epsilon \Big) &\leq \PP\Big(D(\hat h_n) \geq \epsilon, \|\hat h_n\| \leq U_\eta \Big) + \eta\\
    &\leq \PP\Big( \|\hat h_n - h^\dagger\|_w \geq M(\epsilon,\eta) \Big) + \eta \,,
\end{align*}
where
\begin{equation*}
    M(\epsilon,\eta) = \inf_{h \in \Hcal_n : \|h\|_\Hcal \leq U_\eta, D(h) \geq \epsilon} \|h - h^\dagger\|_w \,.
\end{equation*}
Note that this holds since $D(\hat h_n) \geq \epsilon$ and $\|\hat h_n\| \leq U_\eta$ together imply that $\hat h_n \in \{h \in \Hcal_n : \|h\|_\Hcal \leq U_\eta, D(h) \geq \epsilon\}$, and therefore $\|\hat h_n - h^\dagger\|_w \geq M(\epsilon,\eta)$. 

Now, let $\bar\Hcal_U = \{h \in \bar\Hcal : \|h\|_\Hcal \leq U\}$, and $\Hcal_{U,\epsilon} = \{h \in \bar\Hcal : \|h\|_\Hcal \leq U, D(h) \geq \epsilon\}$. It is trivial to see that $\bar\Hcal_{U,\epsilon}$ is totally bounded, since it is a subset of $\bar\Hcal_U$, which we have assumed to be compact for every $U < \infty$. Furthermore, suppose that $h'_n$ is a sequence in $\{h \in \Hcal : D(h) \geq \epsilon\}$, and $\|\cdot\|_\Hcal$ limit $h'$. Then, we have
\begin{align*}
    D(h') &\geq \inf_{h_0 \in \Hcal_0} \|h' - h_0\|_{2,2} \\
    &\geq \inf_{h_0 \in \Hcal_0} \|h'_n - h_0\|_{2,2} - \|h'_n - h'\|_{2,2} \\
    &\geq \inf_{h_0 \in \Hcal_0} \|h'_n - h_0\|_{2,2} - \|h'_n - h'\|_\Hcal \\
    &\geq \epsilon - o(1) \,,
\end{align*}
since $D(h'_n) \geq \epsilon$ and $\|\cdot\|_\Hcal$ dominates the $L_2$ norm. Therefore, we must have $D(h') \geq \epsilon$, so $\{h \in \Hcal : D(h) \geq \epsilon\}$ is closed under $\|\cdot\|_\Hcal$. Therefore, $\Hcal_{U,\epsilon}$ is the intersection of two closed sets, and is closed.

Putting together the previous arguments, we have that $\bar\Hcal_{U,\epsilon}$ is compact under $\|\cdot\|_\Hcal$. In addition $h \mapsto \|h - h^\dagger\|_w$ is continuous under $\|\cdot\|_\Hcal$, since $\|\cdot\|_\Hcal$ dominates the $L_2$ norm, which in turn dominates the projected norm, and so it follows from the extreme value theorem that
\begin{align*}
    M(\epsilon,\eta) &\geq \inf_{h \in \bar\Hcal_{U,\epsilon}} \|h - h^\dagger\|_w \\
    &= \|h_{U,\epsilon} - h^\dagger\|_w \,,
\end{align*}
for some $h_{U,\epsilon} \in \bar\Hcal_{U,\epsilon}$. But clearly $\|h_{U,\epsilon} - h^\dagger\|_w > 0$, since $D(h_{U,\epsilon}) \geq \epsilon$ implies that $h_{U,\epsilon} \neq \Hcal_0$, which implies that $M(\epsilon,\eta) > 0$. Therefore, since we have already argued from the main proof text that $\|\hat h_n - h^\dagger\|_w = o_p(1)$, we have that for every $\epsilon > 0$ and $\eta>0$ that
\begin{equation*}
    \PP\Big( D(\hat h_n) \geq \epsilon\Big) \leq o(1) + \eta \,.
\end{equation*}
Now, consider some fixed $\epsilon>0$. Since the above convergence holds for every $\eta>0$, we must have $\PP(D(\hat h_n) \geq \epsilon) \to 0$.
Furthermore, since this logic holds for every $\epsilon>0$, we have $D(\hat h_n) = \inf_{h_0 \in \Hcal_0} \|\hat h_n - h_0\|_{2,2} = o_p(1)$. That is, we can define some sequence $h_{n,0} \in \Hcal_0$ such that
\begin{equation*}
    \|\hat h_n - h_{n,0}\|_{2,2} = o_p(1) \,.
\end{equation*}

Next, since $\Hcal_0$ is an affine set given by $h^\dagger$ plus the null-space of $P$, it easily follows from the orthogonality that
\begin{equation*}
    \|h_0\|_{2,2}^2 = \|h^\dagger\|_{2,2}^2 + \|h_0 - h^\dagger\|_{2,2}^2 \,,
\end{equation*}
for every $h_0 \in \Hcal_0$. Therefore, we can bound
\begin{align*}
    \|\hat h_n - h^\dagger\|_{2,2} &\leq  \|\hat h_n - h_{n,0}\|_{2,2} + \|h_{n,0} - h^\dagger\|_{2,2} \\
    &=  \|\hat h_n - h_{n,0}\|_{2,2} + \sqrt{\|h_{n,0}\|_{2,2}^2 - \|h^\dagger\|_{2,2}^2}\\
    &=  \|\hat h_n - h_{n,0}\|_{2,2} + \sqrt{(\|h_{n,0}\|_{2,2} + \|h^\dagger\|_{2,2}) (\|h_{n,0}\|_{2,2} - \|h^\dagger\|_{2,2})}\\
    &\leq  \|\hat h_n - h_{n,0}\|_{2,2} + \sqrt{2 \Big(\|h_{n,0}\|_{2,2} - \|h^\dagger\|_{2,2}\Big)}\\
    &\leq  \|\hat h_n - h_{n,0}\|_{2,2} + \sqrt{2 \Big(\|\hat h_n -h_{n,0}\|_{2,2} + \|\hat h_n\|_{2,2} - \|h^\dagger\|_{2,2}\Big)}\\
\end{align*}
Now, we have already argued that $\|\hat h_n - h_{n,0}\|_{2,2} = o_p(1)$. Furthermore, we have from \Cref{lem:h-norm-convergence} that $\|\hat h_n\|_{2,2}^2 \leq \|h^\dagger\|_{2,2}^2 + o_p(1)$, and therefore $\|\hat h_n\|_{2,2} - \|h^\dagger\|_{2,2} \leq o_p(1)$. Putting these bounds together gives
\begin{align*}
    \|\hat h_n - h^\dagger\|_{2,2} = o_p(1) \,,
\end{align*}
which is our required result.

\end{proof}

\subsection{Proof of \Cref{thm:q-estimator-bound-general}}

Here, we consider the dual estimator for $q_0$, given by
\begin{equation*}
    \hat q_n = \argmax_{q \in \widetilde\Qcal_n} \EE_n[\psi(\hat \xi_n,q)] - \tilde\gamma_n^q \|q\|_{\widetilde\Qcal}^2 \,,
\end{equation*}
where
\begin{equation*}
    \hat \xi_n = \argmin_{\xi \in \Xi_n} \sup_{q \in \Qcal_n} \EE_n[\psi(\xi,q)] - \gamma_n^q \|q\|_\Qcal^2 + \gamma_n^\xi \|\xi\|_\Xi^2 \,,
\end{equation*}
and
\begin{equation*}
    \psi(\xi,q) = G(W;\xi,q) - m(W;\xi) - \frac{1}{2} q(T)^\top q(T) \,.
\end{equation*}
That is, $\hat \xi_n$ is the corresponding saddle-point solution for an empirical minimax problem. The proof that follows is very similar to that for \Cref{thm:h-estimator-bound-general}, with some important differences that we will emphasize below. Again, we provide a high-level overview of the proof first, and relegate proofs of intermediate lemmas to the end.

Before we proceed, we will define some simple notation. First, analogous to the proof of \Cref{thm:h-estimator-bound-general}, for any $\xi \in \Hcal$ we define
\begin{align*}
    q_0(\xi) &= \argmax_{q \in \Qcal} \EE[\psi(\xi,q)] \\
    q_n(\xi) &= \argmax_{q \in \Qcal_n} \EE_n[\psi(\xi,q)] - \gamma_n^q \|q\|_\Qcal^2 \\
    \tilde q_n(\xi) &= \argmax_{q \in \widetilde\Qcal_n} \EE_n[\psi(\xi,q)] - \tilde\gamma_n^q \|q\|_{\widetilde\Qcal}^2 \,.
\end{align*}
Note that we have separate definitions $q_n$ and $\tilde q_n$ based on $\Qcal_n$ and $\widetilde \Qcal_n$ respectively.
By some simple algebra and calculus, it is easy to verify that
\begin{equation*}
    q_0(\xi) = \Pi_\Qcal[ k(S,T)^\top \xi(S) \mid T ] \,.
\end{equation*}

Also, for any $q \in \Qcal_n$, $\tilde q \in \widetilde\Qcal_n$, or $\xi \in \Xi_n$, we define
\begin{align*}
    \Omega^q(q) &= \max(\|q\|_\Qcal, M) \\
    \widetilde\Omega^q(\tilde q) &= \max(\|\tilde q\|_{\widetilde\Qcal}, M) \\
    \Omega^\xi(\xi) &= \max(\|\xi\|_\Xi, M) \,.
\end{align*}
Finally, in the theory below we will let $\xi_0$ be defined as in \Cref{assum:complexity-q-general}.

\subsubsection*{Feasibility of Population Saddle Point Solution}

First, we provide a lemma that justifies that under \Cref{assump: new-nuisance}, the population saddle-point solution will be equal to the particular $q^\dagger \in \Qcal_0$ defined above.

\begin{lemma}
\label{lem:dual-feasibility}

Let \Cref{assump: new-nuisance} be given. Then, we have
\begin{equation*}
    \Xi_0 = \argmin_{h \in \Hcal} \sup_{q \in \Qcal} \EE[\psi(h,q)] \,.
\end{equation*}
Furthermore, for any given $\xi_0' \in \Xi_0$, we have that $\argmax_{q \in \Qcal}\EE[\psi(\xi_0',q)] = \{q^\dagger\}$. 

\end{lemma}

This implies that, given \Cref{assump: new-nuisance}, \emph{any} overall saddle-point solution to the population minimax problem recovers $q^\dagger$.

\subsubsection*{High-Probability Bound}

Next, similar to the proof of \Cref{thm:h-estimator-bound-general}, we will provide some high-probability bounds, which will be used extensively to prove the remaining sub-lemmas.

\begin{lemma}
\label{lem:dual-high-prob}

    Let $\epsilon_n$ be defined again as in \Cref{lem:primal-high-prob}, and let
    \begin{align*}
        \Psi_n(q-q^\dagger) &= G(W;\xi_0, q-q^\dagger) - \frac{1}{2} (q-q^\dagger)(T)^\top (q-q^\dagger)(T) - q^\dagger(T)^\top (q-q^\dagger)(T) \\
        \widetilde\Psi_n(\xi-\xi_0) &= G(W;\xi-\xi_0, q^\dagger) - m(W;\xi-\xi_0) \,,
    \end{align*}
    for any $q \in \Qcal_n \cup \widetilde\Qcal_n$ and $\xi \in \Xi_n$. Then, with probability at least $1-7\zeta$, we have
    \begin{equation*}
        \Big| (\EE_n - \EE) \Big[ \Psi_n(q-q^\dagger) \Big] \Big| \leq 108 \Omega^q(q) \epsilon_n \|q-q^\dagger\|_{2,2} + 216 \Omega^q(q)^2 \epsilon_n^2
    \end{equation*}
    for all $q \in \Qcal_n$, 
    \begin{equation*}
        \Big| (\EE_n - \EE) \Big[ \Psi_n(q-q^\dagger) \Big] \Big| \leq 108 \widetilde\Omega^q(q) \epsilon_n \|q-q^\dagger\|_{2,2} + 216 \widetilde\Omega^q(q)^2 \epsilon_n^2
    \end{equation*}
    for all $q \in \widetilde\Qcal_n$,
    \begin{equation*}
        \Big| (\EE_n - \EE) \Big[ \widetilde\Psi_n(\xi-\xi_0) \Big] \Big| \leq 36 \epsilon_n \|\xi-\xi_0\|_{2,2} + 72 \Omega^\xi(\xi) \epsilon_n^2
    \end{equation*}
    for all $\xi \in \Xi_n$, 
    \begin{equation*}
        \Big| (\EE_n - \EE) \Big[ G(W;\xi - \xi_0, q - q^\dagger) \Big] \Big| \leq 36 \epsilon_n \|q-q^\dagger\|_{2,2} + 144 \Omega^q(q) \Omega^\xi(\xi) \epsilon_n^2
    \end{equation*}
    for every $\xi \in \Xi_n$ and $q \in \Qcal_n$, and
    \begin{equation*}
        \Big| (\EE_n - \EE) \Big[ G(W;\xi - \xi_0, q - q^\dagger) \Big] \Big| \leq 36 \epsilon_n \|q-q^\dagger\|_{2,2} + 144 \widetilde\Omega^q(q) \Omega^\xi(\xi) \epsilon_n^2
    \end{equation*}
    for every $\xi \in \Xi_n$ and $q \in \widetilde\Qcal_n$.

\end{lemma}

\subsubsection*{Convergence Lemma for Empirical Saddle Point Solution}

In terms of our above notation, our estimator $\hat q_n$ is given by $\tilde q_n(\hat \xi_n)$. Here we provide a lemma for the convergence of this estimator in terms of the projected-norm behavior of $\hat \xi_n$. This will then provide the basis of our slow-rates theory below.

\begin{lemma}
\label{lem:dual-saddle-point-bound}

Let the conditions of \Cref{thm:q-estimator-bound-general} be given. Then, under the high-probability event of \Cref{lem:dual-high-prob}, we have
\begin{align*}
    \|q_n(\xi) - q^\dagger\|_{2,2} &\leq \Big( 581 \Omega^q(q_n(\xi)) + 34 \Omega^\xi(\xi)\Big) \epsilon_n + 25 \delta_n + 4 \|\xi - \xi_0\|_w + 4 M (\gamma_n^q)^{1/2} \\
    \text{and} \quad \|\tilde q_n(\xi) - q^\dagger\|_{2,2} &\leq \Big( 581 \Omega^q(q_n(\xi)) + 34 \Omega^\xi(\xi)\Big) \epsilon_n + 25 \delta_n + 4 \|\xi - \xi_0\|_w + 4 M (\tilde\gamma_n^q)^{1/2} \,,
\end{align*}
for every $\xi \in \Xi_n$, where the weak norm $\|\cdot\|_w$ is defined as in the proof of \Cref{thm:h-estimator-bound-general}. Furthermore, if $\gamma_n^q \geq 42120 \epsilon_n^2$ or $\tilde\gamma_n^q \geq 42120 \epsilon_n^2$ respectively, these bounds can be respectively tightened to
\begin{align*}
    \|q_n(\xi) - q^\dagger\|_{2,2} &\leq 34 \Omega^\xi(\xi) \epsilon_n + 25 \delta_n + 4 \|\xi - \xi_0\|_w + 4 M (\gamma_n^q)^{1/2} \\
    \text{or} \quad \|\tilde q_n(\xi) - q^\dagger\|_{2,2} &\leq 34 \Omega^\xi(\xi) \epsilon_n + 25 \delta_n + 4 \|\xi - \xi_0\|_w + 4 M (\tilde\gamma_n^q)^{1/2} \,.
\end{align*}

\end{lemma}

This lemma implies that we can bound the $L_2$ error of $\hat q_n$ in terms of the projected error and complexity of the minimax estimate $\hat \xi_n$, since instantiating this result for our estimator $\hat q_n = \tilde q(\hat\xi_n)$ gives us
\begin{equation*}
    \|\hat q_n - q^\dagger\|_{2,2} \leq \Big( 581 \Omega^q(\hat q_n) + 34 \Omega^\xi(\hat\xi_n)\Big) \epsilon_n + 25 M^{1/2} \delta_n + 4 \|\hat\xi_n - \xi_0\|_w + 4 M (\tilde\gamma_n^q)^{1/2} \,,
\end{equation*}
or 
\begin{equation*}
    \|\hat q_n - q^\dagger\|_{2,2} \leq 34 \Omega^\xi(\hat\xi_n) \epsilon_n + 25 M^{1/2} \delta_n + 4 \|\hat\xi_n - \xi_0\|_w + 4 M (\tilde\gamma_n^q)^{1/2}
\end{equation*}
in the case that $\tilde\gamma_n^q$ is sufficiently large.
That is, as long we can either bound the maximum complexity of $\hat q_n$, or we set $\tilde\gamma_n^q$ sufficiently large, then it is sufficient to bound the projected error and complexity of $\hat\xi_n$.

Unfortunately, the minimax problem for $\xi_0$ is not a standard problem based on solving conditional moment restrictions --- rather it is a different kind of problem motivated by orthogonality conditions --- so we cannot easily obtain fast rates for $\|\hat\xi_n - \xi_0\|_w$.
However, we will establish conditions for obtaining slow rates below. The reasoning that follows is very similar as in the fast-rates analysis, except that we need to deal with some additional nuisance terms since here $q_0(\xi_0) \neq 0$.

\subsubsection*{Bounding Weak Norm by Population Objective Sub-optimality}

Here, we provide an analogue of \Cref{lem:primal-strong-convexity-bound} for the dual problem.

\begin{lemma}
\label{lem:dual-strong-convexity-bound}

Under \Cref{assump: new-nuisance}, for every $\xi \in \Hcal$, we have
\begin{equation*}
    \frac{1}{2} \|\xi - \xi_0\|_w^2 \leq \sup_{q \in \Qcal} \EE[\psi(\xi,q)] - \sup_{q \in \Qcal} \EE[\psi(\xi_0,q)] \,.
\end{equation*}

\end{lemma}

\subsubsection*{Reduction to Stochastic Equicontinuity Problem}

Next, given our boundedness and universal approximation assumptions, we can provide the following analogue of \Cref{lem:primal-minimax-erm-bound}.

\begin{lemma}
\label{lem:dual-minimax-erm-bound}
    Under the high-probability event of \Cref{lem:dual-high-prob}, we have
    \begin{align*}
    &\sup_{q \in \Qcal} \EE[\psi(\hat \xi_n,q)] - \sup_{q \in \Qcal} \EE[\psi(\xi_0,q)] \\
    &\leq (\EE - \EE_n)\Big[\psi(\hat \xi_n,\Pi_n q_0(\hat \xi_n)) - \psi(\xi_0,q_n(\Pi_n \xi_0)\Big] + 22734 \Omega^q(q_n(\Pi_n\xi_0))^2 \epsilon_n^2 \\
    &\qquad + 358 \delta_n^2 + 75 \Omega^q(\Pi_n q_0(\hat\xi_n))^2 \gamma_n^q + M^2 \gamma_n^\xi - \gamma_n^q \|q_n(\Pi_n \xi_0)\|_\Qcal^2 - \gamma_n^\xi \|\hat\xi_n\|_\Xi^2 \,.
\end{align*}
\end{lemma}
Therefore, again we have reduced our analysis to bounding a stochastic equicontinuity term. 

\subsubsection*{Bounding the Stochastic Equicontinuity Term using Localization}

Similar as for the estimator of $h^\dagger$, we can decompose
\begin{align*}
    &\psi(\hat \xi_n, \Pi_n q_0(\hat \xi_n)) - \psi(\xi_0, q_n(\xi_0))  \\
    &= G(W;\hat \xi_n-\xi_0,q^\dagger) + G(W;\hat \xi_n-\xi_0,\Pi_n q_0(\hat \xi_n) - q^\dagger) \\
    &\quad + G(W;\xi_0, \Pi_n q_0(\hat \xi_n) - q^\dagger) - G(W; \xi_0, q_n(\xi_0) - q^\dagger) \\
    &\quad - m(W; \hat\xi_n - \xi_0) \\
    &\quad - \frac{1}{2} \Big( \Pi_n q_0(\hat \xi_n) - q^\dagger \Big)(T)^\top \Big( \Pi_n q_0(\hat \xi_n) - q^\dagger \Big)(T) \\
    &\quad + \frac{1}{2} \Big( q_n(\xi_0) - q^\dagger \Big)(T)^\top \Big( q_n(\xi_0) - q^\dagger \Big)(T)  \\
    &\quad - q^\dagger(T)^\top \Big( \Pi_n q_0(\hat \xi_n) - q^\dagger \Big)(T)  + q^\dagger(T)^\top \Big( q_n(\xi_0) - q^\dagger \Big)(T) \,.
\end{align*}
Unfortunately, unlike for \Cref{thm:h-estimator-bound-general}, here we have $q^\dagger = q_0(\xi_0) \neq 0$ in general, so we cannot eliminate any terms in this decomposition. This will lead to a term proportional to $\epsilon_n \|\hat\xi_n - \xi_0\|_{2,2}$ in the localized stochastic equicontinuity bound, which will prevent us from obtaining fast rates in general, given ill-posedness. However, we can still follow a localization analysis and derive an analogue of \Cref{lem:primal-stochastic-equicontinuity-bound}, which will allow us to obtain slow rates.

\begin{lemma}
\label{lem:dual-stochastic-equicontinuity-bound}
    Let the assumptions of \Cref{thm:q-estimator-bound-general} be given. Then, under the high-probability event of \Cref{lem:dual-high-prob}, we have
    \begin{align*}
        &\Big| (\EE_n - \EE) \Big[ \psi(\hat \xi_n, \Pi_n q_0(\hat \xi_n)) - \psi(\xi_0, q_n(\Pi_n \xi_0) \Big] \Big| \\
        &\leq 144 \Omega^q(\Pi_n q_0(\hat\xi_n)) \epsilon_n \|\hat \xi_n - \xi_0\|_w  + 72 \epsilon_n + 3204 \delta_n^2 + 216 M^2 \gamma_n^q \\ 
        &\qquad + \Big( 360 \Omega^q(\Pi_n q_0(\hat\xi_n))^2 + 68418 \Omega^q(q_n(\Pi_n \xi_0))^2 + 144 \Omega^\xi(\hat\xi_n)^2 \Big) \epsilon_n^2  \,,
    \end{align*}
\end{lemma}

Note that since we are only obtaining slow rates, we do not include terms such as $\|q_n(\Pi_n \xi_0) - \Pi_n q_0(\Pi_n \xi_0)\|_{2,2}$ or $\|\hat \xi_n - \xi_0\|_w$ in the bound, and instead give a simplified bound proportional to $\epsilon_n$. Including those terms could allow us to obtain slightly smaller constants, at the expense of a much more complicated analysis.

\subsubsection*{First Strong Norm Bound}

Combining together the results of \Cref{lem:dual-strong-convexity-bound,lem:dual-minimax-erm-bound,lem:dual-stochastic-equicontinuity-bound}, we have
\begin{align}
    &\frac{1}{2} \|\hat\xi_n - \xi_0\|_w^2 \nonumber \\
    &\leq 144 \Omega^q(\Pi_n q_0(\hat\xi_n)) \epsilon_n \|\hat \xi_n - \xi_0\|_w  + 72 \epsilon_n + 3204 \delta_n^2 + 216 M^2 \gamma_n^q \nonumber \\ 
    &\qquad + \Big( 360 \Omega^q(\Pi_n q_0(\hat\xi_n))^2 + 68418 \Omega^q(q_n(\Pi_n \xi_0))^2 + 144 \Omega^\xi(\hat\xi_n)^2 \Big) \epsilon_n^2  \nonumber \\
    &\qquad + 22734 \Omega^q(q_n(\Pi_n\xi_0))^2 \epsilon_n^2 + 358 \delta_n^2 + 75 \Omega^q(\Pi_n q_0(\hat\xi_n))^2 \gamma_n^q + M^2 \gamma_n^\xi \nonumber \\
    &\qquad - \gamma_n^q \|q_n(\Pi_n \xi_0)\|_\Qcal^2 - \gamma_n^\xi \|\hat\xi_n\|_\Xi^2 \nonumber \\
    &\leq 144 \Omega^q(\Pi_n q_0(\hat\xi_n)) \epsilon_n \|\hat \xi_n - \xi_0\|_w + 72 \epsilon_n + 3562 \delta_n^2 + 75 \Omega^q(\Pi_n q_0(\hat\epsilon_n))^2 \gamma_n^q + 216M^2 \gamma_n^q + M^2 \gamma_n^\xi \nonumber \\
    &\qquad + \Big( 360 \Omega^q(\Pi_n q_0(\hat\xi_n))^2 + 91152 \Omega^q(q_n(\Pi_n \xi_0))^2 + 144 \Omega^\xi(\hat\xi_n)^2 \Big) \epsilon_n^2  \nonumber \\
    &\qquad - \gamma_n^q \|q_n(\Pi_n \xi_0)\|_\Qcal^2 - \gamma_n^\xi \|\hat\xi_n\|_\Xi^2 \,.
    \label{eq:dual-bound-master}
\end{align}
which occurs with probability at least $1-7\zeta$. Now, under the conditions of our first bound we have that all $\Omega^q(\cdot)$ And $\Omega^\xi(\cdot)$ terms are bounded by $M$, and therefore the above bound can be further relaxed to
\begin{align*}
    \frac{1}{2} \|\hat\xi_n - \xi_0\|_w^2 &\leq 144 M \epsilon_n \|\hat \xi_n - \xi_0\|_w + 72 \epsilon_n + 3562 \delta_n^2 + 291 M^2 \gamma_n^q + M^2 \gamma_n^\xi + 91656 M^2 \epsilon_n^2  \,.
\end{align*}
Then, applying \Cref{lem:quadratic-inequality}, along with $\sqrt{x+y} \leq \sqrt{x}+\sqrt{y}$ for non-negative $x$ and $y$, this gives us
\begin{align*}
    \|\hat\xi_n - \xi_0\|_w &\leq 717 M \epsilon_n + 12 \epsilon_n^{1/2} + 85 \delta_n + 25 M (\gamma_n^q)^{1/2} + 2 M (\gamma_n^\xi)^{1/2} \,.
\end{align*}
Then, combining this with the result of \Cref{lem:dual-saddle-point-bound} in this setting, we obtain
\begin{equation*}
    \|\hat q_n - q^\dagger\|_{2,2} \leq 3483 M \epsilon_n + 48 \epsilon_n^{1/2} + 365 \delta_n + 100 M (\gamma_n^q)^{1/2} + 4 M (\tilde\gamma_n^q)^{1/2} + 8 M (\gamma_n^\xi)^{1/2} \,,
\end{equation*}
from which our first required bound trivially follows.

\subsubsection*{Second Strong Norm Bound}

Given our additional assumption for the second bound that $\|\Pi_n q_0(\xi)\|_\Qcal \leq L \|\xi\|_\Xi$ for all $\xi \in \Xi_n$, along with the facts that $|\Omega^\xi(\xi)^2 - \|\xi\|_\Xi^2| \leq M^2$ for all $\xi \in \Xi_n$ and $|\Omega^q(q)^2 - \|q\|_\Qcal^2| \leq M^2$ for all $q \in \Qcal_n$, we have
\begin{align*}
    &\frac{1}{2} \|\hat\xi_n - \xi_0\|_w^2 + \frac{1}{2} \gamma_n^\xi \Omega^\xi(\hat\xi_n)^2 \\
    &\leq 144 L \Omega^\xi(\hat\xi_n) \epsilon_n \|\hat \xi_n - \xi_0\|_w + 72 \epsilon_n + 3562 \delta_n^2 + 75 L^2 \Omega^\xi(\hat\xi_n)^2 \gamma_n^q + 217 M^2 \gamma_n^q + 2 M^2 \gamma_n^\xi \\
    &\qquad + \Big( 504 L^2 \Omega^\xi(\hat\xi_n)^2 + 91152 \Omega^q(q_n(\Pi_n \xi_0))^2 \Big) \epsilon_n^2  - \gamma_n^q \Omega^\xi(q_n(\Pi_n \xi_0))^2 - \frac{1}{2} \gamma_n^\xi \Omega^\xi(\hat\xi_n)^2 \,.
\end{align*}

Now, under our assumption that
\begin{equation*}
    \gamma_n^q \geq 91152 \epsilon_n^2 \,,
\end{equation*}
this bound can be further relaxed to
\begin{align*}
    \frac{1}{2} \|\hat\xi_n - \xi_0\|_w^2 + \frac{1}{2} \gamma_n^\xi \Omega^\xi(\hat\xi_n)^2 &\leq 144 L \Omega^\xi(\hat\xi_n) \epsilon_n \|\hat \xi_n - \xi_0\|_w + 72 \epsilon_n + 3562 \delta_n^2 + 75 L^2 \Omega^\xi(\hat\xi_n)^2 \gamma_n^q \\
    &\qquad + 217 M^2 \gamma_n^q + 2 M^2 \gamma_n^\xi + 504 L^2 \Omega^q(\hat\xi_n)^2 \epsilon_n^2  - \frac{1}{2} \gamma_n^\xi \Omega^\xi(\hat\xi_n)^2 \,.
\end{align*}
In addition, by the AM-GM inequality we have
\begin{equation*}
    144 L \Omega^\xi(\hat\xi_n) \|\hat\xi_n - \xi_0\|_w \leq \frac{1}{4} \|\hat\xi_n - \xi_0\|_w^2 + 20736 L^2 \Omega^\xi(\hat\xi_n)^2 \epsilon_n^2 \,,
\end{equation*}
so therefore we have
\begin{align*}
    \frac{1}{4} \|\hat\xi_n - \xi_0\|_w^2 + \frac{1}{2} \gamma_n^\xi \Omega^\xi(\hat\xi_n)^2 &\leq 21240 L^2 \Omega^q(\hat\xi_n)^2 \epsilon_n^2 + 72 \epsilon_n + 3562 \delta_n^2 + 75 L^2 \Omega^\xi(\hat\xi_n)^2 \gamma_n^q \\
    &\qquad + 217 M^2 \gamma_n^q + 2 M^2 \gamma_n^\xi  - \frac{1}{2} \gamma_n^\xi \Omega^\xi(\hat\xi_n)^2 \,.
\end{align*}
Therefore, under our additional assumption that
\begin{equation*}
    \gamma_n^\xi \geq 42480 L^2 \epsilon_n^2 + 150 L^2 \gamma_n^q \,,
\end{equation*}
this the previous bound can be relaxed further to
\begin{align*}
    \frac{1}{4} \|\hat\xi_n - \xi_0\|_w^2 + \frac{1}{2} \gamma_n^\xi \Omega^\xi(\hat\xi_n)^2 &\leq 72 \epsilon_n + 3562 \delta_n^2 + 217 M^2 \gamma_n^q + 2 M^2 \gamma_n^\xi  \,,
\end{align*}
which immediately gives
\begin{equation*}
    \|\hat\xi_n - \xi_0\|_w \leq 17 \epsilon_n^{1/2} + 120 \delta_n + 30 M (\gamma_n^q)^{1/2} + 3 M (\gamma_n^\xi)^{1/2} \,,
\end{equation*}
and
\begin{equation*}
    \Omega^\xi(\hat\xi_n) \leq 12 \frac{\epsilon_n^{1/2}}{(\gamma_n^\xi)^{1/2}} + 85 \frac{\delta_n}{(\gamma_n^\xi)^{1/2}} + 21 M \frac{(\gamma_n^q)^{1/2}}{(\gamma_n^\xi)^{1/2}} + 2 M \,.
\end{equation*}
Also, under our assumptions on the sizes of $\gamma_n^\xi$ and $\gamma_n^q$ we can further bound the above by
\begin{equation*}
    \Omega^\xi(\hat\xi_n) \leq 3 M + \epsilon_n^{-1/2} + \epsilon_n^{-1} \delta_n \,.
\end{equation*}

Plugging this into the version of \Cref{lem:dual-saddle-point-bound} under the assumption that $\tilde\gamma_n^q \geq 42120 \epsilon_n^2$ then gives
\begin{align*}
    \|\hat q_n - q^\dagger\|_{2,2} &\leq 34 \Omega^\xi(\hat\xi_n) \epsilon_n + 68 \epsilon_n^{1/2} + 505 \delta_n + 120 M(\gamma_n^q)^{1/2} + 12 M (\gamma_n^\xi)^{1/2} + 4 M (\tilde \gamma_n^q)^{1/2} \\
    &\leq 102 M \epsilon_n + 102 \epsilon_n^{1/2} + 539 \delta_n + 120 M (\gamma_n^q)^{1/2} + 12 M (\gamma_n^\xi)^{1/2} + 4 M (\tilde\gamma_n^q)^{1/2} \,,
\end{align*}
as required.

\subsubsection*{Proofs of Intermediate Lemmas}

\begin{proof}[Proof of \Cref{lem:dual-feasibility}]

First, let some arbitrary $h \in \Hcal$ and $\xi_0' \in \Xi_0$ be given, and define
\begin{align*}
    J(h) &= \sup_{q \in \Qcal} \EE[\psi(h,q)] \\
    g(t) &= J(\xi_0' + t(h - \xi_0')) \,,
\end{align*}
Now, we have
\begin{align*}
    J(h) &= \sup_{q \in \Qcal} \EE \Big[ G(W;h,q) - m(W;h) - \frac{1}{2} q(T)^\top q(T) \Big] \\
    &= \sup_{q \in \Qcal} \EE \Big[ h(S)^\top k(S,T) q(T) - m(W;h) - \frac{1}{2} q(T)^\top q(T) \Big] \\
    &= \EE \Big[ \frac{1}{2} \Pi_\Qcal\Big[ k(S,T)^\top h(S) \mid T\Big]^\top \Pi_\Qcal\Big[ k(S,T)^\top h(S) \mid T\Big] - m(W;h) \Big] \,.
\end{align*}
Therefore, by the linearity of the projection operator, it easily follows that
\begin{align*}
    g'(t) &= \EE \Big[ \Pi_\Qcal\Big[ k(S,T)^\top (h - \xi_0')(S) \mid T\Big]^\top \Pi_\Qcal\Big[ k(S,T)^\top \Big( \xi_0'(S) + t(h - \xi_0')(S) \Big) \mid T\Big] - m(W;h - \xi_0') \Big] \\
    &= \EE \Big[ \Pi_\Qcal\Big[ k(S,T)^\top (h - \xi_0')(S) \mid T\Big]^\top q^\dagger(T) - m(W;h-\xi_0') \Big] \\
    &\qquad + t \EE\Big[ \Pi_\Qcal\Big[ k(S,T)^\top (h - \xi_0')(S) \mid T\Big]^\top \Pi_\Qcal \Big[ k(S,T)^\top (h - \xi_0')(S) \mid T \Big] \Big] \,,
\end{align*}
where in the second equality we apply \Cref{assump: new-nuisance,lemma: minimum-norm}. Also, we have 
\begin{equation*}
    g''(t) = \Big\| \Pi_\Qcal\Big[ k(S,T)^\top (h - \xi_0')(S) \mid T\Big] \Big\|_{2,2}^2 \,.
\end{equation*}
That is, $g''(t) \geq 0$ for all $t$, so it is convex.
Furthermore, since $q^\dagger \in \Qcal_0$, we have
\begin{align*}
    g'(0) &= \EE \Big[ \Pi_\Qcal\Big[ k(S,T)^\top (h - \xi_0')(S) \mid T\Big]^\top q^\dagger(T) - m(W;h-\xi_0') \Big] \\
    &= \EE\Big[ (h - \xi_0')(S)^\top k(S,T) q^\dagger(T) - m(W;h-\xi_0') \Big] \\
    &= \EE\Big[ G(W;h-\xi_0',q^\dagger) - m(W;h-\xi_0') \Big] \\
    &= 0 \,.
\end{align*}
Therefore, putting the above together $g(t)$ is minimized at $t=0$. Then, given that the above reasoning holds for arbitrary $h \in \Hcal$, and $\Hcal$ is linear (so $\{\xi_0' + h : h \in \Hcal\} = \Hcal\}$) it follows that $J(h)$ is minimized at $h=\xi_0'$. Since this analysis holds for arbitrary $\xi'_0 \in \Xi_0$, this implies that $\Xi_0$ is a subset of the minimizers of $J(h)$.

Conversely, suppose $h$ is also a minimizer of $J(h)$. Then, we must have $g(0) = g(1) = \inf_h J(h)$. Then, since $g(t)$ is quadratic in $t$, minimized at two different points, $g(t)$ must be constant, so
\begin{equation*}
    \EE\Big[ \Pi_\Qcal\Big[ k(S,T)^\top (h - \xi_0')(S) \mid T\Big] = 0 \,.
\end{equation*}
Then, given \Cref{assump: new-nuisance,lemma: minimum-norm}, we have
\begin{equation*}
    \EE\Big[ \Pi_\Qcal\Big[ k(S,T)^\top h(S) \mid T\Big] = \EE\Big[ \Pi_\Qcal\Big[ k(S,T)^\top \xi_0'(S) \mid T\Big] = q^\dagger(T) \,.
\end{equation*}
That is, $h \in \Xi_0$, which implies that any minimizer of $J$ is an element of $\Xi_0$. That is, we have shown that $\Xi_0$ is a subset of the minimizers of $J$ and vice versa, which establishes the first part of the lemma.

Now, for the second part of the lemma, let $\xi_0'$ be an arbitrary minimzer of $J(h)$, which may or may not be equal to $\xi_0'$. Then, by first-order optimality conditions (which must hold since $\Hcal$ is linear), we have we have $\left. \frac{d}{d t} \right|_{t=0} J(\xi_0' + t h) = 0$ for all $h \in \Hcal$. That is, applying similar logic to above, we have
\begin{equation*}
    \EE \Big[ \Pi_\Qcal\Big[ k(S,T)^\top h(S) \mid T\Big]^\top q_0(\xi_0')(T) - m(W;h) \Big] = 0 \qquad \forall h \in \Hcal \,,
\end{equation*}
where $q_0(\xi_0') \in \argmax_{q \in \Qcal} \EE[\psi(\xi_0',q)] = \Pi_\Qcal[k(S,T)^\top \xi_0'(S) \mid T]$. But then, following analogous reasoning as above, this is equivalent to $\EE[G(W;h,q_0(\xi_0')) - m(W;h)] = 0$ for all $h \in \Hcal$, so $q_0(\xi_0') \in \Qcal_0$. Then, since $q_0(\xi_0')(T) = P^\dagger \xi_0'$, and as reasoned previously there is a unique $q^\dagger \in \textup{range}(P^\dagger)$ such that $q^\dagger \in \Qcal_0$, it must be the case that $q_0(\xi_0') = q^\dagger$.

\end{proof}

\begin{proof}[Proof of \Cref{lem:dual-high-prob}]

First, note that
\begin{align*}
    \Big| (\EE_n - \EE)\Big[ \Psi_n(q - q^\dagger) \Big] \Big| &\leq \Big| (\EE_n - \EE) \Big[ G(W;\xi_0,q-q^\dagger) - q^\dagger(T)^\top (q - q^\dagger)(T) \Big] \Big| \\
    &\qquad + \Big| (\EE_n - \EE) \Big[ (q - q^\dagger)(T)^\top(q - q^\dagger)(T) \Big] \Big| \,,
\end{align*}
and therefore following identical arguments as in the proof of \Cref{lem:primal-high-prob} we have
\begin{equation*}
    \Big| (\EE_n - \EE)\Big[ \Psi_n(q - q^\dagger) \Big] \Big| \leq 54 \Omega^q(q-q^\dagger) \epsilon_n \|q - q^\dagger\|_{2,2} + 54 \Omega^q(q-q^\dagger)^2 \epsilon_n^2 \,,
\end{equation*}
which holds uniformly over all $q \in \Qcal_n$ with probability at least $1-2\zeta$, and similarly
\begin{equation*}
    \Big| (\EE_n - \EE)\Big[ \Psi_n(q - q^\dagger) \Big] \Big| \leq 54 \widetilde\Omega^q(q-q^\dagger) \epsilon_n \|q-q^\dagger\|_{2,2} + 54 \widetilde\Omega^q(q-q^\dagger)^2 \epsilon_n^2 \,,
\end{equation*}
which holds uniformly over all $q \in \widetilde\Qcal_n$ with probability at least $1-2\zeta$.

Furthermore, following similar reasoning as in the proof of \Cref{lem:primal-high-prob}, we have
\begin{equation*}
    \Big|(\EE_n - \EE) \Big[ \widetilde\Psi_n(\xi-\xi_0) \Big] \Big| \leq 36 \epsilon_n \|\xi-\xi_0\|_{2,2} + 36 \Omega^\xi(\xi-\xi_0) \epsilon_n^2 \,,
\end{equation*}
which holds for all $\xi \in \Xi_n$ with probability at least $1-\zeta$, and we have
\begin{equation*}
    \Big|(\EE_n - \EE) \Big[ G(W;\xi-\xi_0,q-q^\dagger) \Big] \Big| \leq 36 \epsilon_n \|q-q^\dagger\|_{2,2} + 36 \Omega^\xi(\xi-\xi_0) \Omega^q(q-q_0) \epsilon_n^2 \,,
\end{equation*}
which holds for all $q \in \Qcal_n$ and $\xi \in \Xi_n$ with probability at least $1-\zeta$, and
\begin{equation*}
    \Big|(\EE_n - \EE) \Big[ G(W;\xi-\xi_0,q-q^\dagger) \Big] \Big| \leq 36 \epsilon_n \|q-q^\dagger\|_{2,2} + 36 \Omega^\xi(\xi-\xi_0) \widetilde\Omega^q(q-q_0) \epsilon_n^2 \,,
\end{equation*}
which holds for all $q \in \widetilde\Qcal_n$ and $\xi \in \Xi_n$ with probability at least $1-\zeta$. 

Next, following similar reasoning as in the proof of \Cref{lem:primal-high-prob}, we have $\Omega^q(q-q^\dagger) \leq 2 \Omega^q(q)$, $\widetilde\Omega^q(q-q^\dagger) \leq 2 \widetilde\Omega^q(q)$, and $\Omega^\xi(\xi-\xi_0) \leq 2 \Omega^\xi(\xi)$. Plugging this into the above bounds, and taking a union bound over all of them, gives us our final desired result.

\end{proof}

\begin{proof}[Proof of \Cref{lem:dual-saddle-point-bound}]

The proof here is very similar to that of \Cref{lem:primal-interior-maximization}. We will prove it below for any arbitrary $\xi \in \Xi_n$ for the $\|q_n(\xi) - q^\dagger\|$ bound. Note that the $\|\tilde q_n(\xi) - q^\dagger\|$ bound also follows by an identical proof, replacing $\Qcal_n$ with $\widetilde Q_n$, $\|\cdot\|_\Qcal$ with $\|\cdot\|_{\widetilde\Qcal}$, and $\gamma_n^q$ with $\tilde\gamma_n^q$.

First, note that
\begin{align*}
    &\frac{1}{2} \|q_n(\xi) - q_0(\xi)\|_{2,2}^2 - \frac{1}{2} \|q_0(\xi) - q^\dagger\|_{2,2}^2 \\
    &= \frac{1}{2} \EE\Big[ \Big(q_n(\xi)(T) - \Pi_\Qcal[k(S,T)^\top \xi(S) \mid T]\Big)^\top \Big(q_n(\xi)(T) - \Pi_\Qcal[k(S,T)^\top \xi(S) \mid T]\Big) \\
    &\qquad - \Big(q^\dagger(T) - \Pi_\Qcal[k(S,T)^\top \xi(S) \mid T]\Big)^\top \Big(q^\dagger(T) - \Pi_\Qcal[k(S,T)^\top \xi(S) \mid T]\Big) \Big] \\
    &= \EE\Big[ \Big( q_0(\xi)(T)q^\dagger(T) - \frac{1}{2} q^\dagger(T)^\top q^\dagger(T) \Big) - \Big(q_0(\xi)(T)q_n(\xi)(T) - \frac{1}{2} q_n(\xi)(T)^\top q_n(\xi)(T) \Big) \Big] \\
    &= \EE\Big[ \Big( \xi(S)^\top k(S,T) q^\dagger(T) - \frac{1}{2} q^\dagger(T)^\top q^\dagger(T) \Big) \\
    &\qquad - \Big( \xi(S)^\top k(S,T) q_n(\xi)(T) - \frac{1}{2} q_n(\xi)(T)^\top q_n(\xi)(T) \Big) \Big] \\
    &= \EE\Big[ \psi(\xi,q^\dagger) - \psi(\xi,q_n(\xi)) \Big] \,.
\end{align*}

Furthermore, by the optimality of $q_n(\xi)$ for the empirical minimax problem at $h=\xi$, we have 
\begin{equation*}
    \EE_n \Big[ \psi(\xi,q_n(\xi)) \Big] - \gamma_n^q \|q_n(\xi)\|_\Qcal^2 \geq \EE_n\Big[ \psi(\xi,\Pi_n q^\dagger) \Big] - \gamma_n^q \|\Pi_n q^\dagger\|_\Qcal^2 \,,
\end{equation*}
and therefore
\begin{align*}
    \EE_n\Big[ \psi(\epsilon, q_n(\xi)) - \psi(\epsilon, q^\dagger) \Big] - \gamma_n^q \|q_n(\xi)\|_\Qcal^2 + \gamma_n^q \|\Pi_n q^\dagger\|_\Qcal^2 + \EE_n \Big[ \psi(\epsilon,q^\dagger) - \psi(\epsilon, \Pi_n q^\dagger) \Big] \geq 0
\end{align*}
It follows that 
\begin{align*}
    &\frac{1}{2} \|q_n(\xi) - q_0(\xi)\|_{2,2}^2 - \frac{1}{2} \|q_0(\xi) - q^\dagger\|_{2,2}^2 = -\EE\Big[ \psi(\xi,q_n(\xi)) - \psi(\xi,q^\dagger) \Big] \\
    &\leq (\EE_n - \EE)\Big[ \psi(\xi,q_n(\xi)) - \psi(\xi,q^\dagger) \Big] + \EE_n\Big[ \psi(\xi,q^\dagger) - \psi(\xi,\Pi_n q^\dagger) \Big] - \gamma_n^q \|q_n(\xi)\|_\Qcal^2 + \gamma_n^q \|\Pi_n q^\dagger\|_\Qcal^2 \\
    &\leq (\EE_n - \EE)\Big[ \psi(\xi,q_n(\xi)) - \psi(\xi,q^\dagger) + \psi(\xi,q^\dagger) - \psi(\xi,\Pi_n q^\dagger) \Big] \\
    &\qquad + \EE\Big[ \psi(\xi,q^\dagger) - \psi(\xi,\Pi_n q^\dagger) \Big] - \gamma_n^q \|q_n(\xi)\|_\Qcal^2 + \gamma_n^q \|\Pi_n q^\dagger\|_\Qcal^2 \\
    &= (\EE_n-\EE)\Bigg[ G\Big(W;\xi-\xi_0,q_n(\xi)-q^\dagger\Big) - G\Big(W;\xi-\xi_0,\Pi_n q^\dagger-q^\dagger\Big) + \Psi_n\Big(q_n(\xi) - q^\dagger\Big) \\
    &\qquad\qquad\qquad - \Psi_n\Big(\Pi_n q^\dagger - q^\dagger\Big) \Bigg] + \EE\Big[ \psi(\xi,q^\dagger) - \psi(\xi,\Pi_n q^\dagger) \Big] - \gamma_n^q \|q_n(\xi)\|_\Qcal^2 + \gamma_n^q \|\Pi_n q^\dagger\|_\Qcal^2 \,,
\end{align*}
where $\Psi_n$ is defined as in the statement of \Cref{lem:dual-high-prob}.

Moreover, we can further bound  $\EE\Big[ \psi(\xi,q^\dagger) - \psi(\xi,\Pi_n q^\dagger) \Big]$ as follows. 
\begin{align*}
    &\EE\Big[ \psi(\xi,q^\dagger) - \psi(\xi,\Pi_n q^\dagger) \Big] \\ 
    &= \EE\Big[ -G(\xi, \Pi_n q^\dagger - q^\dagger) + \frac{1}{2} (\Pi_n q^\dagger - q^\dagger)(T)^\top (\Pi_n q^\dagger - q^\dagger)(T) + q^\dagger(T)^\top (\Pi_n q^\dagger - q^\dagger)(T) \Big] \\
    &\le \EE\Big[ \Big(q^\dagger(T) - k(S,T)^\top \xi(S)\Big)^\top (\Pi_n q^\dagger - q^\dagger)(T) \Big] + \frac{1}{2} \delta_n^2 \\
    &\leq \EE\Big[ \Big(q^\dagger(T) - k(S,T)^\top \xi_0(S)\Big)^\top (\Pi_n q^\dagger - q^\dagger)(T) \Big] - \EE\Big[ G(W;\xi - \xi_0,\Pi_n q^\dagger - q^\dagger) \Big] + \frac{1}{2} \delta_n^2 \\
    &= - \EE\Big[ G(W;\xi - \xi_0,\Pi_n q^\dagger - q^\dagger) \Big] + \frac{1}{2} \delta_n^2 \\
    &\leq \|\xi - \xi_0\|_w \delta_n + \frac{1}{2} \delta_n^2 \\
    &\leq \frac{1}{2} \|\xi - \xi_0\|_w^2 + \delta_n^2 \,,
\end{align*}
where in the above we apply the fact that $q^\dagger(T) = q_0(\xi_0)(T) = \Pi_\Qcal[k(S,T)^\top \xi_0(S) \mid T]$, along with \Cref{assum:universal-approximation-q-general}, and the AM-GM inequality.

Now, putting the previous bounds together and plugging in the bounds from the high-probability event of \Cref{lem:dual-high-prob}, as well as the fact that $\Omega^q(\Pi_n q^\dagger) \leq M \leq \Omega^q(q_n(\xi))$, $-\|q\|_\Qcal^2 \leq -\Omega^q(q)^2 + M^2$, and the AM-GM inequality, we get
\begin{align*}
    &\frac{1}{2} \|q_n(\xi) - q_0(\xi)\|_{2,2}^2 - \frac{1}{2} \|q_0(\xi) - \Pi_n q^\dagger\|_{2,2}^2 \\
    &\leq 144 \epsilon_n \Big( \Omega^q(q_n(\xi)) \|q_n(\xi) - q^\dagger\|_{2,2} + M \|\Pi_n q^\dagger - q^\dagger\|_{2,2} \Big) \\
    &\qquad + \Big( 432 \Omega^q(q_n(\xi))^2 + 288 \Omega^q(q_n(\xi)) \Omega^\xi(\xi) \Big)\epsilon_n^2 + \frac{1}{2} \|\xi - \xi_0\|_w^2 + \delta_n^2 - \gamma_n^q \|q_n(\xi)\|_\Qcal^2 + \gamma_n^q \|\Pi_n q^\dagger\|_\Qcal^2 \\
    &\leq 144 \epsilon_n \Omega^q(q_n(\xi)) \|q_n(\xi) - q^\dagger\|_{2,2} + 648 \Omega^q(q_n(\xi))^2 \epsilon_n^2 + 144 \Omega^\xi(\xi)^2 \epsilon_n^2 \\
    &\qquad + 73 \delta_n^2 + \frac{1}{2} \|\xi-\xi_0\|_w^2 + 2 M^2 \gamma_n^q - \gamma_n^q \Omega^q(q_n(\xi))^2 \,.
\end{align*}

Next, applying the above bound, we have
\begin{align*}
    \frac{1}{2} \|q_n(\xi) - q^\dagger\|_{2,2}^2 &\leq \|q_n(\xi) - q_0( \xi_n)\|_{2,2}^2 + \|q_0( \xi_n) - q^\dagger\|_{2,2}^2 \\
    &\leq 2\Big( \frac{1}{2} \|q_n(\xi) - q_0( \xi_n)\|_{2,2}^2 - \frac{1}{2} \|q_0( \xi_n) - q^\dagger\|_{2,2}^2 \Big) + 2 \|q_0( \xi_n) - q^\dagger\|_{2,2}^2 \\
    &\leq 288 \epsilon_n \Omega^q(q_n(\xi)) \|q_n(\xi) - q^\dagger\|_{2,2} + 1296 \Omega^q(q_n(\xi))^2 \epsilon_n^2 \\
    &\qquad + 288 \Omega^\xi(\xi)^2 \epsilon_n^2 + 146 \delta_n^2 + 3 \|\xi-\xi_0\|_w^2 + 4 M^2 \gamma_n^q - 2 \gamma_n^q \Omega^q(q_n(\xi))^2 \,,
\end{align*}
where in the third inequality we apply the fact that $q^\dagger = q_0(\xi_0)$, so therefore $\|q_0(\xi) - q^\dagger\|_{2,2} = \|\xi-\xi_0\|_w$.
Next, noting that $288 \epsilon_n \Omega^q(q_n(\xi)) \|q_n(\xi) - q^\dagger\|_{2,2} \leq \frac{1}{4} \|q_n(\xi) - q^\dagger\|_{2,2}^2 + 82944 \epsilon_n^2 \Omega^q(q_n(\xi))^2$ by the AM-GM inequality, the previous bound reduces further to
\begin{align*}
    \frac{1}{4} \|q_n(\xi) - q^\dagger\|_{2,2}^2 &\leq 84240 \Omega^q(q_n(\xi))^2 \epsilon_n^2 + 288 \Omega^\xi(\xi)^2 \epsilon_n^2 + 146 \delta_n^2 + 3 \|\xi-\xi_0\|_w^2 \\
    &\qquad + 4 M^2 \gamma_n^q - 2 \gamma_n^q \Omega^q(q_n(\xi))^2 \,,
\end{align*}

Now, regardless of the value of $\|q_n(\xi)\|_\Qcal$, and noting that $\sqrt{x+y} \leq \sqrt{x} + \sqrt{y}$ for non-negative $x$ and $y$, the above bound reduces to
\begin{equation*}
    \|q_n(\xi) - q^\dagger\|_{2,2} \leq \Big( 581 \Omega^q(q_n(\xi)) + 34 \Omega^\xi(\xi)\Big) \epsilon_n + 25 \delta_n + 4 \|\xi - \xi_0\|_w + 4 M (\gamma_n^q)^{1/2} \,,
\end{equation*}
which is the first required bound.
Furthermore, in the case that $\gamma_n^q \geq 42120 \epsilon_n^2$, the terms involving $\Omega^q(q_n(\xi))$ vanish, which gives us our second required bound
\begin{equation*}
    \|q_n(\xi) - q^\dagger\|_{2,2} \leq 34 \Omega^\xi(\xi) \epsilon_n + 25 \delta_n + 4 \|\xi - \xi_0\|_w + 4 M (\gamma_n^q)^{1/2} \,,
\end{equation*}

\end{proof}

\begin{proof}[Proof of \Cref{lem:dual-strong-convexity-bound}]

Define $J(h) = \sup_{q \in \Qcal} \EE[\psi(h,q)]$. As argued in the proof of \Cref{lem:dual-feasibility}, we have
\begin{equation*}
    J(h) = \EE \Big[ \frac{1}{2} \Pi_\Qcal\Big[ k(S,T)^\top h(S) \mid T\Big]^\top \Pi_\Qcal\Big[ k(S,T)^\top h(S) \mid T\Big] - m(W;h) \Big] \,.
\end{equation*}
Now, let some arbitrary $\xi \in \Hcal$ be given, and define
\begin{equation*}
    g(t) = J(\xi_0 + t(\xi - \xi_0)) \,.
\end{equation*}
Then, as argued in the proof of \Cref{lem:dual-feasibility}, we have $g'(0) = 0$, and $g''(t) = \|\Pi_\Qcal[k(S,T)^\top (\xi - \xi_0)(S) \mid T]\|_{2,2}^2 = \|\xi - \xi_0\|_w^2$ for all $t$. Therefore, by strong convexity we have
\begin{align*}
    &g(1) - g(0) \geq \frac{1}{2} \|\xi - \xi_0\|_w^2 \\
    &\iff \frac{1}{2} \|\xi - \xi_0\|_w^2 \leq \sup_{q \in \Qcal} \EE[\psi(\xi,q)] - \sup_{q \in \Qcal} \EE[\psi(\xi_0,q)] \,,
\end{align*}
where the second equality follows by noting that $g(0) = J(\xi_0)$, $g(1) = J(\xi)$, and by plugging in the definition of $J$.
Since this holds for arbitrary $\xi \in \Hcal$, we are done.

\end{proof}

\begin{proof}[Proof of \Cref{lem:dual-minimax-erm-bound}]

The proof here is very similar to that of \Cref{lem:primal-minimax-erm-bound}. First, we have
\begin{align*}
    &\sup_{q \in \Qcal} \EE[\psi(\hat \xi_n,q)] - \sup_{q \in \Qcal} \EE[\psi(\xi_0,q)] \\
    &\leq \EE[\psi(\hat \xi_n,q_0(\hat \xi_n))] - \EE[\psi(\xi_0,q_n(\Pi_n \xi_0))] \\
    &\leq \EE[\psi(\hat \xi_n,q_0(\hat \xi_n))] - \EE[\psi(\xi_0,q_n(\Pi_n \xi_0))]  - \sup_{q \in \Qcal_n} \Big( \EE_n[\psi(\hat \xi_n,q)] - \gamma_n^q \|q\|_\Qcal^2 + \gamma_n^\xi \|\hat\xi_n\|_\Xi^2 \Big) \\
    &\qquad + \sup_{q \in \Qcal_n} \Big( \EE_n[\psi(\Pi_n \xi_0,q)] - \gamma_n^q \|q\|_\Qcal^2 + \gamma_n^\xi \|\Pi_n \xi_0\|_\Xi^2 \Big) \\
    &= \EE[\psi(\hat \xi_n,\Pi_n q_0(\hat \xi_n))] - \EE[\psi(\xi_0,q_n(\Pi_n \xi_0))]  - \sup_{q \in \Qcal_n} \Big( \EE_n[\psi(\hat \xi_n,q)] - \gamma_n^q \|q\|_\Qcal^2 \Big) \\
    &\qquad + \EE_n[\psi(\xi_0,q_n(\Pi_n \xi_0))] - \gamma_n^q \|q_n(\Pi_n \xi_0)\|_\Qcal^2 + M^2 \gamma_n^\xi - \gamma_n^\xi \|\hat\xi_n\|_\Xi^2 + \Ecal_3 + \Ecal_4 \\
    &\leq (\EE - \EE_n)\Big[\psi(\hat \xi_n,\Pi_n q_0(\hat \xi_n)) - \psi(\xi_0,q_n(\Pi_n \xi_0))\Big] + \gamma_n^q \|\Pi_n q_0(\hat\xi_n)\|_\Qcal^2 + M^2 \gamma_n^\xi \\
    &\qquad - \gamma_n^q \| q_n(\Pi_n \xi_0) \|_\Qcal^2 - \gamma_n^\xi \|\hat\xi_n\|_\Xi^2 + \Ecal_3 + \Ecal_4 \,,
\end{align*}
where the first and third inequalities follow by relaxing the supremum in the negative terms, and the second inequality follows by the optimality of $\hat \xi_n$ for the empirical minimax objective, and where
\begin{align*}
    \Ecal_3 &= \EE\Big[\psi(\hat \xi_n, q_0(\hat \xi_n)) - \psi(\hat \xi_n, \Pi_n q_0(\hat \xi_n)) \Big] \\
    \Ecal_4 &= \EE_n[\psi(\Pi_n \xi_0, q_n(\Pi_n \xi_0))] - \EE_n[\psi(\xi_0, q_n(\Pi_n \xi_0))]
\end{align*}

Now, let us bound the error terms $\Ecal_3$ and $\Ecal_4$. For the first, using the shorthand $\Delta = q_0(\hat \xi_n) - \Pi_n q_0(\hat \xi_n)$, we have
\begin{align*}
    \Ecal_3 &= \EE\Big[ G(W;\hat \xi_n, \Delta) - \frac{1}{2} q_0(\hat \xi_n)(T)^\top q_0(\hat \xi_n)(T) + \frac{1}{2} \Pi_n q_0(\hat \xi_n)(T)^\top \Pi_n q_0(\hat \xi_n)(T) \Big] \\
    &= \EE\Big[ G(W;\hat \xi_n, \Delta) - q_0(\hat \xi_n)(T)^\top \Delta(T) + \frac{1}{2} \Delta(T)^\top \Delta(T) \Big] \\
    &= \EE\Big[ \frac{1}{2} \Delta(T)^\top \Delta(T) \Big] \leq \frac{1}{2} \delta_n^2 \,,
\end{align*}
where the second last step follows since $\EE[G(W;\xi,\Delta)] = \EE[q_0(\xi)(T)^\top \Delta(T)]$ for any $\xi$, and the final inequality follows by \Cref{assum:universal-approximation-q-general}. 

Next, for the second term above, we can bound
\begin{align*}
    \Ecal_4 &\leq \EE_n\Big[ \psi(\Pi_n \xi_0,  q_n(\Pi_n \xi_0) -  \psi(\xi_0, q_n(\Pi_n \xi_0) \Big] \\
    &= \EE_n\Big[ G(W;\Pi_n \xi_0 - \xi_0, q_n(\Pi_n \xi_0) - q^\dagger) + \widetilde\psi_n(\Pi_n \xi_0 - \xi_0) \Big] \\
    &= (\EE_n - \EE)\Big[ G(W;\Pi_n \xi_0 - \xi_0, q_n(\Pi_n \xi_0) - q^\dagger) + \widetilde\psi_n(\Pi_n \xi_0 - \xi_0) \Big] \\
    &\qquad + \EE\Big[ G(W;\Pi_n \xi_0 - \xi_0; q_n(\Pi_n \xi_0)) - m(W;\Pi_n \xi_0 - \xi_0) \Big] \\
    &= (\EE_n - \EE)\Big[ G(W;\Pi_n \xi_0 - \xi_0, q_n(\Pi_n \xi_0) - q^\dagger) + \widetilde\psi_n(\Pi_n \xi_0 - \xi_0) \Big] \\
    &\qquad + \EE\Big[ G(W;\Pi_n \xi_0 - \xi_0; q_n(\Pi_n \xi_0) - q^\dagger) \Big] \\
    &\leq 36 \epsilon_n \|q_n(\Pi_n \xi_0) - q^\dagger\|_{2,2} + 36 \epsilon_n \delta_n + 144 \Omega^q(q_n(\Pi_n \xi_0)) \Omega^\xi(\Pi_n \xi_0) \epsilon_n^2 + 72 \Omega^\xi(\Pi_n \xi_0) \epsilon_n^2 \\
    &\qquad + \delta_n \|q_n(\Pi_n \xi_0) - q^\dagger\|_{2,2} \\
    &\leq (36 \epsilon_n + \delta_n) \|q_n(\Pi_n \xi_0) - q^\dagger\|_{2,2} + 234 \Omega^q(q_n(\Pi_n \xi_0))^2 \epsilon_n^2 + 18 \delta_n^2
\end{align*}
where the first inequality follows by relaxing the supremum in the negative term, 
the second and third  equalities follow from the fact that $\Eb{G(W; \xi, q^\dagger) - m(W; \xi)} = 0$ for $\xi = \Pi_n\xi_0 - \xi_0$, 
the second inequality follows from the assumed high-probability event of \Cref{lem:dual-high-prob} along with \Cref{assum:universal-approximation-q-general}, and the fact that $\EE[G(W;h,q)] = \EE[(Ph)(S)^\top q(T)]$,
and the third inequality follows by the AM-GM inequality along with $1 \leq \Omega^\xi(\Pi_n \xi_0) \leq M \leq \Omega^q(q_n(\Pi_n \xi_0))$.
Furthermore, applying \Cref{lem:dual-saddle-point-bound}, along with the same bounds used above, we have
\begin{equation*}
    \|q_n(\Pi_n\xi_0) - q^\dagger\|_{2,2} \leq 615 \Omega^q(q_n(\Pi_n\xi_0)) \epsilon_n + 29 \delta_n + 4 M (\gamma_n^q)^{1/2} \,,
\end{equation*}
and so combining these with teh AM-GM inequality gives
\begin{equation*}
    \Ecal_4 \leq 22734 \Omega^q(q_n(\Pi_n \xi_0))^2 \epsilon_n^2 + 357 \delta_n^2 + 74 M^2 \gamma_n^q
\end{equation*}

Finally, putting all of the above together, we get our final required result of
\begin{align*}
    &\sup_{q \in \Qcal} \EE[\psi(\hat \xi_n,q)] - \sup_{q \in \Qcal} \EE[\psi(\xi_0,q)] \\
    &\leq (\EE - \EE_n)\Big[\psi(\hat \xi_n,\Pi_n q_0(\hat \xi_n)) - \psi(\xi_0,q_n(\Pi_n \xi_0)\Big] + 22734 \Omega^q(q_n(\Pi_n\xi_0))^2 \epsilon_n^2 \\
    &\qquad + 358 \delta_n^2 + 75 \Omega^q(\Pi_n q_0(\hat\xi_n))^2 \gamma_n^q + M^2 \gamma_n^\xi - \gamma_n^q \|q_n(\Pi_n \xi_0)\|_\Qcal^2 - \gamma_n^\xi \|\hat\xi_n\|_\Xi^2
\end{align*}

\end{proof}

\begin{proof}[Proof of \Cref{lem:dual-stochastic-equicontinuity-bound}]

First, given the decomposition from the proof overview, we have
\begin{align*}
    &\psi(\hat \xi_n, \Pi_n q_0(\hat \xi_n)) - \psi(\xi_0, q_n(\Pi_n \xi_0)) \\
    &= \Psi_n\Big(\Pi_n q_0(\hat \xi_n) - q^\dagger \Big) - \Psi_n\Big(q_n(\Pi_n \xi_0) - q^\dagger \Big) \\
    &\qquad + \widetilde\Psi_n(\hat\xi_n - \xi_0) + G\Big(W; \hat \xi_n - \xi_0, \Pi_n q_0(\hat \xi_n) - q^\dagger \Big) \,.
\end{align*}
where $\Psi_n$ and $\widetilde\Psi_n$ are defined as in the statement of \Cref{lem:dual-high-prob}.

Now, applying the bounds of \Cref{lem:dual-high-prob} to the above, under its high-probability event we get
\begin{align*}
    &\Big| (\EE_n - \EE) \Big[ \psi(\hat \xi_n, \Pi_n q_0(\hat \xi_n)) - \psi(\xi_0, q_n(\Pi_n \xi_0) \Big] \Big| \\
    &\leq 144 \Omega^q(\Pi_n q_0(\hat\xi_n)) \|\Pi_n q_0(\hat\xi_n) - q^\dagger\|_{2,2} \epsilon_n + 108 \Omega^q(q_n(\Pi_n \xi_0)) \|q_n(\Pi_n \xi_0) - q^\dagger\|_{2,2} \epsilon_n + 36 \|\hat\xi_n - \xi_0\|_{2,2} \epsilon_n \\
    &\qquad + 216 \Omega^q(\Pi_n q_0(\hat\xi_n))^2 \epsilon_n^2 + 216 \Omega^q(q_n(\Pi_n \xi_0))^2 \epsilon_n^2 + 144 \Omega^q(\Pi_n q_0(\hat\xi_n)) \Omega^\xi(\hat\xi_n) \epsilon_n^2 + 72 \Omega^\xi(\hat\xi_n) \epsilon_n^2 \\
    &\leq 144 \Omega^q(\Pi_n q_0(\hat\xi_n)) \|\Pi_n q_0(\hat\xi_n) - q^\dagger\|_{2,2} \epsilon_n + 108 \Omega^q(q_n(\Pi_n \xi_0)) \|q_n(\Pi_n \xi_0) - q^\dagger\|_{2,2} \epsilon_n + 72 \epsilon_n \\
    &\qquad + \Big( 288 \Omega^q(\Pi_n q_0(\hat\xi_n))^2 + 216 \Omega^q(q_n(\Pi_n \xi_0))^2 + 144 \Omega^\xi(\hat\xi_n)^2 \Big) \epsilon_n^2 \,,
\end{align*}
where the second inequality follows from the AM-GM inequality and $\|\xi\|_\infty \leq 1$ for all $\xi \in \{\xi_0\} \cup \Xi_n$. Furthermore, by \Cref{lem:dual-saddle-point-bound} we have
\begin{align*}
    \|q_n(\Pi_n \xi_0) - q^\dagger\|_{2,2} &\leq 615 \Omega^q(q_n(\Pi_n \xi_0)) \epsilon_n + 25 \delta_n + 4\|\Pi_n \xi_0 - \xi_0\|_w + 4 M (\gamma_n^q)^{1/2} \\
    &\leq 615 \Omega^q(q_n(\Pi_n \xi_0)) \epsilon_n + 29 \delta_n + 4 M (\gamma_n^q)^{1/2} \,,
\end{align*}
which holds under the same high-probability event.
Therefore, combining together all of the above, applying the AM-GM inequality again, we have
\begin{align*}
    &\Big| (\EE_n - \EE) \Big[ \psi(\hat \xi_n, \Pi_n q_0(\hat \xi_n)) - \psi(\xi_0, q_n(\Pi_n \xi_0) \Big] \Big| \\
    &\leq 144 \Omega^q(\Pi_n q_0(\hat\xi_n)) \epsilon_n \|\hat \xi_n - \xi_0\|_w  + 72 \epsilon_n + 3204 \delta_n^2 + 216 M^2 \gamma_n^q \\ 
    &\qquad + \Big( 360 \Omega^q(\Pi_n q_0(\hat\xi_n))^2 + 68418 \Omega^q(q_n(\Pi_n \xi_0))^2 + 144 \Omega^\xi(\hat\xi_n)^2 \Big) \epsilon_n^2  \,,
\end{align*}
which is our required bound.

\end{proof}

\section{Proofs for Debiased Inference Results}

Here, we provide proofs for our debiased inference results. Specifically, we provide proofs for \Cref{thm:dml-asymp} and \Cref{lem:inference}, along with thier general analogues \Cref{thm:dml-asymp-general} and \Cref{lem:inference-general}. As with the proofs of the minimax estimation results, since the former are special cases of the latter, we only provide explicit proofs of the latter.

\begin{proof}[Proof of \Cref{thm:dml-asymp-general}]

For any function $f(W)$, denote $\hP f = \int f(w)p(w)\diff w$, and $\hP_{n, k} f =\sum_{i\in\Ical_k}f(W_i)/\abs{\Ical_k}$. 
Also denote $$\hat\theta_k = \frac{1}{\abs{\Ical_k}}\sum_{i \in \Ical_k} \psi(W_i; \hat h^{(k)}, \hat q^{(k)}).$$
It is easy to verify that we have 
\begin{align*}
\hat\theta_k -\theta^\star 
   &= \hP_{n, k}\bracks{\psi(W;  h^\dagger,  q^\dagger) -\theta^\star} \\
   &+ \prns{\hP_{n, k} - \hP}\bracks{\psi(W;  \hat h^{(k)},  \hat q^{(k)}) - \psi(W;  h^\dagger,  q^\dagger)} + {\hP}\bracks{\psi(W;  \hat h^{(k)},  \hat q^{(k)}) - \psi(W;  h^\dagger,  q^\dagger)}.
\end{align*}
By simple algebra, we can show that there exists a universal constant $c$ such that 
\begin{align*}
\hP\bracks{\prns{\psi(W;  \hat h^{(k)},  \hat q^{(k)}) - \psi(W;  h_R,  q_R)}^2} \le c \|\hat h^{(k)} - h_R\|_2^2 + c\|\hat q - q\|_2^2 = o_p\prns{1} 
\end{align*}
since $\|\hat h^{(k)} - h_R\|_{2}^2 = o_p(1)$ and $\|\hat q^{(k)} - q_R\|_{2}^2 = o_p(1)$ according to \Cref{thm:q-estimator-bound-general,thm:h-estimator-bound-general}.

Thus by Markov inequality, we have
\begin{align*}
\abs{\prns{\hP_{n, k} - \hP}\bracks{\psi(W;  \hat h^{(k)},  \hat q^{(k)}) - \psi(W;  h_R,  q_R)}}  = o_p\prns{\frac{1}{\sqrt{\abs{\Ical_k}}}} = o_p\prns{\frac{1}{\sqrt{n}}}.
\end{align*}

Moreover, according to \Cref{lem:general-dr}, 
\begin{align*}
\abs{{\hP}\bracks{\psi(W;  \hat h^{(k)},  \hat q^{(k)}) - \psi(W;  h^\dagger,  q^\dagger)}} \le  \|P[\hat h^{(k)} - h_R]\|_2\|\hat q^{(k)} - q_R\|_2.
\end{align*}
According to  \Cref{thm:q-estimator-bound-general,thm:h-estimator-bound-general}, we have 
   \begin{align*}
        &\|P(\hat h_n - h_0)\|_{2,2} \leq 395 \epsilon_n + 43 \delta_n + 2 \mu_n^{1/2} \,, \\
        &\|\hat q_n - q^\dagger\|_{2} \leq 771 \epsilon_n^{1/2} + 30 \delta_n  \,,
    \end{align*}
where $\epsilon_n = r_n + c_2 \sqrt{\log(c_1/\zeta)/n}$. 

It follows that 
\begin{align*}
\abs{{\hP}\bracks{\psi(W;  \hat h^{(k)},  \hat q^{(k)}) - \psi(W;  h^\dagger,  q^\dagger)}}  = O_p\prns{\epsilon_n^{3/2} + \delta_n\epsilon_n^{1/2}  + \delta_n\epsilon_n + \delta_n^2 + \delta_n\mu_n^{1/2} + \epsilon_n^{1/2}\mu_n^{1/2}}.
\end{align*}
Therefore, when $r_n = o(n^{-1/3}), \delta_n = o(n^{-1/4}), \delta_nr_n^{1/2} = o(n^{-1/2}), \mu_n r_n = o(n^{-1})$ and $\mu_n \delta_n^2 = o(n^{-1})$, we have 
\begin{align*}
\abs{{\hP}\bracks{\psi(W;  \hat h^{(k)},  \hat q^{(k)}) - \psi(W;  h^\dagger,  q^\dagger)}} = o_p(n^{-1/2}), ~~ \text{for } k = 1, \dots, K.
\end{align*}

Therefore, 
\begin{align*}
\hat\theta_k -\theta^\star 
   &= \hP_{n, k}\bracks{\psi(W;  h^\dagger,  q^\dagger) -\theta^\star} + o_p(n^{-1/2}).
\end{align*}
This implies that 
\begin{align*}
\sqrt{n}\prns{\hat\theta_n - \theta^\star} = \frac{1}{\sqrt{n}}\sum_{i=1}^n \prns{\psi(W_i;  h^\dagger,  q^\dagger) -\theta^\star} + o_p(1). 
\end{align*}
By Central Limit Theorem, we have that when $n \to \infty$,
\begin{align*}
\sqrt{n}\prns{\hat\theta_n - \theta^\star} \rightsquigarrow \mathcal{N}(0, \sigma^2_0),
\end{align*}
where 
\begin{align*}
\sigma^2_0 = \Eb{\prns{\psi(W;  h^\dagger,  q^\dagger) -\theta^\star}^2}. 
\end{align*}
\end{proof}

\begin{proof}[Proof of \Cref{lem:inference-general}]

Let arbitrary fold $k \in [K]$ be given, and let $n_k$ denot the size of the fold. Define
\begin{equation*}
    \hat\sigma_{n,k}^2 = \EE_{n,k} \Big[  \Big( \hat\theta_n - \psi(W;\hat h^{(k)}, \hat q^{(k)}) \Big)^2 \Big]\,,
\end{equation*}
where $\EE_{n,k}$ denotes the empirical expectation operator over the data in the $k$'th fold. Since $\hat\sigma_n^2$ is given by the weighted average of $\hat\sigma_{n,k}^2$ for all $k \in [K]$, it is sufficient to argue that $\hat\sigma_{n,k}^2 \to \sigma_0^2$ in probability for arbitrary $k$.

Now, we have
\begin{align*}
    \hat\sigma_{n,k}^2 - \sigma_0^2 &= (\EE_{n,k} - \EE) \Big[ \Big( \theta_0 - \psi(W; h_0, q_0) \big)^2 \Big] \\
    &\qquad + \EE_{n,k} \Big[ \Big( \hat\theta_n - \psi(W;\hat h^{(k)}, \hat q^{(k)}) \Big)^2 - \Big( \theta_0 - \psi(W; h_0, q_0) \big)^2 \Big]
\end{align*}
The former of these terms converges to zero in probability by the law of large numbers, so we can focus on the second of these terms, which we can re-arrange as
\begin{align*}
    &\EE_{n,k} \Big[ \Big( \hat\theta_n - \psi(W;\hat h^{(k)}, \hat q^{(k)}) \Big)^2 - \Big( \theta_0 - \psi(W; h_0, q_0) \big)^2 \Big] \\
    &= \EE_{n,k} \Big[ (\hat\theta_n - \theta_0) \Big(\hat\theta_n + \theta_0 - \psi(W;\hat h^{(k)},\hat q^{(k)}) - \psi(W;h_0,q_0) \Big)\Big] \\
    &\qquad - \EE_{n,k} \Big[ \Big(\psi(W;h_0,q_0) - \psi(W;\hat h^{(k)},\hat q^{(k)}) \Big) \Big(\hat\theta_n + \theta_0 - \psi(W;\hat h^{(k)},\hat q^{(k)}) - \psi(W;h_0,q_0) \Big)\Big] \\
\end{align*}
Furthermore, by \Cref{assum:boundedness-general,assum:bounded-estimation-h-general,assum:bounded-estimation-q-general}, it easily follows that $|\theta_0| \leq 1$, $|\hat\theta_n| \leq 3$, $\|\psi(W;\hat h^{(k)},\hat q^{(k)})\|_\infty \leq 3$, and $\|\psi(W;h_0,q_0)\|_\infty \leq 3$. Therefore, we have
\begin{equation*}
    \Big\| \hat\theta_n + \theta_0 - \psi(W;\hat h^{(k)},\hat q^{(k)}) - \psi(W;h_0,q_0) \Big\|_\infty \leq 10 \,,
\end{equation*}
and so
\begin{align*}
    &\Big| \EE_{n,k} \Big[ \Big( \hat\theta_n - \psi(W;\hat h^{(k)}, \hat q^{(k)}) \Big)^2 - \Big( \theta_0 - \psi(W; h_0, q_0) \big)^2 \Big] \Big| \\
    &\leq 10 |\hat\theta_n - \theta_0| + 10 \EE_{n,k} \Big[ \Big| \psi(W;h_0,q_0) - \psi(W;\hat h^{(k)},\hat q^{(k)}) \Big|\Big]
\end{align*}
Now, by \Cref{thm:dml-asymp-general} we know that $\hat\theta_n - \theta_0 \to 0$ in probability. Furthermore, by \Cref{thm:h-estimator-bound-general,thm:q-estimator-bound-general}, along with the continuity of $\psi$ in $h$ and $q$ by \Cref{assum:boundedness-general}, we know that $\|\psi(W;h_0,q_0) - \psi(W;\hat h^{(k)},\hat q^{(k)})\|_1 \to 0$ in probability. Therefore, we have
\begin{align*}
    &\Big| \EE_{n,k} \Big[ \Big( \hat\theta_n - \psi(W;\hat h^{(k)}, \hat q^{(k)}) \Big)^2 - \Big( \theta_0 - \psi(W; h_0, q_0) \big)^2 \Big] \Big| \\
    &\leq 10 (\EE_{n,k} - \EE)[\Delta_n] + o_p(1) \,,
\end{align*}
where $\Delta_n = | \psi(W;h_0,q_0) - \psi(W;\hat h^{(k)},\hat q^{(k)}) |$. Now, since the randomness of $(\hat h^{(k)}, \hat q^{(k)})$ is independent of that of the $k$'th fold (by cross fitting), we can apply \emph{e.g.} Höffding's inequality to the above for any given realization of $(\hat h^{(k)},\hat q^{(k)})$, to get a high-probability bound on $(\EE_{n,k} - \EE)[\Delta_n]$ over the independent sample of data given by $\EE_{n,k}$. That is, given any fixed $\epsilon > 0$, regardless of the realization of $(\hat h^{(k)},\hat q^{(k)})$ we can define some sequence $\delta_n \to 0$ such that $|(\EE_{n,k} - \EE)[\Delta_n]| \leq \epsilon$ with probability at least $1-\delta_n$ over the randomness of the $k$'th fold. That is, we have $(\EE_{n,k} - \EE)[\Delta_n] = o_p(1)$, so we can conclude that $\hat\sigma_n^2 \to \sigma_0^2$ in probability.

Moreover, the Slutsky's theorem implies that 
\begin{align*}
\frac{\sqrt{n}\prns{\hat\theta_n-\theta^\star}}{\hat \sigma_n^2} \rightsquigarrow \mathcal{N}\prns{0, 1}. 
\end{align*}
It follows that as $n\to\infty$, 
\begin{align*}
\Prb{\theta^\star \in \op{CI}} \to 1-\alpha. 
\end{align*}

\end{proof}

\section{Proofs for Partially Linear Model Estimation}
\begin{proof}[Proof for \Cref{prop: partial-iv-q2}]
According to \Cref{eq: partial-iv-q}, we have 
\begin{align*}
\Eb{{q_0(Z)}h(X)} = \Eb{\alpha(X)h(X)}, ~~ \forall h \in \Hcal.
\end{align*}
Since any partially linear $h\in\Hcal$ can be written as $h(X) = {X^\top_a\theta + g(X_b)}$, we have 
\begin{align*}
\Eb{q_0(Z)\prns{X_a^\top\theta + g(X_b)}} = \theta, ~~ \forall \theta\in\R{d_a}, g \in \Lcal_2(X_b). 
\end{align*}
This is apparently equivalent to \Cref{eq: partial-iv-q2}. 
\end{proof}

\begin{proof}[Proof for \Cref{prop: chen-nuisance}]
Obviously, $\xi_{0, i}$ is partially linear in $X$ so it belongs to the partially linear class $\Hcal$. According to \Cref{thm: new-nuisance}, $\xi_{0, i} \in \Hcal$ is a solution to \Cref{eq: PL-iv-delta} if and only if $q_{0, i} = P\xi_{0, i}$ satisfies \Cref{eq: partial-iv-q},  
so we can prove the desired conclusion once we prove that $q_{0, i} = P\xi_{0, i}$ satisfies \Cref{eq: partial-iv-q} for all $i = 1, \dots, d_a$. According to \Cref{prop: partial-iv-q2}, we only need to verify that 
$q_0 = (P\xi_{0, 1}, \dots, \xi_{0, d_a})$ satisfy \Cref{eq: partial-iv-q2}. 

Since $\rho_{0, i}$ solves the minimization problem in \Cref{eq: chen-nuisance-1}, it satisfies the first order condition that 
\begin{align*}
\Eb{\Eb{X^{(i)}_a - \rho_{0, i}(X_b) \mid Z}\Eb{\rho(X_b) \mid Z}} = \Eb{\Eb{X^{(i)}_a - \rho_{0, i}(X_b) \mid Z}{\rho(X_b)}} = 0, ~~ \forall \rho \in \Lcal_2(X_b). 
\end{align*}
This implies that $$\Eb{[P\xi_{0, i}](Z) \mid X_b}= \Eb{\Gamma^{-1}\Eb{X^{(i)}_a - \rho_{0, i}(X_b) \mid Z}\mid X_b} = {\Gamma^{-1}\Eb{\Eb{X^{(i)}_a - \rho_{0, i}(X_b) \mid Z} \mid X_b}} = 0.$$
It follows that $\Eb{q_0(Z) \mid X_b} = \mathbf{0}_{d_a}$. This further implies that 
\begin{align*}
 \Eb{q_0(Z)\rho_0(X_b)} = 0.
 \end{align*}
 Therefore, 
 \begin{align*}
  \Eb{q_0(Z)X_a^\top} &= \Eb{q_0(Z)\prns{X_a - \rho_0(X_b)}^\top} = \Eb{q_0(Z)X_a^\top} \\
  &= \Eb{q_0(Z)\prns{\Eb{X_a - \rho_0(X_b) \mid Z}}^\top} = \Gamma^{-1}\Eb{\prns{X_a-\rho_0(X_b) \mid Z}\prns{\Eb{X_a - \rho_0(X_b) \mid Z}}^\top} \\
  &= I_{d_a}.
  \end{align*}
This verifies that $q_0 = (P\xi_{0, 1}, \dots, \xi_{0, d_a})$ satisfies \Cref{eq: partial-iv-q2}, which concludes the proof.

The proof above shows the conclusion of \Cref{prop: chen-nuisance} indirectly through \Cref{prop: partial-iv-q2}. Below we provide an alternative proof that verifies the conclusion directly. 

According to the first order condition for \Cref{eq: PL-iv-delta},
function $\xi_{0, i}(X) = \theta_{0, i}^\top X_a + g_{0, i}(X_b)$ is the optimal solution to \Cref{eq: PL-iv-delta} if and only if
\begin{align*}
 &\Eb{\Eb{\theta_{0, i}^\top X_a + g_{0, i}(X_b) \mid Z}X_a^\top} = 
 \ell_i, \\
&\Eb{\Eb{\theta_{0, i}^\top X_a + g_{0, i}(X_b) \mid Z} \mid X_b} = 0,
 \end{align*} 
where $\ell_i \in \R{d_a}$ is a vector whose $i$th element is $1$ and all other elements are $0$.

If we let $\Theta_0$ be a $d_a \times d_a$ matrix whose $i$th row is $\theta_{0, i}^\top$, then we have 
\begin{align*}
&\Eb{\prns{\Theta_0 \Eb{X_a \mid Z} + \Eb{g_0(X_b) \mid Z}}X_a^\top} = I_{d_a \times d_a}, ~~ \\
&\Eb{\Eb{\Theta_{0}X_a + g_{0}(X_b) \mid Z} \mid X_b} = \mathbf{0}_{d_a}.
\end{align*}
Let's consider solutions in the column space of $\Theta_0$, namely, consider $g_0(X_b) = -\Theta_0g_0'(X_b)$. Then we have 
\begin{align*}
 \Theta_0\Eb{\Eb{X_a - g'_0(X_b) \mid Z}X_a^\top} = I_{d_a \times d_a}. 
 \end{align*} 
 This in turn implies that 
 \begin{align*}
  &\Theta_0\Eb{\Eb{X_a - g'_0(X_b) \mid Z}\prns{\Eb{X_a - g'_0(X_b) \mid Z}}^\top} \\
  &=\Theta_0 \Eb{\Eb{X_a - g'_0(X_b) \mid Z}\prns{{X_a - g'_0(X_b)}}^\top} \\
  &=  \Theta_0\Eb{\Eb{X_a - g'_0(X_b) \mid Z}X_a^\top} = I_{d_a \times d_a},
 \end{align*}
where the second equality  follows from the fact that $\Theta_0\Eb{\Eb{X_a - g'_0(X_b) \mid Z}\mid X_b} = \mathbf{0}_{d_a}$.

 This means that $\Theta_0$ and $\Eb{\Eb{X_a - g'_0(X_b) \mid Z}\prns{\Eb{X_a - g'_0(X_b) \mid Z}}^\top}$ have to be invertible, and 
 \begin{align*}
 \Theta_0 = \braces{\Eb{\Eb{X_a - g'_0(X_b) \mid Z}\prns{\Eb{X_a - g'_0(X_b) \mid Z}}^\top}}^{-1}. 
 \end{align*}
 Moreover, 
 \begin{align*}
 \mathbf{0}_{d_a} = \Eb{\Eb{\Theta_{0}X_a + g_{0}(X_b) \mid Z} \mid X_b} = \Theta_{0}\Eb{\Eb{X_a - g'_{0}(X_b) \mid Z} \mid X_b} \implies \Eb{\Eb{X_a - g'_{0}(X_b) \mid Z} \mid X_b} =  \mathbf{0}_{d_a}.
 \end{align*}
In other words, 
\begin{align*}
\Eb{\Eb{X^{(i)}_a - g'_{0, i}(X_b) \mid Z}\rho(X_b)} = 0, ~~ \forall \rho \in \Lcal_2(X_b).
\end{align*}
This is exactly the first order condition for the minimization problem in \Cref{eq: chen-nuisance-1}. 
 These show that the solutions to \Cref{eq: chen-nuisance-1} correspond to solutions to \Cref{eq: PL-iv-delta} with full rank coefficient matrix $\Theta_0$.  
\end{proof}

\begin{proof}[Proof of \Cref{prop: partial-linear-proximal}]

Let some arbitrary $h \in \Hcal_{\op{OB}}$ be given. In addition, define $\xi(v,x,a) = h(v,x,a) - \theta^{\star\top} a$. Now, for any given action $a$, we have
\begin{align*}
    \EE[\theta^{\star\top} A  + \xi(V,X,0) \mid U, X, A=a] &= \theta^{\star\top} a + \EE[\xi(V,X,0) \mid U, X, A=a] \\
    &= \theta^{\star\top} a + \EE[\xi(V,X,0) \mid U, X, A=0] \\
    &= \theta^{\star\top} a + \EE[\xi(V,X,A) \mid U, X, A=0] \\
    &= \theta^{\star^\top} a + \EE[h(V,X,A) - \theta^{\star\top} A \mid U, X, A=0] \\
    &= \theta^{\star\top} a + \EE[Y - \theta^{\star\top} A \mid U, X, A=0] \,.
\end{align*}
where in the second equality we apply the fact that $V \independent A \mid U,X$. Next, by assumption, we have $\EE[Y(a) \mid U,X] = \theta^{\star\top} a + \phi^\star(U,X)$, for each $a$. Furthermore, since $Y(a) \independent A \mid U,X$ for each $a$, we have
\begin{align*}
    \EE[Y \mid U,X,A=a] &= \EE[Y(a) \mid U,X,A=a] \\
    &= \EE[Y(a) \mid U,X] \\
    &= \theta^{\star\top} a + \phi^\star(U,X) \\
    &= \EE[\theta^{\star\top} A + \phi^\star(U,X) \mid U,X,A=a] \,.
\end{align*}
Therefore, we can further derive
\begin{align*}
    \theta^{\star\top} a + \EE[Y - \theta^{\star\top} A \mid U, X, A=0] &= \theta^{\star\top} a + \EE[\phi^\star(U,X) \mid U, X, A=0] \\
    &= \theta^{\star\top} a + \EE[\phi^\star(U,X) \mid U, X, A=a] \\
    &= \EE[\theta^{\star\top} A + \phi^\star(U,X) \mid U, X, A=a] \\
    &= \EE[Y \mid U, X, A=a] \,,
\end{align*}
where in the second equality we again apply the fact that $V \independent A \mid U,X$, and in the final equality we again apply the fact that $Y = \theta^{\star\top} A + \phi^\star(U,X)$.
Therefore, putting the above together, we have
\begin{equation*}
    \EE[Y - \theta^{\star\top} A  - \xi(V,X,0) \mid U, X, A=a] = 0 \quad \forall a \,,
\end{equation*}
so $\theta^{\star\top} A  + \xi(V,X,0) \in \tilde\Hcal_{\op{OB}}$, which proves the first part of the proposition.

For the second part of the proposition, let $q_0 = (q_{0,1},\ldots,q_{0,d_A})^\top$ be defined as in the proposiiton statement, and let $\tilde q(U,X,A) = \EE[q_0(Z,X,A) \mid U,X,A]$. Then, we have
\begin{align*}
    \EE[\tilde q(U,X,A) \mid V,X,A] &= \EE[\EE[q_0(Z,X,A) \mid U,X,A]  \mid V,X,A] \\
    &= \EE[\EE[q_0(Z,X,A) \mid U,V,X,A]  \mid V,X,A] \\
    &= \EE[q_0(Z,X,A) \mid V,X,A] \,,
\end{align*}
where the second equality follows by the proximal causal inference condition independence assumption $V \independent Z \mid U,X$.

Now, suppose that $h_1(V,A,X) = \theta_1^\top A + g_1(V,X)$ and $h_2(V,A,X) = \theta_2^\top A + g_2(V,X)$ are both bridge functions. We know that
\begin{equation*}
    \EE[(\theta_1 - \theta_2)^\top A + (g_1 - g_2)(V,A) \mid U,A,X] = 0 \,,
\end{equation*}
and therefore by the prior equation we have
\begin{align*}
    0 &= \EE\Big[\tilde q(U,X,A) \Big( (\theta_1 - \theta_2)(m(X))^\top \phi(A,X) + (g_1 - g_2)(V,X) \Big) \Big] \\
    &= \EE\Big[q_0(Z,X,A) \Big( (\theta_1 - \theta_2)(m(X))^\top \phi(A,X) + (g_1 - g_2)(V,X) \Big) \Big] \\
    &= \theta_1 - \theta_2 \,.
\end{align*}
That is, we have $\theta_1 = \theta_2$. Since this applies for any arbitrary partially linear bridge functions, and we know that $\theta^{\star\top} A  + \xi(V,X,0) \in \tilde\Hcal_{\op{OB}}$, it follows that the partially linear parameters $\theta^\star$ are unique.

\end{proof}

\section{Proofs for Results in Appendices}
\subsection{Proofs for \Cref{sec: connection}}
\begin{proof}[Proof for \Cref{prop: newey-interpret}]
The optimization problem in \Cref{eq: q-projection2} can be equivalently written as 
\begin{align*}
\min_{\xi\in\Hcal}~\Eb{\prns{q_0(T) - [P\xi](T)}^2}.
\end{align*}
A function $\xi_0$ solves the optimization problem above if and only if 
\begin{align*}
0 = \Eb{\prns{q_0(T) - [P\xi_0](T)}[P\xi](T)} = \Eb{\prns{q_0(T) - [P\xi_0](T)}g_1(W)\xi(S)}, ~~ \forall \xi \in \Hcal.
\end{align*}
According to the definition of $q_0$ in \Cref{eq: cond-moment-q-2}, we have $\Eb{g_1(W)q_0(T)\xi(S)} = \Eb{\alpha(S)\xi(S)}$ for any $\xi\in\Hcal$. So any solution $\xi_0$ can be equivalently characterized by 
\begin{align*}
\Eb{[P\xi_0](T)g_1(W)\xi(S)} = \Eb{\alpha(S)\xi(S)}.
\end{align*}
This exactly means that 
\begin{align*}
[P^\star\xi_0](S) = \alpha(S),
\end{align*}
namely, $\xi_0 \in \Xi_0$ defined in \Cref{assump: new-nuisance}.
\end{proof}

\begin{proof}[Proof for \Cref{prop: ai-chen}]
Since $\nu^\star$ solves \Cref{eq: Ai-chen-nuisance}, by the first order condition, we have 
\begin{align*}
\Eb{[P\nu^\star](T)[Ph](T)} + (1+ \Eb{m(W; \nu^\star)})\Eb{m(W; h)} = 0, ~~ \forall h \in \Hcal.
\end{align*}
It follows that 
\begin{align*}
V^\star 
 &= \Eb{[P\nu^\star](T)[P\nu^\star](T)} + \prns{1 + \Eb{m(W; \nu^\star)}}^2 \\
 &= -\prns{1 + \Eb{m(W; \nu^\star)}}\Eb{m(W; \nu^\star)} + \prns{1+\Eb{m(W; \nu^\star)}}^2 \\
 &= 1+\Eb{m(W; \nu^\star)},
\end{align*}
where the second equality uses the first order condition with $h = \nu^\star$. Dividing the left hand and right hand sides of the first order condition by $V^\star$, we have the following identity for $\xi^\star \coloneqq -\nu^\star/V^\star$:
\begin{align*}
\Eb{[P\xi^\star](T)[Ph](T)} = \Eb{m(W; h)}, ~~ \forall h \in \Hcal. 
\end{align*}
This is equvalent to 
\begin{align*}
\Eb{h(S)[P^\star P\xi^\star](S)} = \Eb{h(S)\alpha(S)}, ~~ \forall h \in \Hcal. 
\end{align*}
Therefore, we have  $\alpha = P^\star P\xi^\star$ for $\xi^\star \coloneqq -\nu^\star/V^\star$.  The expression for $\varphi(W; h^\star, \nu^\star)$ follows immediately from $V^\star = 1+\Eb{m(W; \nu^\star)}$ and the definition of $\xi^\star$. 

Conversely, if there exists $\xi^\star$ such that $\alpha = P^\star P\xi^\star \in \Rcal(P^\star P)$, then we have 
\begin{align*}
\Eb{h(S)[P^\star P\xi^\star](S)} = \Eb{h(S)\alpha(S)} = \Eb{m(W; h)}, ~~ \forall h \in \Hcal. 
\end{align*}
Based on such $\xi^\star$, we now construct a $v^\star$ that solves \Cref{eq: Ai-chen-nuisance}, that is, a $v^\star$ such that 
\begin{align*}
\Eb{[P\nu^\star](T)[Ph](T)} + (1+ \Eb{m(W; \nu^\star)})\Eb{m(W; h)} = 0, ~~ \forall h \in \Hcal.
\end{align*}
In particular, the above needs to hold for $\xi^\star \in \Hcal$, that is 
\begin{align*}
0 &= \Eb{[P\nu^\star](T)[P\xi^\star](T)} + (1+ \Eb{m(W; \nu^\star)})\Eb{m(W; \xi^\star)} \\
  &= \Eb{\nu^\star(S)[P^\star P\xi^\star](S)} + (1+ \Eb{m(W; \nu^\star)})\Eb{\xi^\star(S)[P^\star P\xi^\star](S)} \\
  &= \Eb{m(W; \nu^\star)}+ (1+ \Eb{m(W; \nu^\star)})\Eb{\prns{{\bracks{P\xi^\star}(T)}}^2}. 
\end{align*}
This is equivalent to 
\begin{align*}
1 = \prns{1 + \Eb{m(W; \nu^\star)}} + (1+ \Eb{m(W; \nu^\star)})\Eb{\prns{{\bracks{P\xi^\star}(T)}}^2}, 
\end{align*}
or $$1 + \Eb{m(W; \nu^\star)} = \frac{1}{1 + \Eb{\prns{{\bracks{P\xi^\star}(T)}}^2}}.$$
Therefore, we only need to find $v^\star$ such that the following identity holds for any $h \in \Hcal$:
\begin{align*}
0 &= \Eb{[P\nu^\star](T)[Ph](T)} + (1+ \Eb{m(W; \nu^\star)})\Eb{m(W; h)} \\
  &= \Eb{[P\nu^\star](T)[Ph](T)} + \frac{\Eb{m(W; h)}}{1 + \Eb{\prns{{\bracks{P\xi^\star}(T)}}^2}} \\
  &= \Eb{h(S)[P^\star P \nu^\star](S)} + \frac{\Eb{h(S)[P^\star P\xi^\star](S)}}{1 + \Eb{\prns{{\bracks{P\xi^\star}(T)}}^2}}.
\end{align*}
Obviously, $v^\star \coloneqq -\xi^\star/(1 + \E[\prns{\bracks{P\xi^\star}(T)}^2])$ satisfies the identity above so it solves \Cref{eq: Ai-chen-nuisance}. 
\end{proof}

\subsection{Proofs for \Cref{sec: convex}}

\begin{proof}[Proof for \Cref{theorem: convex}]
\Cref{eq: delta-eq-3} holds for $\xi_0\in\op{int}(\Hcal)$ if and only if the first order condition holds: for any  $\xi\in\Hcal$,
\begin{align*}
&\frac{\partial }{\partial t} \frac{1}{2}\Eb{\prns{[P(\xi_0 + t(\xi - \xi_0))](T)}^2} - \Eb{m(W; \xi_0 + t(\xi-\xi_0))} \vert_{t=0} \\
=&\Eb{[P\xi_0](T)[P(\xi-\xi_0)](T)} -\Eb{\alpha(S)(\xi-\xi_0)(S)}  = 0,
\end{align*}
or equivalently, 
\begin{align}\label{eq: foc-2}
\Eb{[P\xi_0](T)g_1(W)(\xi(S)-\xi_0(S))} = \Eb{\alpha(S)\prns{\xi(S)-\xi_0(S)}}.
\end{align}

We can take $\xi = h \in \Hcal$ and $\xi = \alpha \in \Hcal$ in \Cref{eq: foc-2}. This leads to 
\begin{align*}
\Eb{\prns{g_1(W)[P\xi_0](T) - \alpha(S)}\prns{h(S) - \xi_0(S)}} = 0 \\
\Eb{\prns{g_1(W)[P\xi_0](T) - \alpha(S)}\prns{\alpha(S) - \xi_0(S)}} = 0.
\end{align*}
Taking a difference of the tw
o equations above gives \begin{align}\label{eq: foc-1}
\Eb{\prns{g_1(W)[P\xi_0](T) - \alpha(S)}\prns{h(S) - \alpha(S)}} = 0, ~~ \forall h \in \Hcal.
\end{align}
This proves the conclusion in \Cref{eq: new-nuisance-convex}.  
\end{proof}

\begin{proof}[Proof for \Cref{lemma: bias-product-convex}]
We first prove \Cref{eq: DR-identification-convex}. Note that \Cref{eq: foc-2} implies for any $\xi\in\Hcal$,
\begin{align*}
\Eb{q_0(T)g_1(W)(\xi(S)-\xi_0(S))} = \Eb{\alpha(S)\prns{\xi(S)-\xi_0(S)}}, ~~ \text{where } q_0 = P\xi_0. 
\end{align*}
In particular, the above holds for $\xi = h\in\Hcal$ and $\xi = h_0\in\Hcal_0\subseteq \Hcal$ respectively: 
\begin{align*}
&\Eb{q_0(T)g_1(W)(h(S)-\xi_0(S))} = \Eb{\alpha(S)\prns{h(S)-\xi_0(S)}}, \\
&\Eb{q_0(T)g_1(W)(h_0(S)-\xi_0(S))} = \Eb{\alpha(S)\prns{h_0(S)-\xi_0(S)}}.
\end{align*}
Taking a difference of the two equations above gives 
\begin{align}\label{eq: H-convex-key-condition}
\Eb{q_0(T)g_1(W)(h(S)-h_0(S))} = \Eb{\alpha(S)\prns{h(S)-h_0(S)}}, ~~ \forall h\in\Hcal.
\end{align}

Then for any $h \in \Lcal_2\prns{S}, q \in \Lcal_2\prns{T}, h_0 \in \Hcal_0$ and  $q_0 = P\xi_0$ where $\xi_0$ is an interior solution to \Cref{eq: delta-eq-3}, we have 
\begin{align*}
\Eb{\psi(W; h, q)} - \theta^\star
    &= \Eb{\psi(W; h, q)} - \Eb{\psi(W; h_0, q_0)} \\
    &= \Eb{m(W; h- h_0)} + \Eb{q(T)(g_2(W) - g_1(W)h(S))} \\
    &= \Eb{\alpha(S)(h(S)- h_0(S))} + \Eb{q(T)(\Eb{g_2(W) \mid T} - g_1(W)h(S))} \\
    &= \Eb{g_1(W)q_0(T)(h(S)- h_0(S))} + \Eb{g_1(W)q(T)\prns{h_0(S) - h(S)}} \\
    &= \Eb{g_1(W)\prns{ q(T) - q_0(T)}\prns{h(S) -  h_0(S)}}.
\end{align*}
Here the fourth equality follows from \Cref{eq: H-convex-key-condition} and the definition of $\Hcal_0$. 
The rest of the proof can follow from the proof for \Cref{lemma: bias-product-general}.

WE next prove \Cref{eq: orthogonality-convex}. Note that 
\begin{align*}
 \frac{\partial}{\partial t}  \Eb{\psi(W; h_0 + t(h-h_0), q_0)}\big\vert_{t = 0} 
    &= \Eb{\alpha(S)(h(S)-h_0(S)} - \Eb{g_1(W)q_0(T)\prns{h(S)-h_0(S)}}.
 \end{align*} 
This is equal to $0$ according to  \Cref{eq: H-convex-key-condition}.

Moreover,  
\begin{align*}
 \frac{\partial}{\partial t}  \Eb{\psi(W; h_0, q_0  + t(q-q_0))}\big\vert_{t = 0}  = \Eb{\prns{q(T) - q_0(T)}\prns{g_2(W) - g_1(S)h_0(S)}}.
 \end{align*} 
This is equal to $0$ for any $q \in \Lcal_2(T)$ because $\Eb{g_2(W) - g_1(S)h_0(S) \mid T} = 0$.
\end{proof}

\end{document}